\documentclass[times,authoryear,nopreprintline]{elsarticle}

\usepackage[
  a4paper,
  left=22mm,
  right=22mm,
  top=22mm,
  bottom=24mm
]{geometry}

\usepackage{amssymb}
\usepackage{latexsym}
\usepackage{amsmath}
\usepackage{amsthm}
\usepackage{graphicx}
\usepackage{subcaption}
\usepackage{float}
\usepackage{booktabs}
\usepackage[hidelinks]{hyperref}
\usepackage{comment}
\newtheorem{theorem}{Theorem}
\newtheorem{lemma}{Lemma}
\newtheorem{corollary}{Corollary}

\usepackage{enumitem}
\graphicspath{{./figures/}}

\begin{document}

\begin{frontmatter}

\title{Estimated-State Adaptive Sliding Mode Control and Disturbance Observation Using Second-Order Surfaces for Spacecraft Formation Reconfiguration}

\author[1]{Jaein Lee}
\ead{jaein828@yonsei.ac.kr}
\author[1]{Hancheol Cho\corref{cor1}}
\cortext[cor1]{Corresponding author:
  Tel.: +82-2-2123-2685;
  Fax: +82-2-392-7680;}
\ead{hancho37@yonsei.ac.kr}
\author[2]{Tiago Roux Oliveira}
\ead{tiagoroux@uerj.br}

\affiliation[1]{organization={Yonsei University},
                addressline={50 Yonsei-Ro, Seodaemun-Gu},
                city={Seoul},
                postcode={03722},
                country={Republic of Korea}}

\affiliation[2]{organization={State University of Rio de Janeiro},
                addressline={Rua São Francisco Xavier 524 - sala 5018E - Bloco E},
                city={Rio de Janeiro},
                postcode={20550-900},
                country={Brazil}}


\begin{abstract}

This paper presents a two-phase relative orbit control framework for
spacecraft formation flying in which an analytic energy-optimal transfer is combined with a robust adaptive-gain sliding mode tracking law. In the first phase, the chaser is transferred from an arbitrary initial relative state to a projected circular orbit (PCO) in the Local-Vertical, Local-Horizontal frame by using the Clohessy--Wiltshire dynamics and the analytic minimum-energy formulation of the nominal trajectory. The phase at which the PCO is entered is not selected by numerical sweeping; instead, the transfer cost is parameterized by the PCO phase angle and the stationarity condition is reduced to a quartic polynomial whose real roots provide all candidate transition instants. In the second phase, the chaser maintains the PCO against external disturbances. Although a nominal PCO requires no sustained thrust, the presence of perturbations necessitates active tracking control. An adaptive sliding mode controller (ASMC) augmented by a sliding mode disturbance observer (SMDO) is therefore employed across both phases: during the reaching phase to compensate for disturbances acting on
the energy-optimal transfer, and during the PCO maintenance phase to reject perturbations that would otherwise cause the chaser to drift from the nominal orbit. The observer reduces the lumped disturbance to a bounded residual, while the tracking layer uses an adaptive gain updated from an estimated-state second-order sliding variable.
This sliding surface is introduced to tighten the ultimate tracking error bound, and a practical derivative estimation method is proposed that reuses reference velocity and acceleration signals already generated in the baseline reference surface computation, thereby avoiding the noise amplification of finite differencing and the tuning burden of an external differentiator. Simulation results demonstrate the effectiveness of the integrated architecture in achieving accurate tracking, disturbance rejection, and smooth phase transition.

\end{abstract}

\begin{keyword}
Spacecraft formation flying \sep
Projected circular orbit \sep
Estimated-state feedback \sep
Second-order-surface-based adaptive sliding mode control \sep
Second-order-surface-based sliding mode disturbance observer \sep
Energy-optimal reconfiguration
\end{keyword}

\end{frontmatter}
\section{Introduction}
Satellite formation flying is an enabling technology for future distributed space missions, including Earth observation, space-based interferometry, distributed space telescopes, inspection, rendezvous, and on-orbit servicing (\cite{schaub2003analytical, alfriend2009spacecraft, d2010autonomous, kang2020nanosat}). In such missions, a chaser spacecraft is often required to reconfigure its relative motion with respect to a target spacecraft and settle onto a prescribed bounded relative orbit. Among various relative orbital geometries, the projected circular orbit (PCO) is particularly attractive because it provides a periodic bounded solution under the Clohessy-Wiltshire (CW) dynamics and can therefore serve as a natural reference for formation keeping (\cite{clohessy1960terminal, vallado2001fundamentals, bonin2015canx}). 

Although a PCO is passively maintained in the nominal CW model once the chaser is placed exactly on it (\cite{schaub2001j2}), the problem of reaching the PCO is not unique. A PCO is parameterized by its phase angle, and different entry phases generally lead to different transfer trajectories and different control effort from the same initial relative state. Therefore, the PCO entry phase is not merely a geometric parameter, but a mission-level design variable that directly determines the energy efficiency of the reconfiguration maneuver. Previous studies have investigated spacecraft formation reconfiguration using analytical optimal control formulations, generating functions, differential orbital elements, impulsive maneuver planning, and low-thrust or operational-constraint-aware methods (\cite{cho2009analytic, wang2011new, scala2021design, pippia2022reconfiguration}). However, in many formulations, the terminal relative state or target relative orbit is prescribed in advance, and the phase at which the chaser enters the PCO is not explicitly optimized as part of the minimum-energy design. This motivates an analytic PCO insertion strategy in which the entry phase is selected directly from the stationarity condition of the transfer cost.

In realistic space environments, the chaser dynamics are not governed exactly by the ideal CW model. Environmental perturbations, higher-order gravitational effects, actuator imperfections, navigation errors, and unmodeled dynamics can drive the chaser away from both the nominal transfer trajectory and the desired PCO. A practical PCO-based reconfiguration strategy should therefore address two coupled tasks. First, it should exploit the analytic structure of the CW dynamics to determine an energy-efficient PCO insertion point. Second, it should maintain robust closed-loop tracking performance in the presence of disturbances, uncertainties, and imperfect state information. The first task concerns nominal mission efficiency, whereas the second concerns closed-loop reliability and implementability. 

Sliding mode control (SMC) has been extensively studied as a robust control framework for uncertain nonlinear systems because of its strong disturbance rejection capability and insensitivity to matched uncertainties (\cite{utkin2013sliding,edwards1998sliding}). Nevertheless, conventional SMC often requires conservative switching gains based on prior uncertainty bounds and may suffer from chattering in practical implementation. To address these issues, adaptive, output-feedback, monitoring-based, and higher-order sliding mode methods have been developed. \cite{cho2020adaptive} proposed a smooth adaptive SMC scheme for disturbances with unknown bounds, and \cite{hsu2019adaptive} developed a monitoring-function-based adaptive unit-vector control method for multivariable systems. \cite{rodrigues2022adaptive} introduced an adaptive SMC framework with guaranteed performance using monitoring and barrier functions, while \cite{guo2022performance} developed a performance-guaranteed adaptive asymptotic tracking method for nonlinear systems with unknown sign-switching control direction. In output-feedback settings, \cite{oliveira2013peaking} proposed a peaking-free exact tracking approach using dwell-time and norm observers, and \cite{oliveira2014overcoming} applied SMC with a switching monitoring scheme to uncertain uncalibrated visual servoing. High-order sliding mode (HOSM) tools have also been used for differentiator-based implementation, as in the global HOSM differentiator-based generalized model-reference adaptive control of \cite{oliveira2018generalized}. 

Sliding mode control and disturbance-observer-based control have also been investigated specifically for spacecraft formation flying and relative-motion control. \cite{cho2016satellite} proposed a smooth adaptive sliding mode controller for satellite formation control, reducing the dependence on a priori uncertainty bounds while maintaining robustness against external disturbances. \cite{lin2020adaptive} developed an adaptive tracking controller for spacecraft formation flying using a modified fast integral terminal sliding mode surface and a reduced-order disturbance observer to handle disturbances and parameter uncertainties. \cite{lee2018nonlinear} introduced a nonlinear disturbance-observer-based robust control method for spacecraft formation flying, showing how disturbance estimation can improve tracking performance under nonvanishing perturbations. \cite{wu2020disturbance} proposed a disturbance-observer-based fixed-time sliding mode control method for spacecraft proximity operations with coupled dynamics. More recently, \cite{jeon2025adaptive} applied adaptive smooth control with nonsingular fast terminal sliding modes to a CubeSat formation flying mission for distributed space telescope demonstration. These studies show the effectiveness of adaptive sliding mode and disturbance-observer-based control in formation-flying applications. Nevertheless, most existing approaches employ the adaptive controller and disturbance observer as separate design elements, and the observer-controller interconnection is often treated implicitly. In addition, although second-order and higher-order sliding mode methods have been widely used in spacecraft attitude control, their use in estimated-state observer-controller architectures for PCO-based relative orbit reconfiguration remains comparatively limited.

Several practical issues also remain in applying sliding mode and observer-based methods to PCO reconfiguration. First, many sliding mode control laws require a priori knowledge of disturbance upper bounds to select sufficiently large switching or adaptive gains. Overestimating these bounds can increase control effort and chattering, whereas underestimating them may degrade robustness. Second, the resulting tracking-error bound is not always explicitly related to mission-level accuracy requirements. Third, high-order sliding mode and observer-based methods often require derivatives of sliding variables. Direct finite differencing amplifies measurement noise, while external dynamic differentiators introduce additional tuning and may suffer from phase lag or noise sensitivity under navigation-error-contaminated measurements. These issues motivate a control architecture that combines assignable practical error bounds, adaptive disturbance rejection, estimated-state feedback, and an implementation method that avoids algebraic-loop and noisy-differentiation problems.

This paper proposes a two-phase relative orbit control framework that combines analytic minimum-energy PCO insertion with an auxiliary-state second-order sliding mode disturbance observer and adaptive sliding mode controller. In the reaching phase, the chaser is transferred from an arbitrary initial relative state to a prescribed PCO within a fixed transfer time. The PCO entry phase is treated as an optimization variable. By parameterizing the minimum-energy transfer cost with respect to the PCO phase angle, the stationarity condition is reduced to a quartic polynomial. The real roots of this polynomial provide all candidate PCO entry phases, and the minimum-energy insertion point is selected by evaluating the original transfer cost over this finite candidate set. After insertion, the reference switches to the analytical PCO trajectory, the nominal feedforward input is removed, and the same robust observer-controller compensation structure is retained for PCO maintenance.

The feedback architecture consists of two robust layers integrated through an auxiliary estimated state. The sliding mode disturbance observer (SMDO) uses a second-order observer surface to reduce the effective lumped disturbance to a bounded residual. The adaptive sliding mode controller (ASMC) uses a second-order estimated-state tracking surface to regulate the auxiliary-state tracking error. Both layers share an adaptive boundary layer structure, so the gains increase when the corresponding second-order surface lies outside its prescribed boundary layer and relax after the surface enters the layer. The auxiliary-state formulation induces a triangular observer-controller structure; the SMDO bounds the observer error, the ASMC bounds the estimated tracking error, and the actual tracking error is recovered algebraically from their sum. This provides a practical boundedness result without requiring a full composite Lyapunov proof for a coupled observer-controller system.

For implementation, this paper also introduces a surrogate derivative generator for evaluating the second-order sliding surfaces. Direct computation of the required surface derivatives can create an implicit algebraic loop because the acceleration-level quantities depend on the control commands being evaluated. Numerical differentiation is also undesirable because it amplifies navigation noise. To avoid these issues, a parallel first-order surrogate branch is used to generate loop-free acceleration-level surrogate signals, while velocity-level terms are evaluated from available measured or estimated velocities and analytic reference signals. The branch is used only for derivative generation; its control commands are not applied to the actual closed-loop plant.

The main contributions of this paper are summarized as follows.
\begin{itemize}
\item An analytic PCO entry phase selection method is developed for minimum-energy formation reconfiguration. By parameterizing the transfer cost with respect to the PCO phase angle, the stationarity condition is reduced to a quartic polynomial whose real roots provide all candidate entry phases. The minimum-energy PCO insertion point is then selected by evaluating the original cost over this finite candidate set, avoiding numerical phase sweeping.
\item An auxiliary-state second-order SMDO-ASMC architecture is proposed for robust two-phase PCO reconfiguration. The proposed structure integrates a second-order sliding mode disturbance observer and a second-order adaptive sliding mode tracking controller through an auxiliary estimated state. The same compensation architecture is retained in both the reaching and PCO-maintenance phases, while the nominal feedforward input is used only during the reaching phase and removed during the PCO phase.
\item A parameter-explicit triangular boundedness analysis is established for the observer-controller interconnection. Owing to the auxiliary-state formulation, the observer-error dynamics and the estimated tracking-error dynamics admit a triangular structure: the SMDO bounds the observer error, the ASMC bounds the estimated tracking error, and the actual tracking error is obtained algebraically from their sum. This analysis provides theoretical and practical upper bounds for the observer error, estimated tracking error, and actual tracking error, which can be adjusted in advance through the boundary layer widths and surface parameters.
\item A surrogate derivative generator is introduced for practical implementation of the second-order sliding surfaces. A parallel first-order surrogate branch generates loop-free acceleration-level signals for surface evaluation, while avoiding finite-difference noise amplification, external differentiator tuning, and algebraic-loop issues caused by acceleration-level quantities that depend on the control commands being evaluated.
\end{itemize}

The remainder of this paper is organized as follows. Section~\ref{sec:system_modeling_nominal_design} presents the relative motion model, the analytic minimum-energy transfer formulation, and the PCO entry phase selection method. Section~\ref{sec:adaptive_asmc_smdo} develops the auxiliary-state SMDO-ASMC control architecture. Section~\ref{sec:boundedness_stability} provides the boundedness analysis and derives the ultimate tracking-error bounds. Section~\ref{sec:implementation_summary} describes the surrogate derivative generator used to implement the second-order surfaces. Section~\ref{sec:numerical_results} presents numerical simulations, including baseline closed-loop performance and navigation-error cases. Section~\ref{sec:conclusion} concludes the paper.

\section{System Modeling and Nominal Trajectory Design}
\label{sec:system_modeling_nominal_design}

\subsection{Relative Motion Model}
\label{subsec:relative_motion_model}

The relative motion between the target spacecraft (target) and the chaser spacecraft (chaser) is described in the Local Vertical Local Horizontal (LVLH) frame attached to the target. The origin of this frame coincides with the target center of mass; the $x$-axis points radially outward from the Earth center through the target, the $z$-axis points along the orbital angular momentum vector, and the $y$-axis completes the right-handed triad. The chaser position and velocity relative to the target are denoted by
\begin{equation}
\boldsymbol{r}(t)=[x(t)\ y(t)\ z(t)]^{\top}, \qquad
\boldsymbol{v}(t)=\dot{\boldsymbol{r}}(t)
=[\dot{x}(t)\ \dot{y}(t)\ \dot{z}(t)]^{\top}.
\end{equation}

When the target follows a circular reference orbit with constant mean motion $n$, and the relative separation is small enough for linearization, the Hill equations reduce to the Clohessy--Wiltshire (CW)
model~\cite{clohessy1960terminal,schaub2003analytical}:
\begin{equation}
\ddot{x}-2n\dot{y}-3n^2x=u_x, \qquad
\ddot{y}+2n\dot{x}=u_y, \qquad
\ddot{z}+n^2z=u_z .
\label{eq:cw_component}
\end{equation}
Here $\boldsymbol{u}(t)=[u_x\ u_y\ u_z]^{\top}$ is the commanded specific acceleration. The in-plane dynamics ($x$--$y$) are coupled through the Coriolis terms $\pm 2n\dot{(\cdot)}$, while the out-of-plane component ($z$) is decoupled and behaves as a simple harmonic oscillator at frequency $n$.
Nominal control effectiveness is taken as unity and absorbed into
$\boldsymbol{u}(t)$; actuator imperfections and other effectiveness errors enter later as part of the lumped disturbance.

Defining the state vector
\begin{equation}
\boldsymbol{\xi}(t)=
\begin{bmatrix} \boldsymbol{r}(t) \\ \boldsymbol{v}(t) \end{bmatrix}
=[x\ y\ z\ \dot{x}\ \dot{y}\ \dot{z}]^{\top},
\end{equation}
the CW equations take the linear time-invariant state-space form
\begin{equation}
\dot{\boldsymbol{\xi}}(t)=A\boldsymbol{\xi}(t)+B\boldsymbol{u}(t),
\label{eq:ss_cw}
\end{equation}
with
\begin{equation}
A = \begin{bmatrix} 0_{3\times3} & I_{3\times3} \\ A_1 & A_2 \end{bmatrix},
\qquad
B = \begin{bmatrix} 0_{3\times3} \\ I_{3\times3} \end{bmatrix},
\end{equation}
\begin{equation}
A_1 = \begin{bmatrix} 3n^2 & 0 & 0 \\ 0 & 0 & 0 \\ 0 & 0 & -n^2 \end{bmatrix},
\qquad
A_2 = \begin{bmatrix} 0 & 2n & 0 \\ -2n & 0 & 0 \\ 0 & 0 & 0 \end{bmatrix}.
\label{eq:a_matrices}
\end{equation}
Since $A$ is constant, the homogeneous motion is generated by the matrix
exponential $\Phi(t)=\exp(At) \in \mathbb{R}^{6 \times 6}$, which can be partitioned as
\begin{equation}
\Phi(t)=\begin{bmatrix}\Phi_A(t)\\ \dot{\Phi}_A(t)\end{bmatrix},
\end{equation}
where $\Phi_A(t) \in \mathbb{R}^{3 \times 6}$ maps an initial state to the relative position and
$\dot{\Phi}_A(t) \in \mathbb{R}^{3 \times 6}$ maps it to the relative velocity:
\begin{equation}
\boldsymbol{r}(t)=\Phi_A(t)\boldsymbol{\xi}(0), \qquad
\boldsymbol{v}(t)=\dot{\Phi}_A(t)\boldsymbol{\xi}(0).
\end{equation}
The explicit form of $\Phi(t)$ is given in ~\ref{sec:appendix_nominal_trajectory}.

\subsection{Problem Statement and Two-Phase Mission Architecture}
\label{subsec:problem_two_phase_architecture}

In practice the chaser dynamics are not governed by the ideal CW model.
Environmental perturbations, higher-order gravitational terms, actuator imperfections, and other unmodeled effects all contribute to an unknown acceleration-level disturbance. Letting $\boldsymbol{q}(t)=\boldsymbol{r}(t)$ denote the measured relative position vector, the actual dynamics are
\begin{equation}
\ddot{\boldsymbol{q}}(t)
=
\boldsymbol{a}(t)
+
\boldsymbol{u}(t)
+
\boldsymbol{d}(t),
\label{eq:real_dyn}
\end{equation}
where $\boldsymbol{a}(t)$ denotes the nominal CW acceleration predicted by the model and $\boldsymbol{u}(t)$ is the control input designed to ensure tracking of the desired reference in the presence of $\boldsymbol{d}(t)$,
a lumped disturbance that accounts for all unmodeled effects.

The control objective is to achieve an energy-optimal transfer of the chaser from an arbitrary initial relative state $\boldsymbol{\xi}(0)=\boldsymbol{\xi}_0$
to a Projected Circular Orbit (PCO) of radius $\rho$ within a prescribed reaching time $t_f$, and subsequently to maintain that PCO despite $\boldsymbol{d}(t)$. A PCO is an elliptical relative trajectory in the LVLH
frame, characterized by a 2:1 aspect ratio in the $x$--$y$ plane and a circular projection of radius $\rho$ in the $y$--$z$ plane~\cite{sabol2001satellite}; in the absence of disturbances it is sustained passively by the CW dynamics, making it a natural formation-flying
reference.

The mission is organized into two phases. During the \textit{Reaching Phase}, $0\leq t\leq t_f$, an energy-optimal feedforward input $\boldsymbol{u}_r(t)$
is computed analytically under the disturbance-free CW model and drives the chaser toward a selected point on the PCO. During the \textit{PCO Phase}, $t>t_f$, the reference switches to the analytical PCO trajectory and $\boldsymbol{u}_r$ is removed, since the PCO is a homogeneous solution of the nominal CW dynamics that requires no continuous forcing in the absence of disturbances. This two-phase architecture decouples the nominal geometric problem---identifying the PCO entry point that minimizes transfer cost---from the robust tracking problem of following the reference when $\boldsymbol{d}(t)\neq\boldsymbol{0}$. Since $\boldsymbol{u}_r(t)$ alone is insufficient to compensate for the disturbance in Eq.~\eqref{eq:real_dyn} in both phases, a feedback layer is essential; its design is detailed in Section~\ref{sec:adaptive_asmc_smdo}.

The commanded input $\boldsymbol{u}(t)$ in Eq.~\eqref{eq:real_dyn} is decomposed as
\begin{equation}
\boldsymbol{u}(t)
=
\boldsymbol{u}_r(t)
+\hat{\boldsymbol{u}}_{\mathrm{asmc}}(t)
+\hat{\boldsymbol{u}}_{\mathrm{smdo}}(t),
\label{eq:total_control_input}
\end{equation}
where the hat notation $(\,\hat{\cdot}\,)$ indicates that the corresponding signal is evaluated using the estimated state $(\hat{\boldsymbol{q}}(t),\dot{\hat{\boldsymbol{q}}}(t))$ generated by an auxiliary observer system designed in Section~\ref{sec:adaptive_asmc_smdo}, rather than the directly measured state $(\boldsymbol{q}(t),\dot{\boldsymbol{q}}(t))$. The feedforward input $\boldsymbol{u}_r(t)$ is pre-computed from the nominal disturbance-free CW model and carries no hat, as it depends only on the analytically designed reference trajectory and not on the current state estimate. The ASMC tracking command $\hat{\boldsymbol{u}}_{\mathrm{asmc}}$ and the SMDO compensation command $\hat{\boldsymbol{u}}_{\mathrm{smdo}}$ are both evaluated using $(\hat{\boldsymbol{q}}(t),\dot{\hat{\boldsymbol{q}}}(t))$.
The phase-dependent feedforward is
\begin{equation}
\boldsymbol{u}_r(t)
=
\begin{cases}
\boldsymbol{u}_r(t), & \text{Reaching Phase}\ (0\leq t\leq t_f),\\[2pt]
\boldsymbol{0},      & \text{PCO Phase}\ (t>t_f).
\end{cases}
\label{eq:phase_dependent_feedforward}
\end{equation}

Throughout this paper, two norms are used for different purposes.
The Euclidean norm $\|\boldsymbol{x}\|
:=\bigl(\sum_i x_i^2\bigr)^{1/2}$ is used for vector-valued sliding
variables and Lyapunov function arguments, as it is the natural norm for
the boundary-layer analysis in the stability proofs. The infinity norm
$\|\boldsymbol{x}\|_\infty:=\max_i|x_i|$ is used when componentwise
ultimate bounds are obtained from a scalar Lyapunov argument applied to
each component independently. The two norms are equivalent in finite dimensions and
satisfy $\|\boldsymbol{x}\|_\infty\leq\|\boldsymbol{x}\|\leq
\sqrt{n}\|\boldsymbol{x}\|_\infty$ for $\boldsymbol{x}\in\mathbb{R}^n$.

The resulting control allocation in each phase is summarized in
Table~\ref{tab:phase_control_allocation}, and the corresponding mission
sequence is illustrated in Fig.~\ref{fig:mission_phase_distribution}.

\begin{table}[h]
\centering
\caption{Control allocation in the two mission phases.}
\label{tab:phase_control_allocation}
\begin{tabular}{lll}
\toprule
Mission phase & Time interval & Applied control input \\
\midrule
Reaching Phase & $0\leq t\leq t_f$ &
  $\boldsymbol{u}_r
  +\hat{\boldsymbol{u}}_{\mathrm{asmc}}
  +\hat{\boldsymbol{u}}_{\mathrm{smdo}}$ \\[2pt]
PCO Phase      & $t>t_f$           &
  $\hat{\boldsymbol{u}}_{\mathrm{asmc}}
  +\hat{\boldsymbol{u}}_{\mathrm{smdo}}$ \\
\bottomrule
\end{tabular}
\end{table}

\begin{figure}[ht]
\centering
\includegraphics[width=0.5\textwidth]{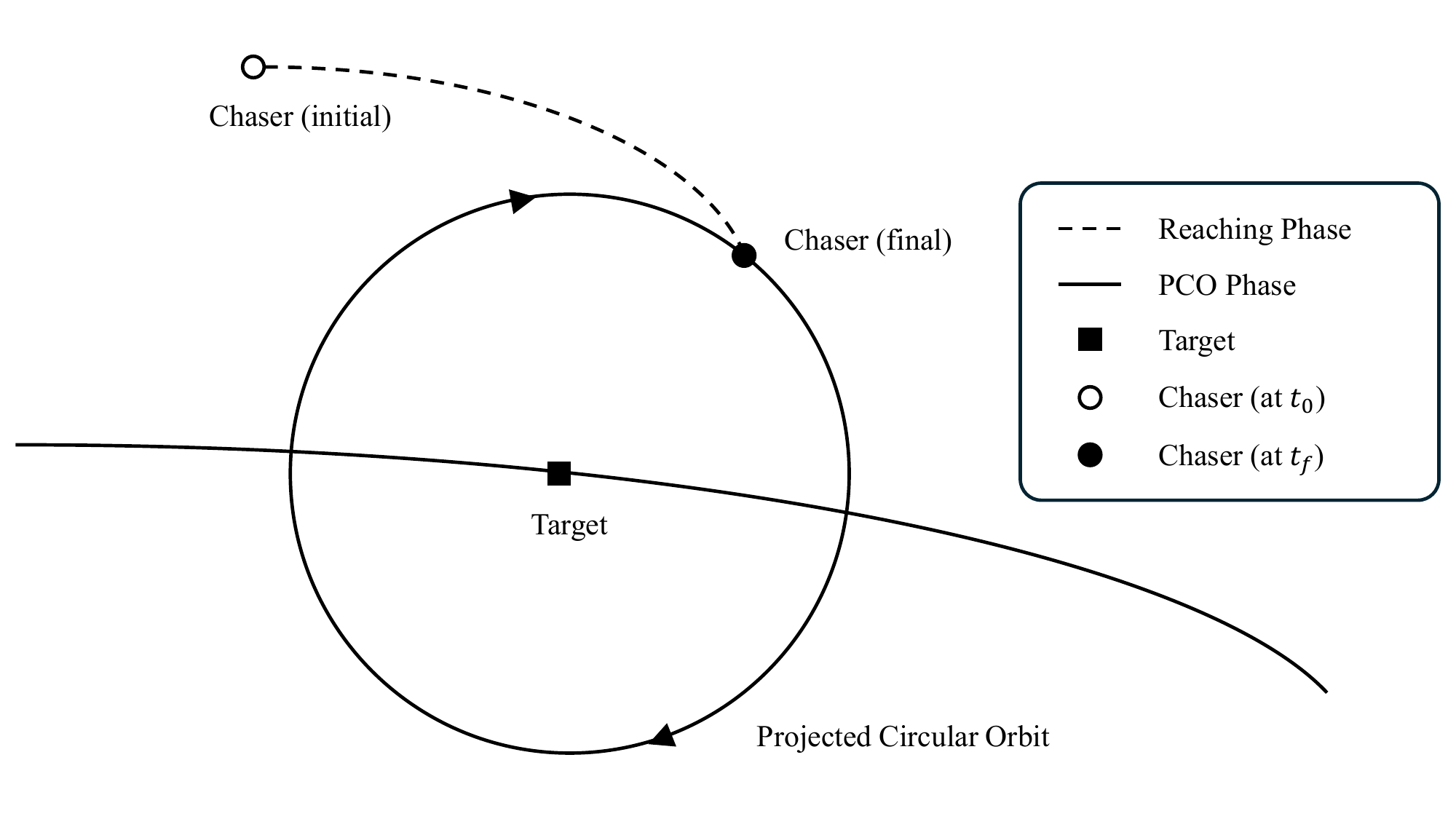}
\caption{Two-phase mission architecture. The chaser transfers from its
initial state at $t_0$ to the selected PCO insertion point at $t_f$
(Reaching Phase) and subsequently maintains the PCO around the target
spacecraft (PCO Phase).}
\label{fig:mission_phase_distribution}
\end{figure}

\subsection{Nominal Reaching Trajectory Design}
\label{subsec:nominal_reaching_design}

For the nominal reaching problem, the disturbance term in Eq.~\eqref{eq:real_dyn}
is set to zero, and the transfer effort is measured by
\begin{equation}
J = \frac{1}{2}\int_{0}^{t_f}
\boldsymbol{u}_r(\tau)^{\top}\boldsymbol{u}_r(\tau)\,d\tau .
\label{eq:cost_J}
\end{equation}
For a prescribed terminal state $\boldsymbol{\xi}_f$ and transfer time $t_f$,
the analytic energy-optimal solution of (\cite{cho2009analytic}) is constructed from the upper block $\Phi_A(t)$ of the CW state transition matrix. The key quantities are
\begin{align}
S(t) &= \int_{0}^{t} \Phi_A(\tau)^{\top} \Phi_A(\tau) \, d\tau,
\label{eq:S_def} \\
G &= \Phi_f^{-1}\boldsymbol{\xi}_f - \Phi_0^{-1}\boldsymbol{\xi}_0,
\label{eq:G_def} \\
C_{\Phi} &= \Phi_A^{\top}\dot{\Phi}_A
- \left(\Phi_A^{\top}\dot{\Phi}_A\right)^{\top}
- \Phi_A^{\top}A_2\Phi_A,
\label{eq:Cphi_def}
\end{align}
where $\Phi_0=\Phi(0)$ and $\Phi_f=\Phi(t_f)$. Here $S(t)$ is a positive definite Gramian-like matrix that accumulates the reachable input directions up to time $t$; $G$ encodes the boundary conditions by comparing the transformed terminal and initial states; and $C_{\Phi}$, which is constant and skew-symmetric for the CW model, mediates the coupling between position and velocity in the optimality conditions. The nominal reaching input minimizing the transfer cost Eq.~\eqref{eq:cost_J}, the associated transfer cost, and the reference state trajectory are then (\cite{cho2009analytic})
\begin{align}
\boldsymbol{u}_r(t) &= \Phi_A(t) S_f^{-1} C_{\Phi} G,
\label{eq:ur} \\
J &= \frac{1}{2} G^{\top} C_{\Phi}^{\top} S_f^{-1} C_{\Phi} G,
\label{eq:J_opt} \\
\boldsymbol{\xi}_n(t) &= \Phi(t)
\left(\Phi_0^{-1}\boldsymbol{\xi}_0
+ C_{\Phi}^{-1} S(t) S_f^{-1} C_{\Phi} G\right),
\label{eq:xi_n}
\end{align}
where $S_f=S(t_f)$. The position component of $\boldsymbol{\xi}_n(t)$, denoted $\boldsymbol{q}_n(t)$, serves as the reference tracked by the feedback controller during the Reaching Phase. A derivation of Eqs.~\eqref{eq:ur}--\eqref{eq:xi_n} from the optimality conditions is given in
~\ref{sec:appendix_nominal_trajectory}.

\subsection{Analytical Selection of the PCO Entry Phase}
\label{subsec:pco_entry_phase_selection}

The terminal point of the Reaching Phase on the PCO is not fixed a priori but is parameterized by the entry phase angle $\phi_0$. A point on the PCO has the following state:
\begin{equation}
\boldsymbol{\xi}_f(\phi_0) =
\begin{bmatrix}
\dfrac{\rho}{2}\sin(n t_f + \phi_0) \\[6pt]
\rho \cos(n t_f + \phi_0) \\[4pt]
\rho \sin(n t_f + \phi_0) \\[4pt]
\dfrac{n\rho}{2}\cos(n t_f + \phi_0) \\[6pt]
-n\rho \sin(n t_f + \phi_0) \\[4pt]
 n\rho \cos(n t_f + \phi_0)
\end{bmatrix},
\label{eq:xi_f_phi}
\end{equation}
so that substituting Eq.~\eqref{eq:xi_f_phi} into $G$ makes the transfer cost in Eq.~\eqref{eq:J_opt} a scalar $2\pi$-periodic function $J(\phi_0)$. Choosing $\phi_0$ to minimize $J(\phi_0)$ therefore determines the energy-optimal PCO entry point without altering the transfer time $t_f$.

Because $J(\phi_0)$ is smooth and periodic, its global minimum over
$[0,2\pi)$ is attained at a stationary point satisfying $dJ/d\phi_0=0$.
To solve this condition in closed form, the positive-definite matrix
$C_{\Phi}^{\top}S_f^{-1}C_{\Phi}$ is factored via Cholesky decomposition as
\begin{equation}
C_{\Phi}^{\top} S_f^{-1} C_{\Phi} = R^{\top} R,
\label{eq:cholesky}
\end{equation}
so that $J(\phi_0)=\frac{1}{2}\|RG(\phi_0)\|^2$ where $R$ is an upper triangular matrix. Each component of $RG$ is then an affine sinusoidal function of the form
\begin{equation}
f_k(\phi_0)=a_k\sin\phi_0+b_k\cos\phi_0+c_k,
\qquad k=1,\ldots,6.
\label{eq:fk}
\end{equation}
Differentiating $J=\frac{1}{2}\sum_k f_k^2$ and setting the result to zero yields
\begin{equation}
\frac{dJ}{d\phi_0}
= \sum_{k=1}^{6} f_k(\phi_0)\,f'_k(\phi_0) = 0, \qquad
f'_k = a_k\cos\phi_0 - b_k\sin\phi_0,
\label{eq:dJ_dphi0}
\end{equation}
where $f'_k$ denotes differentiation of $f_k$ with respect to $\phi_0$. After expanding and collecting double-angle and single-angle terms, Eq.~\eqref{eq:dJ_dphi0} becomes
\begin{equation}
A_q \sin 2\phi_0 + 2B_q \cos 2\phi_0
+ 2C_q \cos \phi_0 - 2D_q \sin \phi_0 = 0,
\label{eq:stationarity}
\end{equation}
with
\begin{equation}
A_q = \sum_{k=1}^{6}(a_k^2-b_k^2), \quad
B_q = \sum_{k=1}^{6} a_k b_k, \quad
C_q = \sum_{k=1}^{6} a_k c_k, \quad
D_q = \sum_{k=1}^{6} b_k c_k .
\label{eq:ABCD_q}
\end{equation}
Applying the Weierstrass substitution $\tau=\tan(\phi_0/2)$ converts
Eq.~\eqref{eq:stationarity} into the degree-four polynomial
\begin{equation}
(B_q-C_q)\tau^4
+(-2A_q-2D_q)\tau^3
-6B_q\tau^2
+(2A_q-2D_q)\tau
+(B_q+C_q)=0,
\label{eq:phase_quartic_polynomial}
\end{equation}
whose finite real roots are mapped back to phase candidates by
$\phi_{0,i}=\operatorname{mod}(2\arctan\tau_i,\,2\pi), i=1,2,3,4$. Since $\tau=\infty$
corresponds to $\phi_0=\pi$, the point $\phi_0=\pi$ is a stationary
candidate if and only if the leading coefficient vanishes, i.e.\
$B_q-C_q=0$. If all coefficients vanish simultaneously, the cost is
constant and any entry phase is equally energy-optimal. The handling of these
special cases is summarized in~\ref{sec:appendix_nominal_trajectory},
Table~\ref{tab:phi0_case_summary}.
Because the quartic roots include both minima and maxima, the selected entry
phase is finally determined by directly evaluating the original cost over the
admissible candidate set $\mathcal{C}$:
\begin{equation}
\phi_0^{\star}=\arg\min_{\phi_0\in\mathcal{C}}J(\phi_0).
\label{eq:phi0_star}
\end{equation}

\subsection{Projected Circular Orbit Reference}
\label{subsec:pco_reference}

With the optimal $\phi_0^{\star}$ determined by Eq.~\eqref{eq:phi0_star}, the desired PCO trajectory in the LVLH frame is
\begin{equation}
\boldsymbol{q}_d(t)=
\begin{bmatrix}
\dfrac{\rho}{2}\sin(nt+\phi_0^{\star}) \\[4pt]
\rho\cos(nt+\phi_0^{\star}) \\[4pt]
\rho\sin(nt+\phi_0^{\star})
\end{bmatrix},
\label{eq:q_d}
\end{equation}
with velocity and acceleration obtained by differentiation. In the
undisturbed CW model a chaser placed exactly on Eq.~\eqref{eq:q_d} remains on
it without additional control effort, which is precisely why the PCO is
chosen as the formation-keeping reference.

The phase-dependent nominal reference commanded to the feedback controller
is
\begin{equation}
\boldsymbol{q}_n(t)=
\begin{cases}
\begin{bmatrix}I_{3\times3} & 0_{3\times3}\end{bmatrix}
\boldsymbol{\xi}_n(t), & 0\leq t\leq t_f, \\[4pt]
\boldsymbol{q}_d(t), & t>t_f .
\end{cases}
\label{eq:qn_ref}
\end{equation}
During the Reaching Phase, $\boldsymbol{q}_n(t)$ is the position component
of the energy-optimal nominal trajectory $\boldsymbol{\xi}_n(t)$ from
Eq.~\eqref{eq:xi_n}. During the PCO Phase, $\boldsymbol{q}_n(t)$ switches
continuously to the analytical PCO trajectory Eq.~\eqref{eq:q_d}.

The nominal design provides $\boldsymbol{q}_n(t)$ and $\boldsymbol{u}_r(t)$
under the idealized disturbance-free assumption. When $\boldsymbol{d}(t)$
is present, however, the chaser will deviate from this reference and
$\boldsymbol{u}_r(t)$ alone cannot restore tracking. The following section
therefore introduces a feedback architecture that augments the feedforward
input with a Adaptive Sliding Mode Controller (ASMC) for
robust tracking and a Sliding Mode Disturbance Observer (SMDO) for
online estimation and compensation of $\boldsymbol{d}(t)$.

\section{SMDO--ASMC Control Design Based on Second-Order Surfaces Using the Estimated State}
\label{sec:adaptive_asmc_smdo}

The robust feedback law developed in this section augments the nominal
feedforward $\boldsymbol{u}_r(t)$ of Section~\ref{sec:system_modeling_nominal_design}
with two complementary layers: a Sliding Mode Disturbance Observer
(SMDO) that estimates and compensates for the lumped disturbance, and an Adaptive Sliding Mode Controller (ASMC) that tracks the nominal reference using the
estimated state. Both layers share the same adaptive boundary-layer architecture based on second-order surfaces, which is analyzed in Section~\ref{sec:boundedness_stability}.

\subsection{SMDO for Disturbance Compensation via Second-Order Sliding Surfaces}
\label{subsec:second_order_smdo}

The actual acceleration-level dynamics are
\begin{equation}
\ddot{\boldsymbol{q}}
=
\boldsymbol{a}
+
\boldsymbol{u}_r
+
\hat{\boldsymbol{u}}_{\mathrm{asmc}}
+
\hat{\boldsymbol{u}}_{\mathrm{smdo}}
+
\boldsymbol{d},
\label{eq:actual_normalized_dynamics}
\end{equation}
where $\boldsymbol{a}$ is the nominal CW acceleration evaluated along the actual
state, $\boldsymbol{d}$ is the lumped disturbance, and hat notation indicates that
the corresponding control signal is computed using the estimated state
$\hat{\boldsymbol{q}}$ rather than the measured state $\boldsymbol{q}$. The estimated state is generated by the auxiliary system
\begin{equation}
\ddot{\hat{\boldsymbol{q}}}
=
\hat{\boldsymbol{a}}
+
\boldsymbol{u}_r
+
\hat{\boldsymbol{u}}_{\mathrm{asmc}},
\label{eq:auxiliary_normalized_dynamics}
\end{equation}
where $\hat{\boldsymbol{a}}$ is the nominal acceleration evaluated along
$\hat{\boldsymbol{q}}$. The observer error and model-acceleration mismatch are defined as
\begin{equation}
\tilde{\boldsymbol{q}}=\boldsymbol{q}-\hat{\boldsymbol{q}},
\qquad
\tilde{\boldsymbol{a}}=\boldsymbol{a}-\hat{\boldsymbol{a}}.
\label{eq:model_acceleration_error}
\end{equation}
Subtracting Eq.~\eqref{eq:auxiliary_normalized_dynamics} from
Eq.~\eqref{eq:actual_normalized_dynamics}, we obtain
\begin{equation}
\ddot{\tilde{\boldsymbol{q}}}
=
\tilde{\boldsymbol{a}}
+
\boldsymbol{d}
+
\hat{\boldsymbol{u}}_{\mathrm{smdo}}.
\label{eq:observer_error_dynamics}
\end{equation}

In this paper the SMDO is built on the baseline observer surface
\begin{equation}
\hat{\boldsymbol{s}}_1
=
\dot{\tilde{\boldsymbol{q}}}
+
\lambda_1\tilde{\boldsymbol{q}},
\qquad \lambda_1>0,
\label{eq:smdo_first_surface}
\end{equation}
and the second-order observer sliding surface
\begin{equation}
\hat{\boldsymbol{s}}_2
=
\dot{\hat{\boldsymbol{s}}}_1
+
\lambda_2
\operatorname{sgn}\!\left(\hat{\boldsymbol{s}}_1\right)
\left|\hat{\boldsymbol{s}}_1\right|^{\nu_{\mathrm{smdo}}},
\qquad
\lambda_2>0,
\quad
0<\nu_{\mathrm{smdo}}<1,
\label{eq:smdo_second_surface}
\end{equation}
where $\operatorname{sgn}(\cdot)$ and $|\cdot|^{\nu_{\mathrm{smdo}}}$ are applied componentwise. The selection of surfaces, such as those in Eqs.~\eqref{eq:smdo_first_surface} and~\eqref{eq:smdo_second_surface}, is a design choice that significantly affects observer/control performance, including the error convergence rate and the upper bound on the error magnitude.

To identify the cancellation component, differentiate
Eq.~\eqref{eq:smdo_first_surface} and use
Eq.~\eqref{eq:observer_error_dynamics} to obtain
\begin{equation}
\dot{\hat{\boldsymbol{s}}}_1
=
\ddot{\tilde{\boldsymbol{q}}}
+
\lambda_1\dot{\tilde{\boldsymbol{q}}}
=
\tilde{\boldsymbol{a}}
+
\lambda_1\dot{\tilde{\boldsymbol{q}}}
+
\boldsymbol{d}
+
\hat{\boldsymbol{u}}_{\mathrm{smdo}}.
\label{eq:smdo_sdot_expanded}
\end{equation}
Substituting Eq.~\eqref{eq:smdo_sdot_expanded} into
Eq.~\eqref{eq:smdo_second_surface} gives
\begin{align}
\hat{\boldsymbol{s}}_2
={}&
\tilde{\boldsymbol{a}}
+
\lambda_1\dot{\tilde{\boldsymbol{q}}}
+
\lambda_2
\operatorname{sgn}(\hat{\boldsymbol{s}}_1)
|\hat{\boldsymbol{s}}_1|^{\nu_{\mathrm{smdo}}}
+
\boldsymbol{d}
+
\hat{\boldsymbol{u}}_{\mathrm{smdo}}.
\label{eq:smdo_s2_expanded}
\end{align}

The first three terms in Eq.~\eqref{eq:smdo_s2_expanded} are computable
from the observer states and known model quantities. The SMDO compensation
is therefore decomposed as
\begin{equation}
\hat{\boldsymbol{u}}_{\mathrm{smdo}}
=
\hat{\boldsymbol{u}}_{\mathrm{smdo}}^{\mathrm{can}}
+
\hat{\boldsymbol{u}}_{\mathrm{smdo}}^{\mathrm{ad}},
\label{eq:smdo_input_decomposition}
\end{equation}
where the cancellation component removes the nominally computable 
model- and surface-dependent terms, whereas the adaptive component 
dynamically regulates $\hat{\boldsymbol{s}}_2$ against residual 
uncertainties. The cancellation component is selected as
\begin{equation}
\hat{\boldsymbol{u}}_{\mathrm{smdo}}^{\mathrm{can}}
=
-\tilde{\boldsymbol{a}}
-
\lambda_1\dot{\tilde{\boldsymbol{q}}}
-
\lambda_2
\operatorname{sgn}(\hat{\boldsymbol{s}}_1)
|\hat{\boldsymbol{s}}_1|^{\nu_{\mathrm{smdo}}}.
\label{eq:smdo_cancellation_component}
\end{equation}
Substituting Eq.~\eqref{eq:smdo_cancellation_component} into
Eq.~\eqref{eq:smdo_s2_expanded} eliminates the nominally computable
terms. In practice, however, the cancellation relies on estimated
states and nominal model information; modeling errors and estimation
imperfections introduce an implementation-level bounded residual $\boldsymbol{\Delta}_{can}$,
so that
\begin{equation}
    \hat{\boldsymbol{s}}_2
    =
    \boldsymbol{d} + \boldsymbol{\Delta}_{can}
    +
    \hat{\boldsymbol{u}}_{\mathrm{smdo}}^{\mathrm{ad}}
    =:
    \boldsymbol{d}_{\mathrm{eff}}
    +
    \hat{\boldsymbol{u}}_{\mathrm{smdo}}^{\mathrm{ad}},
    \label{eq:smdo_s2_residual_relation}
\end{equation}
where
\begin{equation}
    \boldsymbol{d}_{\mathrm{eff}} := \boldsymbol{d} + \boldsymbol{\Delta}_{can}
    \label{eq:deff_definition}
\end{equation}
is the effective lumped disturbance. The residual
\begin{equation}
    \tilde{\boldsymbol{d}}_{\mathrm{smdo}}
    :=
    \boldsymbol{d}_{\mathrm{eff}}
    +
    \hat{\boldsymbol{u}}_{\mathrm{smdo}}^{\mathrm{ad}}
    \label{eq:smdo_residual_definition}
\end{equation}
is ultimately bounded, as established in 
Theorem~\ref{thm:smdo_residual_boundedness}.

The adaptive-component rate is selected as
\begin{equation}
\dot{\hat{\boldsymbol{u}}}_{\mathrm{smdo}}^{\mathrm{ad}}
=
-\frac{L(t)+L^*}{\epsilon}
\hat{\boldsymbol{s}}_2,
\qquad
L^*>0,
\quad
\epsilon>0,
\label{eq:smdo_adaptive_injection}
\end{equation}
and the adaptive gain is updated according to
\begin{equation}
\dot{L}(t)
=
\eta
\left[
L(t)
\left(
\frac{\|\hat{\boldsymbol{s}}_2\|}{\epsilon}
-1
\right)
+
\frac{L^*}{\epsilon}
\|\hat{\boldsymbol{s}}_2\|
\right],
\qquad
\eta>0,
\quad
L(0)>0.
\label{eq:smdo_adaptive_gain_update}
\end{equation}
The gain-update rule follows the adaptive boundary-layer mechanism of~\cite{cho2020autonomous}, while the adaptive input is
implemented here in an integral form through
$\dot{\hat{\boldsymbol{u}}}_{\mathrm{smdo}}^{\mathrm{ad}}$ so that it acts
on $\hat{\boldsymbol{s}}_2$. Thus, the
adaptive component is generated dynamically by integrating the second-order surface $\hat{\boldsymbol{s}}_2$,
rather than being prescribed as an algebraic
function of $\hat{\boldsymbol{s}}_1$. The gain remains positive under the
conditions of Lemma~\ref{lem:generic_gain_upper_bound} in
Section~\ref{sec:boundedness_stability}. It increases while
$\hat{\boldsymbol{s}}_2$ lies outside the $\epsilon$-boundary layer and
relaxes after $\hat{\boldsymbol{s}}_2$ enters the prescribed layer.

Integrating Eq.~\eqref{eq:smdo_adaptive_injection} and applying the decomposition
$\hat{\boldsymbol{u}}_{\mathrm{smdo}}
=\hat{\boldsymbol{u}}_{\mathrm{smdo}}^{\mathrm{can}}+\hat{\boldsymbol{u}}_{\mathrm{smdo}}^{\mathrm{ad}}$, we have
\begin{equation}
\hat{\boldsymbol{u}}_{\mathrm{smdo}}(t)
=
\hat{\boldsymbol{u}}_{\mathrm{smdo}}^{\mathrm{can}}(t)
+
\hat{\boldsymbol{u}}_{\mathrm{smdo}}^{\mathrm{ad}}(0)
-
\int_0^t
\frac{L(\tau)+L^*}{\epsilon}
\hat{\boldsymbol{s}}_2(\tau)\,d\tau.
\label{eq:smdo_total}
\end{equation}
The residual disturbance
$\tilde{\boldsymbol{d}}_{\mathrm{smdo}}:=\boldsymbol{d}_{\mathrm{eff}}+\hat{\boldsymbol{u}}_{\mathrm{smdo}}^{\mathrm{ad}}$
that remains after observer compensation is ultimately bounded as will be shown in Section~\ref{sec:boundedness_stability}.

\subsection{ASMC-Based Tracking Control via Second-Order Sliding Surfaces Using the Estimated State}
\label{subsec:second_order_asmc}

The ASMC layer tracks the phase-dependent nominal reference $\boldsymbol{q}_n(t)$
of Eq.~\eqref{eq:qn_ref} using the auxiliary or estimated state $\hat{\boldsymbol{q}}(t)$ in
place of direct measurement $\boldsymbol{q}(t)$. The actual, estimated, and observer tracking errors are defined as
\begin{equation}
\boldsymbol{e}(t)=\boldsymbol{q}(t)-\boldsymbol{q}_n(t),
\qquad
\hat{\boldsymbol{e}}(t)=\hat{\boldsymbol{q}}(t)-\boldsymbol{q}_n(t),
\qquad
\tilde{\boldsymbol{q}}(t)=\boldsymbol{q}(t)-\hat{\boldsymbol{q}}(t).
\label{eq:tracking_error_definitions}
\end{equation}
They satisfy $\boldsymbol{e}=\tilde{\boldsymbol{q}}+\hat{\boldsymbol{e}}$, and so the following inequality is satisfied:
\begin{equation}
\|\boldsymbol{e}(t)\|
\leq
\|\tilde{\boldsymbol{q}}(t)\|+
\|\hat{\boldsymbol{e}}(t)\|.
\label{eq:tracking_error_decomposition}
\end{equation}
This decomposition reflects the two-layer structure of the proposed
architecture: $\tilde{\boldsymbol{q}}$ is the observer error driven by the
SMDO, and $\hat{\boldsymbol{e}}$ is the estimated tracking error driven by
the ASMC. If each layer keeps its respective error bounded, the actual
tracking error $\boldsymbol{e}$ is bounded as well. The SMDO analysis in
Theorem~\ref{thm:smdo_residual_boundedness} of Section~\ref{sec:boundedness_stability} establishes ultimate bounds on
$\tilde{\boldsymbol{q}}$ and $\dot{\tilde{\boldsymbol{q}}}$; the ASMC analysis in Theorem~\ref{thm:asmc_tracking_boundedness} of Section~\ref{sec:boundedness_stability} then uses those bounds to establish ultimate boundedness of $\hat{\boldsymbol{e}}$ and $\dot{\hat{\boldsymbol{e}}}$, from which a bound on $\boldsymbol{e}$ follows directly via Eq.~\eqref{eq:tracking_error_decomposition}.
The ASMC is built on the baseline estimated-state surface
\begin{equation}
\hat{\boldsymbol{\sigma}}_1(t)
=
\dot{\hat{\boldsymbol{e}}}(t)
+
C_1\hat{\boldsymbol{e}}(t),
\qquad
C_1>0,
\label{eq:estimated_first_order_surface}
\end{equation}
and the second-order estimated tracking sliding surface
\begin{equation}
\hat{\boldsymbol{\sigma}}_2
=
\dot{\hat{\boldsymbol{\sigma}}}_1
+
C_2
\operatorname{sgn}\!\left(\hat{\boldsymbol{\sigma}}_1\right)
\left|\hat{\boldsymbol{\sigma}}_1\right|^{\nu_{\mathrm{asmc}}},
\qquad
C_2>0,
\quad
0<\nu_{\mathrm{asmc}}<1,
\label{eq:second_order_tracking_surface}
\end{equation}
where the sign and fractional-power operations are applied componentwise.

To expand $\hat{\boldsymbol{\sigma}}_2$ explicitly,
the auxiliary system dynamics and the phase-dependent feedforward input give
\begin{equation}
  \ddot{\hat{\boldsymbol{q}}}
  =
  \hat{\boldsymbol{a}}
  +
  \boldsymbol{u}_r
  +
  \hat{\boldsymbol{u}}_{\mathrm{asmc}},
  \qquad
  \boldsymbol{u}_r
  =
  -\boldsymbol{a}_n
  +
  \ddot{\boldsymbol{q}}_n,
  \label{eq:auxiliary_tracking_dynamics}
\end{equation}
where $\boldsymbol{a}_n = A_1 \boldsymbol{q}_n + A_2 \dot{\boldsymbol{q}}_n$
with $A_1$ and $A_2$ defined in Eq.~\eqref{eq:a_matrices}.
Subtracting $\ddot{\boldsymbol{q}}_n$ from
Eq.~\eqref{eq:auxiliary_tracking_dynamics} yields
\begin{equation}
  \ddot{\hat{\boldsymbol{e}}}
  =
  \hat{\boldsymbol{a}}-\boldsymbol{a}_n
  +
  \hat{\boldsymbol{u}}_{\mathrm{asmc}}.
  \label{eq:estimated_error_acceleration}
\end{equation}
Using
$\dot{\hat{\boldsymbol{\sigma}}}_1
=\ddot{\hat{\boldsymbol{e}}}+C_1\dot{\hat{\boldsymbol{e}}}$
and substituting Eq.~\eqref{eq:estimated_error_acceleration}
into Eq.~\eqref{eq:second_order_tracking_surface},
$\hat{\boldsymbol{\sigma}}_2$ can be written as
\begin{equation}
\hat{\boldsymbol{\sigma}}_2
=
\hat{\boldsymbol{a}}
-
\boldsymbol{a}_n
+
C_1\dot{\hat{\boldsymbol{e}}}
+
C_2
\operatorname{sgn}\!\left(\hat{\boldsymbol{\sigma}}_1\right)
\left|\hat{\boldsymbol{\sigma}}_1\right|^{\nu_{\mathrm{asmc}}}
+
\hat{\boldsymbol{u}}_{\mathrm{asmc}}.
\label{eq:estimated_sigma2_expanded}
\end{equation}
The estimated-state ASMC command is decomposed as
\begin{equation}
\hat{\boldsymbol{u}}_{\mathrm{asmc}}
=
\hat{\boldsymbol{u}}_{\mathrm{asmc}}^{\mathrm{can}}
+
\hat{\boldsymbol{u}}_{\mathrm{asmc}}^{\mathrm{ad}},
\label{eq:asmc_input_decomposition}
\end{equation}
where the cancellation component algebraically removes the known nonlinear surface term, whereas the adaptive component is generated dynamically using an adaptive gain to regulate $\hat{\boldsymbol{\sigma}}_2$. The cancellation component is selected as
\begin{equation}
\hat{\boldsymbol{u}}_{\mathrm{asmc}}^{\mathrm{can}}
=
-\hat{\boldsymbol{a}}
-C_1\dot{\hat{\boldsymbol{e}}}
-
C_2
\operatorname{sgn}\!\left(\hat{\boldsymbol{\sigma}}_1\right)
\left|\hat{\boldsymbol{\sigma}}_1\right|^{\nu_{\mathrm{asmc}}}.
\label{eq:second_order_asmc_cancellation_component}
\end{equation}

Substituting Eqs.~\eqref{eq:asmc_input_decomposition} and
\eqref{eq:second_order_asmc_cancellation_component} into
Eq.~\eqref{eq:estimated_sigma2_expanded} gives
\begin{equation}
\hat{\boldsymbol{\sigma}}_2
=
-\boldsymbol{a}_n
+
\hat{\boldsymbol{u}}_{\mathrm{asmc}}^{\mathrm{ad}}.
\label{eq:estimated_sigma2_reduced_dynamics}
\end{equation}
Unlike the SMDO cancellation, which relies on the mismatch between
the actual and auxiliary dynamics and therefore admits a residual
$\boldsymbol{\Delta}_{can}$, the ASMC cancellation component is
constructed entirely from quantities internal to the auxiliary system
and is therefore exact. The remaining term $-\boldsymbol{a}_n$ is
retained as the bounded perturbation addressed in
Theorem~\ref{thm:asmc_tracking_boundedness}.

The adaptive gain $K(t)$ and the fixed margin $K^*$ determine the
adaptive-component rate through $\hat{\boldsymbol{\sigma}}_2$:
\begin{equation}
\dot{\hat{\boldsymbol{u}}}_{\mathrm{asmc}}^{\mathrm{ad}}
=
-
\frac{K(t)+K^*}{\delta}
\hat{\boldsymbol{\sigma}}_2,
\qquad
K^*>0,
\quad
\delta>0,
\label{eq:second_order_asmc_control}
\end{equation}
\begin{equation}
\dot{K}(t)
=
\gamma
\left[
K(t)
\left(
\frac{\|\hat{\boldsymbol{\sigma}}_2(t)\|}{\delta}
-1
\right)
+
\frac{K^*}{\delta}
\|\hat{\boldsymbol{\sigma}}_2(t)\|
\right],
\qquad
\gamma>0,
\quad
K(0)>0.
\label{eq:second_order_asmc_gain_update}
\end{equation}
The gain update law~\eqref{eq:second_order_asmc_gain_update} follows
\cite{cho2020autonomous}, and the adaptive-component
rate~\eqref{eq:second_order_asmc_control} is an integral extension of
the corresponding form therein. The positivity of $K(t)$ follows from
Lemma~\ref{lem:generic_gain_upper_bound} in
Section~\ref{sec:boundedness_stability}. The gain increases while
$\hat{\boldsymbol{\sigma}}_2$ lies outside the $\delta$-boundary layer
and relaxes after $\hat{\boldsymbol{\sigma}}_2$ enters the prescribed layer.

Integrating Eq.~\eqref{eq:second_order_asmc_control} gives
\begin{equation}
\hat{\boldsymbol{u}}_{\mathrm{asmc}}^{\mathrm{ad}}(t)
=
\hat{\boldsymbol{u}}_{\mathrm{asmc}}^{\mathrm{ad}}(0)
-
\int_0^t
\frac{K(\tau)+K^*}{\delta}
\hat{\boldsymbol{\sigma}}_2(\tau)\,d\tau.
\label{eq:asmc_adaptive_component}
\end{equation}
Therefore, the total estimated-state ASMC command is
\begin{align}
\hat{\boldsymbol{u}}_{\mathrm{asmc}}(t)
={}&
-\hat{\boldsymbol{a}}(t)
-C_1\dot{\hat{\boldsymbol{e}}}(t)
-
C_2
\operatorname{sgn}\!\left(\hat{\boldsymbol{\sigma}}_1(t)\right)
\left|\hat{\boldsymbol{\sigma}}_1(t)\right|^{\nu_{\mathrm{asmc}}}
\notag\\
&+
\hat{\boldsymbol{u}}_{\mathrm{asmc}}^{\mathrm{ad}}(0)
-
\int_0^t
\frac{K(\tau)+K^*}{\delta}
\hat{\boldsymbol{\sigma}}_2(\tau)\,d\tau.
\label{eq:asmc_total}
\end{align}

The SMDO and ASMC adaptive components share the same adaptive boundary-layer
architecture. In particular,
Eqs.~\eqref{eq:smdo_adaptive_injection}--\eqref{eq:smdo_adaptive_gain_update}
and
Eqs.~\eqref{eq:second_order_asmc_control}--\eqref{eq:second_order_asmc_gain_update}
are structurally identical under the correspondence
\[
(\hat{\boldsymbol{s}}_2,\,L,\,L^*,\,\epsilon,\,\eta)
\leftrightarrow
(\hat{\boldsymbol{\sigma}}_2,\,K,\,K^*,\,\delta,\,\gamma).
\]
This common structure is formalized in
Section~\ref{sec:boundedness_stability}, where a single canonical lemma is
applied to both layers.

In summary, the two compensation signals entering
Eq.~\eqref{eq:total_control_input} are
\begin{align}
\hat{\boldsymbol{u}}_{\mathrm{smdo}}(t)
={}&
\underbrace{
-\tilde{\boldsymbol{a}}(t)
-
\lambda_1\dot{\tilde{\boldsymbol{q}}}(t)
-
\lambda_2
\operatorname{sgn}\!\left(\hat{\boldsymbol{s}}_1(t)\right)
\left|\hat{\boldsymbol{s}}_1(t)\right|^{\nu_{\mathrm{smdo}}}
}_{\hat{\boldsymbol{u}}_{\mathrm{smdo}}^{\mathrm{can}}(t)}
\notag\\
&+
\underbrace{
\left[
\hat{\boldsymbol{u}}_{\mathrm{smdo}}^{\mathrm{ad}}(0)
-
\int_0^t
\frac{L(\tau)+L^*}{\epsilon}
\hat{\boldsymbol{s}}_2(\tau)\,d\tau
\right]
}_{\hat{\boldsymbol{u}}_{\mathrm{smdo}}^{\mathrm{ad}}(t)},
\label{eq:smdo_explicit}\\[6pt]
\hat{\boldsymbol{u}}_{\mathrm{asmc}}(t)
={}&
\underbrace{
-\hat{\boldsymbol{a}}(t)
-C_1\dot{\hat{\boldsymbol{e}}}(t)
-
C_2
\operatorname{sgn}\!\left(\hat{\boldsymbol{\sigma}}_1(t)\right)
\left|\hat{\boldsymbol{\sigma}}_1(t)\right|^{\nu_{\mathrm{asmc}}}
}_{\hat{\boldsymbol{u}}_{\mathrm{asmc}}^{\mathrm{can}}(t)}
\notag\\
&+
\underbrace{
\left[
\hat{\boldsymbol{u}}_{\mathrm{asmc}}^{\mathrm{ad}}(0)
-
\int_0^t
\frac{K(\tau)+K^*}{\delta}
\hat{\boldsymbol{\sigma}}_2(\tau)\,d\tau
\right]
}_{\hat{\boldsymbol{u}}_{\mathrm{asmc}}^{\mathrm{ad}}(t)}.
\label{eq:asmc_explicit}
\end{align}

\subsection{Integrated Control Architecture}
\label{subsec:integrated_control_architecture}

In brief, during the Reaching Phase, the total command in acceleration level is
\begin{equation}
\boldsymbol{u}_{\mathrm{total}}
=
\boldsymbol{u}_r
+
\hat{\boldsymbol{u}}_{\mathrm{asmc}}
+
\hat{\boldsymbol{u}}_{\mathrm{smdo}}.
\label{eq:total_control_reaching}
\end{equation}
During the PCO Phase, $\hat{\boldsymbol{u}}_r$ is set to zero so that
\begin{equation}
\boldsymbol{u}_{\mathrm{total}}
=
\hat{\boldsymbol{u}}_{\mathrm{asmc}}
+
\hat{\boldsymbol{u}}_{\mathrm{smdo}}.
\label{eq:total_control_pco}
\end{equation}
The SMDO operates on $\hat{\boldsymbol{s}}_2$ to reduce $\boldsymbol{d}$ to the
residual $\tilde{\boldsymbol{d}}_{\mathrm{smdo}}$; the ASMC then operates on
$\hat{\boldsymbol{\sigma}}_2$ to track $\boldsymbol{q}_n(t)$ against this residual
and the observer-induced mismatch $\tilde{\boldsymbol{q}}$. The surrogate
construction of Section~\ref{subsec:surrogate_surface_derivatives} is purely an
implementation device for evaluating $\hat{\boldsymbol{\sigma}}_2$; it does not
constitute a separate first-order feedback loop. The stability and boundedness
properties of the closed-loop system are established in
Section~\ref{sec:boundedness_stability}.

\section{Boundedness and Stability Analysis}
\label{sec:boundedness_stability}

\subsection*{Canonical second-order sliding variable form}

The SMDO observer layer and the ASMC tracking layer share the same 
adaptive boundary-layer structure. To avoid repeating the
cascade argument for each layer, a generic canonical form is introduced
here and applied to both. Let $\boldsymbol{x}\in\mathbb{R}^{3}$ be a
generic error vector and define
\begin{equation}
  \boldsymbol{z}_1
  =
  \dot{\boldsymbol{x}}
  +
  \alpha\boldsymbol{x},
  \qquad
  \alpha>0,
  \label{eq:canonical_z1}
\end{equation}
and
\begin{equation}
  \boldsymbol{z}_2
  =
  \dot{\boldsymbol{z}}_1
  +
  \beta
  \operatorname{sgn}(\boldsymbol{z}_1)
  |\boldsymbol{z}_1|^\nu,
  \qquad
  \beta>0,
  \quad
  0<\nu<1,
  \label{eq:canonical_z2}
\end{equation}
where the sign and fractional-power operations are applied componentwise.
The coefficient $\alpha$ determines the first-order error dynamics, and
$\beta$ sets the strength of the terminal nonlinear term.

The correspondence between the canonical variables and the two layers is
given in Table~\ref{tab:canonical_mapping}. An important distinction
concerns the ASMC layer: the canonical result is applied to the
\emph{estimated} tracking variables $\hat{\boldsymbol{e}}$,
$\hat{\boldsymbol{\sigma}}_1$, and $\hat{\boldsymbol{\sigma}}_2$,
not to the actual tracking error. The actual error is recovered
subsequently via $\boldsymbol{e}=\hat{\boldsymbol{e}}+\tilde{\boldsymbol{q}}$,
using the observer-error bounds from
Theorem~\ref{thm:smdo_residual_boundedness}.

\begin{table}[h]
  \centering
  \caption{Correspondence between the canonical form and the two control layers.}
  \label{tab:canonical_mapping}
  \begin{tabular}{c|c|c}
    \hline
    Canonical & SMDO layer & ASMC layer \\
    \hline
    $\boldsymbol{x}$  & $\tilde{\boldsymbol{q}}$       & $\hat{\boldsymbol{e}}$ \\
    $\boldsymbol{z}_1$ & $\hat{\boldsymbol{s}}_1$      & $\hat{\boldsymbol{\sigma}}_1$ \\
    $\boldsymbol{z}_2$ & $\hat{\boldsymbol{s}}_2$      & $\hat{\boldsymbol{\sigma}}_2$ \\
    $\alpha$          & $\lambda_1$                    & $C_1$ \\
    $\beta$           & $\lambda_2$                    & $C_2$ \\
    $\nu$             & $\nu_{\mathrm{smdo}}$           & $\nu_{\mathrm{asmc}}$ \\
    $\Delta$          & $\epsilon$                     & $\delta$ \\
    $P$               & $L$                            & $K$ \\
    $P^*$             & $L^*$                          & $K^*$ \\
    $\kappa$          & $\eta$                         & $\gamma$ \\
    \hline
  \end{tabular}
\end{table}

\subsection*{Canonical adaptive gain and input-rate laws}

The adaptive gain update law is written in the generic form
\begin{equation}
  \dot{P}(t)
  =
  \kappa
  \!\left[
    P(t)
    \!\left(
      \frac{\|\boldsymbol{z}_2(t)\|}{\Delta}
      -1
    \right)
    +
    \frac{P^*}{\Delta}
    \|\boldsymbol{z}_2(t)\|
  \right],
  \qquad
  P(0)>0,
  \label{eq:canonical_gain_update}
\end{equation}
where $P^*>0$ is a fixed gain margin, $\kappa>0$ is the adaptation rate,
and $\Delta>0$ is the prescribed boundary-layer width. The corresponding
input-rate law is
\begin{equation}
  \dot{\boldsymbol{u}}(t)
  =
  -
  \frac{P(t)+P^*}{\Delta}
  \boldsymbol{z}_2(t).
  \label{eq:canonical_adaptive_input}
\end{equation}
The dynamics of the sliding variable $\boldsymbol{z}_2$ are represented as
\begin{equation}
  \dot{\boldsymbol{z}}_2
  =
  \boldsymbol{\xi}(t)
  +
  \dot{\boldsymbol{u}}(t),
  \label{eq:z2_dynamics_generic}
\end{equation}
where $\boldsymbol{\xi}(t)$ collects all terms not directly assigned to
$\dot{\boldsymbol{u}}$.

\begin{lemma}[Gain upper bound under admissible input-rate authority]
\label{lem:generic_gain_upper_bound}
Consider Eqs.~\eqref{eq:canonical_gain_update} and
\eqref{eq:canonical_adaptive_input}. Suppose that $P(0)\in(0,U^*)$ for
some finite constant $U^*>0$, and that $\|\dot{\boldsymbol{u}}(t)\|\leq U^*$
for all $t\geq0$. This admissibility condition is not derived from the
adaptive gain law itself but is imposed externally, reflecting actuator-rate
limitations, propulsion-system bandwidth, and command-rate saturation
enforced at the implementation level of digital spacecraft control systems;
all subsequent boundedness results are conditional on the closed-loop
trajectory satisfying this bound.
Then $P(t)>0$ for all $t\geq0$, and
\begin{equation}
  0<P(t)<U^*
  \label{eq:generic_gain_upper_result}
\end{equation}
whenever $\|\boldsymbol{z}_2(t)\|>\Delta$.
\end{lemma}

\begin{proof}
Let $r(t):=\|\boldsymbol{z}_2(t)\|/\Delta\geq0$.
Equation~\eqref{eq:canonical_gain_update} becomes
\begin{equation}
  \dot{P}(t)
  =
  \kappa(r(t)-1)P(t)
  +
  \kappa r(t)P^*.
  \label{eq:Pdot_rewritten}
\end{equation}
Since $\kappa,P^*>0$ and $r(t)\geq0$, one has
$\dot{P}(t)\geq\kappa(r(t)-1)P(t)$. Gr\"{o}nwall's inequality
(\cite{gronwall1919note}) therefore gives
\[
P(t)\geq P(0)\exp\!\left(\int_0^t\kappa(r(\tau)-1)\,d\tau\right)>0
\]
for all $t\geq0$.

When $r(t)>1$, all terms in Eq.~\eqref{eq:Pdot_rewritten} are positive,
so $P$ is strictly increasing. Suppose $P(t_1)\geq U^*$ at some $t_1$
with $\|\boldsymbol{z}_2(t_1)\|>\Delta$. Then
\[
  \|\dot{\boldsymbol{u}}(t_1)\|
  =
  r(t_1)\bigl[P(t_1)+P^*\bigr]
  >
  P(t_1)+P^*
  \geq
  U^*+P^*
  >
  U^*,
\]
which contradicts the admissibility bound $\|\dot{\boldsymbol{u}}\|\leq U^*$.
Hence $P(t)<U^*$ whenever $\|\boldsymbol{z}_2(t)\|>\Delta$.
\end{proof}

\begin{lemma}[Cascade ultimate boundedness under the adaptive second-order sliding law]
\label{lem:canonical_adaptive_boundary_layer}
Consider the canonical surfaces in
Eqs.~\eqref{eq:canonical_z1}--\eqref{eq:canonical_z2} and the
second-order dynamics in Eq.~\eqref{eq:z2_dynamics_generic}. Assume that
there exists a finite constant $\Gamma>0$ such that
\begin{equation}
  \|\boldsymbol{\xi}(t)\|
  \leq
  \Gamma,
  \qquad
  \forall\,t\geq0.
  \label{eq:canonical_xi_bound}
\end{equation}
Let the input-rate and gain-update laws be given by
Eqs.~\eqref{eq:canonical_adaptive_input} and
\eqref{eq:canonical_gain_update}, respectively. Suppose that
\begin{equation}
  \|\dot{\boldsymbol{u}}(t)\|
  \leq
  U^*,
  \qquad
  \forall\,t\geq0,
  \label{eq:canonical_input_rate_bound}
\end{equation}
where $U^*>\Gamma$. The bound in~\eqref{eq:canonical_input_rate_bound}
constitutes an admissibility condition of the type introduced in
Lemma~\ref{lem:generic_gain_upper_bound}: it is imposed by
actuator-rate authority and command-rate saturation at the implementation
level, not derived from the adaptive law itself. The conclusions of this
lemma are therefore conditional on the closed-loop trajectory
satisfying~\eqref{eq:canonical_input_rate_bound} throughout the operating
interval. Suppose further that the hypotheses of
Lemma~\ref{lem:generic_gain_upper_bound} are satisfied. For finite initial
conditions $\boldsymbol{x}(0)$ and $\dot{\boldsymbol{x}}(0)$, the following are satisfied:
\begin{align}
  \limsup_{t\to\infty}
  \|\boldsymbol{z}_2(t)\|
  &\leq
  \Delta,
  \label{eq:canonical_z2_limsup}\\
  \limsup_{t\to\infty}
  \|\boldsymbol{z}_1(t)\|_\infty
  &\leq
  \left(
    \frac{\Delta}{\beta}
  \right)^{1/\nu},
  \label{eq:canonical_z1_limsup}\\
  \limsup_{t\to\infty}
  \|\boldsymbol{x}(t)\|_\infty
  &\leq
  \frac{1}{\alpha}
  \left(
    \frac{\Delta}{\beta}
  \right)^{1/\nu},
  \label{eq:canonical_x_limsup}\\
  \limsup_{t\to\infty}
  \|\dot{\boldsymbol{x}}(t)\|_\infty
  &\leq
  2
  \left(
    \frac{\Delta}{\beta}
  \right)^{1/\nu}.
  \label{eq:canonical_xdot_limsup}
\end{align}
\end{lemma}
\begin{proof}
Substituting Eq.~\eqref{eq:canonical_adaptive_input} into
Eq.~\eqref{eq:z2_dynamics_generic} gives
\begin{equation}
  \dot{\boldsymbol{z}}_2
  =
  \boldsymbol{\xi}(t)
  -
  \frac{P(t)+P^*}{\Delta}
  \boldsymbol{z}_2(t).
  \label{eq:canonical_z2_closed}
\end{equation}
By Lemma~\ref{lem:generic_gain_upper_bound}, $0<P(t)<U^*$ whenever
$\|\boldsymbol{z}_2(t)\|>\Delta$. Define
\begin{equation}
  V
  =
  \frac{1}{2}
  \boldsymbol{z}_2^{\mathsf{T}}\boldsymbol{z}_2
  +
  \frac{1}{2\chi}
  (P(t)-U^*)^2,
  \qquad
  \chi
  :=
  \frac{\kappa P^*}{\Delta}
  >0.
  \label{eq:canonical_Lyapunov}
\end{equation}
Differentiating Eq.~\eqref{eq:canonical_Lyapunov} along Eq.~\eqref{eq:canonical_z2_closed} and applying Eq.~\eqref{eq:canonical_xi_bound}, we have
\begin{equation}
  \dot{V}
  \leq
  \Gamma\|\boldsymbol{z}_2\|
  -
  \frac{P+P^*}{\Delta}
  \|\boldsymbol{z}_2\|^2
  -
  \frac{U^*-P}{\chi}
  \dot{P}.
  \label{eq:canonical_Vdot_2}
\end{equation}
Since $\|\boldsymbol{z}_2\|^2/\Delta
>\|\boldsymbol{z}_2\|$ for $\|\boldsymbol{z}_2\|>\Delta$, Eq.~\eqref{eq:canonical_Vdot_2} becomes
\begin{equation}
  \dot{V}
  \leq
  -(U^*-\Gamma+P^*)
  \|\boldsymbol{z}_2\|
  -
  (U^*-P)
  \!\left(
    \frac{\dot{P}}{\chi}
    -
    \|\boldsymbol{z}_2\|
  \right).
  \label{eq:canonical_Vdot_factored}
\end{equation}
The first term is strictly negative because $U^*>\Gamma$ and $P^*>0$.
The second term is nonpositive: $U^*-P>0$ by
Lemma~\ref{lem:generic_gain_upper_bound}, and from
Eq.~\eqref{eq:canonical_gain_update},
\[
  \frac{\dot{P}}{\|\boldsymbol{z}_2\|}
  =
  \kappa
  \!\left[
    P\!\left(\frac{1}{\Delta}-\frac{1}{\|\boldsymbol{z}_2\|}\right)
    +
    \frac{P^*}{\Delta}
  \right]
  >
  \frac{\kappa P^*}{\Delta}
  =
  \chi
\]
for $\|\boldsymbol{z}_2\|>\Delta$.
Therefore $\dot{V}<0$ outside
$\Omega_\Delta=\{\boldsymbol{z}_2:\|\boldsymbol{z}_2\|\leq\Delta\}$,
which establishes Eq.~\eqref{eq:canonical_z2_limsup}.

For the bound on $\boldsymbol{z}_1$, rewrite
Eq.~\eqref{eq:canonical_z2} componentwise as
\begin{equation}
  \dot{z}_{1,i}
  +
  \beta
  \operatorname{sgn}(z_{1,i})
  |z_{1,i}|^\nu
  =
  z_{2,i}.
  \label{eq:canonical_z1_component_dynamics}
\end{equation}
Once $\boldsymbol{z}_2$ has reached its ultimate boundary layer in
Eq.~\eqref{eq:canonical_z2_limsup}, each component satisfies
\begin{equation}
  |z_{2,i}|
  \leq
  \|\boldsymbol{z}_2\|
  \leq
  \Delta.
  \label{eq:canonical_z2_component_bound}
\end{equation}
Combining Eqs.~\eqref{eq:canonical_z1_component_dynamics} and
\eqref{eq:canonical_z2_component_bound} gives
\begin{equation}
  -\Delta
  \leq
  \dot{z}_{1,i}
  +
  \beta
  \operatorname{sgn}(z_{1,i})
  |z_{1,i}|^\nu
  \leq
  \Delta.
  \label{eq:canonical_z1_double_inequality}
\end{equation}

Define the interval
\begin{equation}
  \mathcal{L}_{\Delta,i}
  :=
  \left\{
    z_{1,i}\in\mathbb{R}
    :
    |z_{1,i}|
    \leq
    \left(
      \frac{\Delta}{\beta}
    \right)^{1/\nu}
  \right\}.
  \label{eq:canonical_z1_attracting_set}
\end{equation}
To show that $\mathcal{L}_{\Delta,i}$ is attracting, consider the
Lyapunov function
\begin{equation}
  W_i
  =
  \frac{1}{2}z_{1,i}^2.
  \label{eq:canonical_Wi_def}
\end{equation}
Its derivative is
\begin{equation}
  \dot{W}_i
  =
  z_{1,i}\dot{z}_{1,i}.
  \label{eq:canonical_Wi_dot_def}
\end{equation}
Two cases must be considered outside
$\mathcal{L}_{\Delta,i}$.

If
\begin{equation}
  z_{1,i}
  >
  \left(
    \frac{\Delta}{\beta}
  \right)^{1/\nu}
  >
  0,
  \label{eq:canonical_z1_positive_case}
\end{equation}
then
$\beta z_{1,i}^{\nu}>\Delta$. From the upper inequality in
Eq.~\eqref{eq:canonical_z1_double_inequality},
\begin{equation}
  \dot{z}_{1,i}
  \leq
  \Delta
  -
  \beta z_{1,i}^{\nu}
  <
  0.
  \label{eq:canonical_z1dot_positive_case}
\end{equation}
Since $z_{1,i}>0$,
\begin{equation}
  \dot{W}_i
  =
  z_{1,i}\dot{z}_{1,i}
  \leq
  z_{1,i}
  \left(
    \Delta-\beta z_{1,i}^{\nu}
  \right)
  <
  0.
  \label{eq:canonical_Wdot_positive_case}
\end{equation}

Alternatively, if
\begin{equation}
  z_{1,i}
  <
  -
  \left(
    \frac{\Delta}{\beta}
  \right)^{1/\nu}
  <
  0,
  \label{eq:canonical_z1_negative_case}
\end{equation}
then
$\beta|z_{1,i}|^\nu>\Delta$. Because
$\operatorname{sgn}(z_{1,i})=-1$, the lower inequality in
Eq.~\eqref{eq:canonical_z1_double_inequality} gives
\begin{equation}
  \dot{z}_{1,i}
  \geq
  -\Delta
  +
  \beta|z_{1,i}|^\nu
  >
  0.
  \label{eq:canonical_z1dot_negative_case}
\end{equation}
Since $z_{1,i}<0$,
\begin{equation}
  \dot{W}_i
  =
  z_{1,i}\dot{z}_{1,i}
  \leq
  z_{1,i}
  \left(
    -\Delta
    +
    \beta|z_{1,i}|^\nu
  \right)
  <
  0.
  \label{eq:canonical_Wdot_negative_case}
\end{equation}

Thus, $\dot{W}_i<0$ for both possible cases outside
$\mathcal{L}_{\Delta,i}$. Hence, each component of
$\boldsymbol{z}_1$ is ultimately attracted to
$\mathcal{L}_{\Delta,i}$, and
\begin{equation}
  \limsup_{t\to\infty}
  \|\boldsymbol{z}_1(t)\|_\infty
  \leq
  \left(
    \frac{\Delta}{\beta}
  \right)^{1/\nu},
\end{equation}
which proves Eq.~\eqref{eq:canonical_z1_limsup}.

From Eq.~\eqref{eq:canonical_z1}, the dynamics of
$\boldsymbol{x}$ are
\begin{equation}
  \dot{\boldsymbol{x}}
  =
  -\alpha\boldsymbol{x}
  +
  \boldsymbol{z}_1.
  \label{eq:canonical_x_dynamics}
\end{equation}
The corresponding variation-of-constants solution is
\begin{equation}
  \boldsymbol{x}(t)
  =
  e^{-\alpha(t-t_0)}
  \boldsymbol{x}(t_0)
  +
  \int_{t_0}^{t}
  e^{-\alpha(t-\tau)}
  \boldsymbol{z}_1(\tau)\,d\tau.
  \label{eq:canonical_x_variation}
\end{equation}
Taking the infinity norm gives
\begin{align}
  \|\boldsymbol{x}(t)\|_\infty
  &\leq
  e^{-\alpha(t-t_0)}
  \|\boldsymbol{x}(t_0)\|_\infty
  \notag\\
  &\quad+
  \int_{t_0}^{t}
  e^{-\alpha(t-\tau)}
  \|\boldsymbol{z}_1(\tau)\|_\infty\,d\tau.
  \label{eq:canonical_x_norm_variation}
\end{align}
Using Eq.~\eqref{eq:canonical_z1_limsup} and
\[
  \int_{t_0}^{t}
  e^{-\alpha(t-\tau)}\,d\tau
  =
  \frac{1-e^{-\alpha(t-t_0)}}{\alpha},
\]
the upper limit of Eq.~\eqref{eq:canonical_x_norm_variation} satisfies
\begin{align}
  \limsup_{t\to\infty}
  \|\boldsymbol{x}(t)\|_\infty
  &\leq
  \frac{1}{\alpha}
  \limsup_{t\to\infty}
  \|\boldsymbol{z}_1(t)\|_\infty
  \notag\\
  &\leq
  \frac{1}{\alpha}
  \left(
    \frac{\Delta}{\beta}
  \right)^{1/\nu},
\end{align}
which proves Eq.~\eqref{eq:canonical_x_limsup}.

Finally, Eq.~\eqref{eq:canonical_x_dynamics} gives
\begin{equation}
  \|\dot{\boldsymbol{x}}(t)\|_\infty
  \leq
  \|\boldsymbol{z}_1(t)\|_\infty
  +
  \alpha\|\boldsymbol{x}(t)\|_\infty.
  \label{eq:canonical_xdot_norm_bound}
\end{equation}
Therefore,
\begin{align}
  \limsup_{t\to\infty}
  \|\dot{\boldsymbol{x}}(t)\|_\infty
  &\leq
  \left(
    \frac{\Delta}{\beta}
  \right)^{1/\nu}
  +
  \alpha
  \left[
    \frac{1}{\alpha}
    \left(
      \frac{\Delta}{\beta}
    \right)^{1/\nu}
  \right]
  \notag\\
  &=
  2
  \left(
    \frac{\Delta}{\beta}
  \right)^{1/\nu},
\end{align}
which establishes Eq.~\eqref{eq:canonical_xdot_limsup}.
\end{proof}

\begin{theorem}[SMDO residual-disturbance boundedness]
\label{thm:smdo_residual_boundedness}
Define
\begin{equation}
  \hat{\boldsymbol{s}}_1
  =
  \dot{\tilde{\boldsymbol{q}}}
  +
  \lambda_1\tilde{\boldsymbol{q}},
  \qquad
  \hat{\boldsymbol{s}}_2
  =
  \dot{\hat{\boldsymbol{s}}}_1
  +
  \lambda_2
  \operatorname{sgn}(\hat{\boldsymbol{s}}_1)
  |\hat{\boldsymbol{s}}_1|^{\nu_{\mathrm{smdo}}},
  \qquad
  \lambda_1 > 0,
  \quad
  \lambda_2 > 0,
  \quad
  0<\nu_{\mathrm{smdo}}<1,
  \label{eq:smdo_surfaces_theorem}
\end{equation}
and decompose the SMDO compensation as
$\hat{\boldsymbol{u}}_{\mathrm{smdo}}
=\hat{\boldsymbol{u}}_{\mathrm{smdo}}^{\mathrm{can}}
+\hat{\boldsymbol{u}}_{\mathrm{smdo}}^{\mathrm{ad}}$.
Here, the cancellation component removes the known model- and
surface-dependent terms, while the adaptive component is generated
dynamically using the adaptive gain. The cancellation component is given by Eq.~\eqref{eq:smdo_cancellation_component}.
The SMDO adaptive component follows
Eqs.~\eqref{eq:smdo_adaptive_injection}
and~\eqref{eq:smdo_adaptive_gain_update}, which correspond to the canonical
adaptive-input and gain-update forms in
Eqs.~\eqref{eq:canonical_adaptive_input}
and~\eqref{eq:canonical_gain_update}, respectively.
Define the residual disturbance after SMDO compensation as
\begin{equation}
  \tilde{\boldsymbol{d}}_{\mathrm{smdo}}
  :=
  \boldsymbol{d}_{\mathrm{eff}}
  +
  \hat{\boldsymbol{u}}_{\mathrm{smdo}}^{\mathrm{ad}}.
  \label{eq:smdo_residual_def_theorem}
\end{equation}
Assume:
\begin{enumerate}[label=\textup{(S\arabic*)}]
  \item $\boldsymbol{d}_{\mathrm{eff}}(t)$ is continuously differentiable
        with $\|\dot{\boldsymbol{d}}_{\mathrm{eff}}(t)\|\leq \Gamma_s$
        for all $t\geq0$ where $\Gamma_s$ is a finite constant.
  \item $L(0)\in(0,U_s^*)$ for some $U_s^*>\Gamma_s$.
  \item $\|\dot{\hat{\boldsymbol{u}}}_{\mathrm{smdo}}^{\mathrm{ad}}(t)\|\leq U_s^*$
        for all $t\geq0$, reflecting actuator-rate saturation at the
        implementation level.
  \item $\tilde{\boldsymbol{q}}(0)$ and
        $\dot{\tilde{\boldsymbol{q}}}(0)$ are finite.
\end{enumerate}

Then
\begin{align}
  \limsup_{t\to\infty}
  \|\hat{\boldsymbol{s}}_2(t)\|
  &\leq
  \epsilon,
  \label{eq:smdo_s2_bound_theorem}\\
  \limsup_{t\to\infty}
  \|\hat{\boldsymbol{s}}_1(t)\|_\infty
  &\leq
  \left(
    \frac{\epsilon}{\lambda_2}
  \right)^{1/\nu_{\mathrm{smdo}}},
  \label{eq:smdo_s1_bound_theorem}\\
  \limsup_{t\to\infty}
  \|\tilde{\boldsymbol{q}}(t)\|_\infty
  &\leq
  \frac{1}{\lambda_1}
  \left(
    \frac{\epsilon}{\lambda_2}
  \right)^{1/\nu_{\mathrm{smdo}}},
  \label{eq:smdo_qtilde_bound_theorem}\\
  \limsup_{t\to\infty}
  \|\dot{\tilde{\boldsymbol{q}}}(t)\|_\infty
  &\leq
  2
  \left(
    \frac{\epsilon}{\lambda_2}
  \right)^{1/\nu_{\mathrm{smdo}}},
  \label{eq:smdo_qdottilde_bound_theorem}\\
  \limsup_{t\to\infty}
  \|\tilde{\boldsymbol{d}}_{\mathrm{smdo}}(t)\|
  &\leq
  \epsilon.
  \label{eq:smdo_residual_bound_theorem}
\end{align}
\end{theorem}

\begin{proof}
The actual system and auxiliary observer satisfy
\begin{align}
  \ddot{\boldsymbol{q}}
  &=
  \boldsymbol{a}
  +
  \boldsymbol{u}_r
  +
  \hat{\boldsymbol{u}}_{\mathrm{asmc}}
  +
  \hat{\boldsymbol{u}}_{\mathrm{smdo}}
  +
  \boldsymbol{d},
  \label{eq:actual_system_smdo_proof}\\
  \ddot{\hat{\boldsymbol{q}}}
  &=
  \hat{\boldsymbol{a}}
  +
  \boldsymbol{u}_r
  +
  \hat{\boldsymbol{u}}_{\mathrm{asmc}}.
  \label{eq:auxiliary_system_smdo_proof}
\end{align}
Subtracting and writing $\tilde{\boldsymbol{a}}:=\boldsymbol{a}-\hat{\boldsymbol{a}}$ yields
\begin{equation}
  \ddot{\tilde{\boldsymbol{q}}}
  =
  \tilde{\boldsymbol{a}}
  +
  \hat{\boldsymbol{u}}_{\mathrm{smdo}}
  +
  \boldsymbol{d}.
  \label{eq:qtilde_dynamics}
\end{equation}
Substituting the decomposition
$\hat{\boldsymbol{u}}_{\mathrm{smdo}}
=\hat{\boldsymbol{u}}_{\mathrm{smdo}}^{\mathrm{can}}
+\hat{\boldsymbol{u}}_{\mathrm{smdo}}^{\mathrm{ad}}$
with Eq.~\eqref{eq:smdo_cancellation_component} into
Eq.~\eqref{eq:qtilde_dynamics} gives
\begin{equation}
  \ddot{\tilde{\boldsymbol{q}}}
  =
  -\lambda_1\dot{\tilde{\boldsymbol{q}}}
  -
  \lambda_2
  \operatorname{sgn}(\hat{\boldsymbol{s}}_1)
  |\hat{\boldsymbol{s}}_1|^{\nu_{\mathrm{smdo}}}
  +
  \hat{\boldsymbol{u}}_{\mathrm{smdo}}^{\mathrm{ad}}
  +
  \boldsymbol{d}
  +
  \boldsymbol{\Delta}_{\mathrm{can}},
  \label{eq:qtilde_reduced_dynamics}
\end{equation}
where $\boldsymbol{\Delta}_{\mathrm{can}}$ denotes any residual mismatch
from the imperfect cancellation of $\tilde{\boldsymbol{a}}$ and the
surface-dependent terms at the implementation level.
From the definition of $\hat{\boldsymbol{s}}_1$ and
Eq.~\eqref{eq:qtilde_reduced_dynamics},
\begin{equation}
  \dot{\hat{\boldsymbol{s}}}_1
  =
  \ddot{\tilde{\boldsymbol{q}}}
  +
  \lambda_1\dot{\tilde{\boldsymbol{q}}}
  =
  -\lambda_2
  \operatorname{sgn}(\hat{\boldsymbol{s}}_1)
  |\hat{\boldsymbol{s}}_1|^{\nu_{\mathrm{smdo}}}
  +
  \hat{\boldsymbol{u}}_{\mathrm{smdo}}^{\mathrm{ad}}
  +
  \boldsymbol{d}
  +
  \boldsymbol{\Delta}_{\mathrm{can}},
\end{equation}
and therefore
\begin{equation}
  \hat{\boldsymbol{s}}_2
  =
  \dot{\hat{\boldsymbol{s}}}_1
  +
  \lambda_2
  \operatorname{sgn}(\hat{\boldsymbol{s}}_1)
  |\hat{\boldsymbol{s}}_1|^{\nu_{\mathrm{smdo}}}
  =
  \boldsymbol{d}
  +
  \boldsymbol{\Delta}_{\mathrm{can}}
  +
  \hat{\boldsymbol{u}}_{\mathrm{smdo}}^{\mathrm{ad}}
  =:
  \boldsymbol{d}_{\mathrm{eff}}
  +
  \hat{\boldsymbol{u}}_{\mathrm{smdo}}^{\mathrm{ad}}
  =
  \tilde{\boldsymbol{d}}_{\mathrm{smdo}},
  \label{eq:s2_residual_identity}
\end{equation}
where $\boldsymbol{d}_{\mathrm{eff}} := \boldsymbol{d} + \boldsymbol{\Delta}_{\mathrm{can}}$
denotes the effective lumped disturbance defined by~\eqref{eq:deff_definition}. Differentiating Eq.~\eqref{eq:s2_residual_identity}, we obtain
\begin{equation}
  \dot{\hat{\boldsymbol{s}}}_2
  =
  \dot{\boldsymbol{d}}_{\mathrm{eff}}
  +
  \dot{\hat{\boldsymbol{u}}}_{\mathrm{smdo}}^{\mathrm{ad}}.
  \label{eq:s2dot_smdo_canonical}
\end{equation}
Equation~\eqref{eq:s2dot_smdo_canonical} matches
Eq.~\eqref{eq:z2_dynamics_generic} with
\[
  \boldsymbol{\xi}_{\mathrm{smdo}}
  =
  \dot{\boldsymbol{d}}_{\mathrm{eff}},
  \qquad
  \|\boldsymbol{\xi}_{\mathrm{smdo}}(t)\|
  \leq
  \Gamma_s,
\]
where the bound follows from assumption~(S1). Under the correspondence
\[
  (\boldsymbol{x},\,
  \boldsymbol{z}_1,\,
  \boldsymbol{z}_2,\,
  \alpha,\,\beta,\,
  P,\,P^*,\,\Delta,\,\kappa,\,\Gamma,\,U^*,\,\nu)
  \leftrightarrow
  (\tilde{\boldsymbol{q}},\,
  \hat{\boldsymbol{s}}_1,\,
  \hat{\boldsymbol{s}}_2,\,
  \lambda_1,\,\lambda_2,\,
  L,\,L^*,\,\epsilon,\,\eta,\,\Gamma_s,\,U_s^*,\,\nu_{smdo}),
\]
assumptions~(S1)--(S4) verify all hypotheses of
Lemma~\ref{lem:canonical_adaptive_boundary_layer}, and
Eqs.~\eqref{eq:smdo_s2_bound_theorem}--\eqref{eq:smdo_qdottilde_bound_theorem}
follow directly from
Eqs.~\eqref{eq:canonical_z2_limsup}--\eqref{eq:canonical_xdot_limsup}.
Finally, the identity~\eqref{eq:s2_residual_identity} and Eq.~\eqref{eq:smdo_s2_bound_theorem}
give $\limsup_{t\to\infty}\|\tilde{\boldsymbol{d}}_{\mathrm{smdo}}(t)\|\leq\epsilon$,
which is Eq.~\eqref{eq:smdo_residual_bound_theorem}.
\end{proof}

\begin{theorem}[ASMC tracking-error boundedness]
\label{thm:asmc_tracking_boundedness}
Define the estimated tracking error and estimated ASMC surfaces as
\begin{equation}
  \hat{\boldsymbol{e}}
  =
  \hat{\boldsymbol{q}}
  -
  \boldsymbol{q}_n,
  \qquad
  \hat{\boldsymbol{\sigma}}_1
  =
  \dot{\hat{\boldsymbol{e}}}
  +
  C_1\hat{\boldsymbol{e}},
  \label{eq:estimated_error_sigma1_theorem}
\end{equation}
and
\begin{equation}
  \hat{\boldsymbol{\sigma}}_2
  =
  \dot{\hat{\boldsymbol{\sigma}}}_1
  +
  C_2
  \operatorname{sgn}(\hat{\boldsymbol{\sigma}}_1)
  |\hat{\boldsymbol{\sigma}}_1|^{\nu_{\mathrm{asmc}}},
  \qquad
  0<\nu_{\mathrm{asmc}}<1.
  \label{eq:estimated_sigma2_theorem}
\end{equation}
The estimated-state ASMC command is decomposed as
$\hat{\boldsymbol{u}}_{\mathrm{asmc}}
=\hat{\boldsymbol{u}}_{\mathrm{asmc}}^{\mathrm{can}}
+\hat{\boldsymbol{u}}_{\mathrm{asmc}}^{\mathrm{ad}}$.
Here, the cancellation component removes the known nonlinear surface
term, while the adaptive component is generated dynamically using the
adaptive gain. The cancellation component is given by Eq.~\eqref{eq:second_order_asmc_cancellation_component}.
The ASMC adaptive-component rate and gain update are specified in
Eqs.~\eqref{eq:second_order_asmc_control}
and~\eqref{eq:second_order_asmc_gain_update}. These equations have the same
structure as the canonical adaptive-input and gain-update laws in
Eqs.~\eqref{eq:canonical_adaptive_input}
and~\eqref{eq:canonical_gain_update}, respectively. Assume:
\begin{enumerate}[label=\textup{(A\arabic*)}]
  \item $K(0)\in(0,U_a^*)$ for some $U_a^*>\Gamma_a$ where $\Gamma_a$ is an upper bound on $\boldsymbol{a}_n$ defined by $\boldsymbol{a}_n = A_1 \boldsymbol{q}_n +A_2 \dot{\boldsymbol{q}}_n$.
  \item $\|\dot{\hat{\boldsymbol{u}}}_{\mathrm{asmc}}^{\mathrm{ad}}(t)\|\leq U_a^*$
        for all $t\geq0$, reflecting actuator-rate saturation at the
        implementation level.
  \item $\hat{\boldsymbol{e}}(0)$ and $\dot{\hat{\boldsymbol{e}}}(0)$
        are finite.
\end{enumerate}

Then the estimated tracking variables satisfy
\begin{align}
  \limsup_{t\to\infty}
  \|\hat{\boldsymbol{\sigma}}_2(t)\|
  &\leq
  \delta,
  \label{eq:asmc_sigma2_bound_theorem}\\
  \limsup_{t\to\infty}
  \|\hat{\boldsymbol{\sigma}}_1(t)\|_\infty
  &\leq
  \left(
    \frac{\delta}{C_2}
  \right)^{1/\nu_{\mathrm{asmc}}},
  \label{eq:asmc_sigma1_bound_theorem}\\
  \limsup_{t\to\infty}
  \|\hat{\boldsymbol{e}}(t)\|_\infty
  &\leq
  \frac{1}{C_1}
  \left(
    \frac{\delta}{C_2}
  \right)^{1/\nu_{\mathrm{asmc}}},
  \label{eq:asmc_ehat_bound_theorem}\\
  \limsup_{t\to\infty}
  \|\dot{\hat{\boldsymbol{e}}}(t)\|_\infty
  &\leq
  2
  \left(
    \frac{\delta}{C_2}
  \right)^{1/\nu_{\mathrm{asmc}}}.
  \label{eq:asmc_edothat_bound_theorem}
\end{align}
Moreover, since $\boldsymbol{e}=\hat{\boldsymbol{e}}+\tilde{\boldsymbol{q}}$
and $\dot{\boldsymbol{e}}=\dot{\hat{\boldsymbol{e}}}+\dot{\tilde{\boldsymbol{q}}}$,
the actual tracking errors satisfy
\begin{align}
  \limsup_{t\to\infty}
  \|\boldsymbol{e}(t)\|_\infty
  &\leq
  \frac{1}{C_1}
  \left(
    \frac{\delta}{C_2}
  \right)^{1/\nu_{\mathrm{asmc}}}
  +
  \frac{1}{\lambda_1}
  \left(
    \frac{\epsilon}{\lambda_2}
  \right)^{1/\nu_{\mathrm{smdo}}},
  \label{eq:asmc_actual_e_bound_theorem}\\
  \limsup_{t\to\infty}
  \|\dot{\boldsymbol{e}}(t)\|_\infty
  &\leq
  2
  \left(
    \frac{\delta}{C_2}
  \right)^{1/\nu_{\mathrm{asmc}}}
  +
  2
  \left(
    \frac{\epsilon}{\lambda_2}
  \right)^{1/\nu_{\mathrm{smdo}}}.
  \label{eq:asmc_actual_edot_bound_theorem}
\end{align}
\end{theorem}

\begin{proof}
The auxiliary system dynamics and the feedforward structure of
Eq.~\eqref{eq:auxiliary_tracking_dynamics} yield
\begin{equation}
  \ddot{\hat{\boldsymbol{e}}}
  =
  \hat{\boldsymbol{a}}-\boldsymbol{a}_n
  +
  \hat{\boldsymbol{u}}_{\mathrm{asmc}}.
  \label{eq:estimated_error_acceleration_2}
\end{equation}
Substituting the decomposition
$\hat{\boldsymbol{u}}_{\mathrm{asmc}}
=\hat{\boldsymbol{u}}_{\mathrm{asmc}}^{\mathrm{can}}+\hat{\boldsymbol{u}}_{\mathrm{asmc}}^{\mathrm{ad}}$
with Eq.~\eqref{eq:second_order_asmc_cancellation_component}
into the definition of $\hat{\boldsymbol{\sigma}}_2$
and using the exactness of the ASMC cancellation
established in Section~\ref{subsec:second_order_asmc}, we obtain
\begin{equation}
  \hat{\boldsymbol{\sigma}}_2
  =
  -\boldsymbol{a}_n
  +
  \hat{\boldsymbol{u}}_{\mathrm{asmc}}^{\mathrm{ad}}.
  \label{eq:estimated_sigma2_reduced}
\end{equation}
Differentiating
Eq.~\eqref{eq:estimated_sigma2_reduced} gives
\begin{equation}
  \dot{\hat{\boldsymbol{\sigma}}}_2
  =
  \boldsymbol{\xi}_{\mathrm{asmc}}
  + \dot{\hat{\boldsymbol{u}}}_{\mathrm{asmc}}^{\mathrm{ad}},
  \label{eq:estimated_sigma2_canonical}
\end{equation}
where $\boldsymbol{\xi}_{\mathrm{asmc}}=-\dot{\boldsymbol{a}}_n$. Because the nominal trajectory $\boldsymbol{q}_n$ and $\dot{\boldsymbol{q}}_n$ are generated analytically over a finite maneuver interval, they are clearly bounded as can be seen from Eq.~\eqref{eq:qn_ref}. Hence $\boldsymbol{a}_n$ is also bounded. Moreover,
\begin{equation}
  \dot{\boldsymbol{a}}_n
  =
  A_1\dot{\boldsymbol{q}}_n
  +
  A_2\ddot{\boldsymbol{q}}_n
  =
  A_1\dot{\boldsymbol{q}}_n
  +
  A_2
  \left(
    \boldsymbol{a}_n+\boldsymbol{u}_r
  \right),
  \label{eq:cw_nominal_acceleration_derivative_bounded}
\end{equation}
which shows that $\dot{\boldsymbol{a}}_n$ is bounded as well since $\boldsymbol{u}_r$ is bounded according to Eq.~\eqref{eq:ur}.
Consequently, $\boldsymbol{\xi}_{\mathrm{asmc}}=-\dot{\boldsymbol{a}}_n$ is uniformly bounded and there therefore exists $\Gamma_a>0$ such that
\begin{equation}
  \|\boldsymbol{\xi}_{\mathrm{asmc}}(t)\|
  \leq
  \Gamma_a,
  \qquad
  \forall\,t\geq0.
  \label{eq:xi_asmc_bound}
\end{equation}
Equation~\eqref{eq:estimated_sigma2_canonical} matches the canonical
form~\eqref{eq:z2_dynamics_generic} under the correspondence
\[
  (\boldsymbol{x},\,
  \boldsymbol{z}_1,\,
  \boldsymbol{z}_2,\,
  \alpha,\,\beta,\,
  P,\,P^*,\,\Delta,\,\kappa,\,\Gamma,\,U^*,\,\nu)
  \leftrightarrow
  (\hat{\boldsymbol{e}},\,
  \hat{\boldsymbol{\sigma}}_1,\,
  \hat{\boldsymbol{\sigma}}_2,\,
  C_1,\,C_2,\,
  K,\,K^*,\,\delta,\,\gamma,\,\Gamma_a,\,U_a^*,\,\nu_{asmc}).
\]
Assumptions~(A1)--(A3) verify all hypotheses of
Lemma~\ref{lem:canonical_adaptive_boundary_layer}, and
Eqs.~\eqref{eq:asmc_sigma2_bound_theorem}--\eqref{eq:asmc_edothat_bound_theorem}
follow from
Eqs.~\eqref{eq:canonical_z2_limsup}--\eqref{eq:canonical_xdot_limsup}. 
The actual tracking errors satisfy
$\boldsymbol{e}=\hat{\boldsymbol{e}}+\tilde{\boldsymbol{q}}$
and
$\dot{\boldsymbol{e}}=\dot{\hat{\boldsymbol{e}}}+\dot{\tilde{\boldsymbol{q}}}$,
as shown in Eq.~\eqref{eq:actual_estimated_error_relation} below.
\begin{equation}
  \boldsymbol{e}
  =
  \boldsymbol{q}-\boldsymbol{q}_n
  =
  \hat{\boldsymbol{e}}+\tilde{\boldsymbol{q}},
  \qquad
  \dot{\boldsymbol{e}}
  =
  \dot{\hat{\boldsymbol{e}}}+\dot{\tilde{\boldsymbol{q}}}.
  \label{eq:actual_estimated_error_relation}
\end{equation}
Applying the triangle inequality to
Eqs.~\eqref{eq:asmc_ehat_bound_theorem},
\eqref{eq:asmc_edothat_bound_theorem},
\eqref{eq:smdo_qtilde_bound_theorem}, and
\eqref{eq:smdo_qdottilde_bound_theorem}
gives
Eqs.~\eqref{eq:asmc_actual_e_bound_theorem} and
\eqref{eq:asmc_actual_edot_bound_theorem}. It is noted that the present analysis does not claim a full composite Lyapunov proof for the entire observer-controller interconnection. Instead, it establishes practical ultimate boundedness by exploiting the triangular structure induced by the auxiliary-state formulation and by combining the SMDO and ASMC bounds through $\boldsymbol{e} = \hat{\boldsymbol{e}} + \tilde{\boldsymbol{q}}$.
\end{proof}

The SMDO and ASMC layers are not treated as two independently applied controllers acting on unrelated systems. Their interconnection has a triangular structure induced by the auxiliary-state formulation. Since the ASMC command is applied to both the actual dynamics and the auxiliary dynamics, it cancels from the observer-error dynamics when the two systems are subtracted. Therefore, the SMDO subsystem regulates the observer error ($\tilde{\boldsymbol{q}}$), whereas the ASMC subsystem regulates the estimated tracking error ($\hat{\boldsymbol{e}}$). The actual tracking error is then obtained through the algebraic decomposition
\[
\boldsymbol{e}=\hat{\boldsymbol{e}}+\tilde{\boldsymbol{q}},\qquad \dot{\boldsymbol{e}}=\dot{\hat{\boldsymbol{e}}}+\dot{\tilde{\boldsymbol{q}}}.
\]
Thus, the final tracking-error bound is obtained by combining the ultimate bounds of the two subsystems. This should be interpreted as a practical boundedness result for a triangular observer-controller interconnection, rather than as a full composite Lyapunov proof for a strongly coupled cascade system. If residual modeling errors, imperfect cancellation, or implementation mismatch introduce additional coupling terms, these terms can be included in the effective bounded perturbations of the SMDO and ASMC surface dynamics. In that case, the same analysis applies with enlarged perturbation bounds and consequently enlarged practical ultimate bounds.

\subsection*{Surrogate Derivatives for the Second-Order SMDO and ASMC Surfaces}
\label{subsec:surrogate_surface_derivatives}

Evaluation of $\hat{\boldsymbol{s}}_2$ and $\hat{\boldsymbol{\sigma}}_2$ requires the
derivatives of their respective first-order surfaces,
$\dot{\hat{\boldsymbol{s}}}_1$ and
$\dot{\hat{\boldsymbol{\sigma}}}_1$. Direct evaluation of these derivatives
requires the acceleration-level quantities
$\ddot{\boldsymbol{q}}$ and
$\ddot{\hat{\boldsymbol{q}}}$. Within the second-order branch, however, these
accelerations depend on the SMDO and ASMC commands being evaluated, whereas
the commands themselves depend on
$\hat{\boldsymbol{s}}_2$ and
$\hat{\boldsymbol{\sigma}}_2$. Computing the required accelerations from the
same branch would therefore create a circular dependency between
the surface derivatives, the control commands, and the resulting closed-loop
accelerations.

To avoid this implicit loop, a parallel first-order branch is introduced.
It uses the same plant, auxiliary-system, and tracking formulations as the second-order branch, but employs first-order control laws that do not require derivatives of second-order sliding surfaces. Consequently, the corresponding acceleration-level signals of second-order sliding surfaces can be generated explicitly, while the velocity-level terms are evaluated from the available measured/estimated velocities and the analytic reference velocity. The superscript $(1)$ is used below to distinguish signals supplied by this first-order branch.

The plant and auxiliary-system dynamics of the first-order branch are governed by
\begin{equation}
\ddot{\boldsymbol{q}}^{(1)}
=
\boldsymbol{a}^{(1)}
+
\boldsymbol{u}_r
+
\boldsymbol{u}_{\mathrm{asmc}}^{(1)}
+
\boldsymbol{u}_{\mathrm{smdo}}^{(1)}
+
\boldsymbol{d}^{(1)},
\label{eq:first-order_branch}
\end{equation}
\begin{equation}
\ddot{\hat{\boldsymbol{q}}}^{(1)}
=
\hat{\boldsymbol{a}}^{(1)}
+
\boldsymbol{u}_r
+
\boldsymbol{u}_{\mathrm{asmc}}^{(1)},
\label{eq:first-order_aux_branch}
\end{equation}
where the feedforward input $\boldsymbol{u}_r$ is pre-computed from the given
initial conditions and is identical to that used in the second-order branch.
This branch is not an additional physical plant and its commands are not applied to the actual closed-loop system. The term $\boldsymbol{d}^{(1)}$ denotes an bounded branch residual used only for surrogate-signal generation. It may represent a known perturbation model or a filtered disturbance estimate. It is not assumed to be identical to the actual lumped disturbance acting on the second-order closed-loop branch. Any acceleration-level discrepancy between the surrogate branch and the implemented second-order branch, including the effects of branch-state mismatch, command mismatch, navigation-error-induced implementation errors, sampling effects, and unmodeled residuals, is collected into the surrogate mismatch terms $\boldsymbol{\omega}_s$ and $\boldsymbol{\omega}_{\sigma}$, as established in Corollary~\ref{cor:surrogate_derivative_implementation}.

The first-order sliding surfaces associated with the SMDO observation error and the ASMC tracking error are defined, respectively, as
\begin{align}
\boldsymbol{s}_{1}^{(1)}
&=
\dot{\boldsymbol{q}}^{(1)}
-
\dot{\hat{\boldsymbol{q}}}^{(1)}
+
\lambda_{1}
\!\left(
\boldsymbol{q}^{(1)}
-
\hat{\boldsymbol{q}}^{(1)}
\right).
\label{eq:s1_first_order}
\\[4pt]
\boldsymbol{\sigma}_{1}^{(1)}
&=
\dot{\boldsymbol{q}}^{(1)}
-
\dot{\boldsymbol{q}}_{n}
+
C_{1}
\!\left(
\boldsymbol{q}^{(1)}
-
\boldsymbol{q}_{n}
\right),
\label{eq:sigma1_first_order}
\end{align}
The corresponding first-order control laws are given as
\begin{align}
\boldsymbol{u}_{\mathrm{smdo}}^{(1)}
&=
-\frac{L^{(1)}+L^{*}}{\epsilon}\,
\boldsymbol{s}_{1}^{(1)},
\label{eq:u_smdo_first_order}
\\[4pt]
\boldsymbol{u}_{\mathrm{asmc}}^{(1)}
&=
-\frac{K^{(1)}+K^{*}}{\delta}\,
\boldsymbol{\sigma}_{1}^{(1)},
\label{eq:u_asmc_first_order}
\end{align}
with adaptive gain update rules
\begin{align}
\dot{L}^{(1)}(t)
&=
\eta
\!\left[
L^{(1)}(t)
\!\left(
\frac{\bigl\|\boldsymbol{s}_{1}^{(1)}(t)\bigr\|}{\epsilon}
-1
\right)
+
\frac{L^{*}}{\epsilon}
\bigl\|\boldsymbol{s}_{1}^{(1)}(t)\bigr\|
\right],
\qquad L^{(1)}(0)>0.
\label{eq:L_first_order}
\\[4pt]
\dot{K}^{(1)}(t)
&=
\gamma
\!\left[
K^{(1)}(t)
\!\left(
\frac{\bigl\|\boldsymbol{\sigma}_{1}^{(1)}(t)\bigr\|}{\delta}
-1
\right)
+
\frac{K^{*}}{\delta}
\bigl\|\boldsymbol{\sigma}_{1}^{(1)}(t)\bigr\|
\right],
\qquad K^{(1)}(0)>0,
\label{eq:K_first_order}
\end{align}

For the SMDO, the surrogate derivative of the first-order observer surface is constructed as
\begin{equation}
\dot{\hat{\boldsymbol{s}}}_{1,\mathrm{sur}}
=
\left(
\ddot{\boldsymbol{q}}^{(1)}
-
\ddot{\hat{\boldsymbol{q}}}^{(1)}
\right)
+
\lambda_1
\left(
\dot{\boldsymbol{q}}
-
\dot{\hat{\boldsymbol{q}}}
\right).
\label{eq:smdo_surrogate_derivative}
\end{equation}
Similarly, the surrogate derivative used to evaluate the second-order ASMC surface is
\begin{equation}
\dot{\hat{\boldsymbol{\sigma}}}_{1,\mathrm{sur}}
=
\left(
\ddot{\hat{\boldsymbol{q}}}^{(1)}
-
\ddot{\boldsymbol{q}}_n
\right)
+
C_1
\left(
\dot{\hat{\boldsymbol{q}}}
-
\dot{\boldsymbol{q}}_n
\right).
\label{eq:asmc_surrogate_derivative}
\end{equation}

The first-order branch is used only to provide dynamically consistent
surrogate signals for evaluating the second-order surfaces. Its control
commands are not applied to the actual closed-loop system and do not replace
the proposed SMDO and ASMC commands. This parallel construction
therefore preserves the proposed sliding surface-based
control architecture while eliminating
the algebraic loop that would arise if the required acceleration-level
signals were computed within the same branch.

The proposed implementation should be distinguished from a Levant-type exact differentiator (\cite{levant1998robust}). The main stability analysis is derived for the exact second-order surfaces, whereas the implemented controller evaluates these surfaces using surrogate derivatives generated by the parallel first-order branch. Therefore, the analysis does not require the finite-time exactness conditions of Levant’s differentiator. If a Levant-type differentiator is used as an alternative implementation, then the usual regularity assumptions must be imposed on the differentiated signals, such as a known bound on the Lipschitz constant of their derivatives and bounded measurement noise. In contrast, the proposed surrogate-branch implementation is characterized by the bounded mismatch between the exact and surrogate derivatives. This mismatch is treated as a bounded acceleration-level approximation error and may result in enlarged practical ultimate bounds. A quantitative comparison with finite differencing and Levant's differentiator is provided in~\ref{sec:comparison}.

\begin{corollary}[Boundedness under surrogate-derivative implementation]
\label{cor:surrogate_derivative_implementation}

Theorems~\ref{thm:smdo_residual_boundedness}
and~\ref{thm:asmc_tracking_boundedness} are derived using the exact
derivatives $\dot{\hat{\boldsymbol{s}}}_1$ and
$\dot{\hat{\boldsymbol{\sigma}}}_1$. In the proposed
architecture, direct evaluation of these derivatives requires the
acceleration-level quantities $\ddot{\boldsymbol{q}}$ and
$\ddot{\hat{\boldsymbol{q}}}$, which depend on the SMDO and
ASMC commands; those commands in turn depend on
$\hat{\boldsymbol{s}}_2$ and $\hat{\boldsymbol{\sigma}}_2$, creating an
implicit loop. To avoid this dependency, the required acceleration-level
quantities are supplied by the parallel first-order branch, yielding the surrogate
derivatives $\dot{\hat{\boldsymbol{s}}}_{1,\mathrm{sur}}$ and
$\dot{\hat{\boldsymbol{\sigma}}}_{1,\mathrm{sur}}$.

Define the surrogate-derivative mismatches as
\begin{align}
  \boldsymbol{\omega}_s
  &:=
  \dot{\hat{\boldsymbol{s}}}_{1,\mathrm{sur}}
  -
  \dot{\hat{\boldsymbol{s}}}_1
  =
  \ddot{\tilde{\boldsymbol{q}}}^{(1)}
  -
  \ddot{\tilde{\boldsymbol{q}}},
  \label{eq:omega_s_def}\\
  \boldsymbol{\omega}_\sigma
  &:=
  \dot{\hat{\boldsymbol{\sigma}}}_{1,\mathrm{sur}}
  -
  \dot{\hat{\boldsymbol{\sigma}}}_1
  =
  \ddot{\hat{\boldsymbol{q}}}^{(1)}
  -
  \ddot{\hat{\boldsymbol{q}}},
  \label{eq:omega_sigma_def}
\end{align}
where $\ddot{\tilde{\boldsymbol{q}}}^{(1)} = \ddot{\boldsymbol{q}}^{(1)} - \ddot{\hat{\boldsymbol{q}}}^{(1)}$ and the superscript $(1)$ denotes signals generated by the first-order surrogate branch.

\textit{Boundedness of $\boldsymbol{\omega}_s$ and $\boldsymbol{\omega}_{\sigma}$.}
From Eq.~\eqref{eq:qtilde_reduced_dynamics} in the proof of
Theorem~\ref{thm:smdo_residual_boundedness},
\[
  \ddot{\tilde{\boldsymbol{q}}}
  =
  -\lambda_1\dot{\tilde{\boldsymbol{q}}}
  -
  \lambda_2
  \operatorname{sgn}(\hat{\boldsymbol{s}}_1)
  |\hat{\boldsymbol{s}}_1|^{\nu_{\mathrm{smdo}}}
  +
  \tilde{\boldsymbol{d}}_{\mathrm{smdo}}.
\]
Since $\dot{\tilde{\boldsymbol{q}}}$, $\hat{\boldsymbol{s}}_1$, and
$\tilde{\boldsymbol{d}}_{\mathrm{smdo}}$ are all bounded by
Theorem~\ref{thm:smdo_residual_boundedness}, it follows that
$\ddot{\tilde{\boldsymbol{q}}}$ is bounded. Also, from Eqs.~\eqref{eq:estimated_error_acceleration},~\eqref{eq:second_order_asmc_cancellation_component}, and~\eqref{eq:estimated_sigma2_reduced}, we obtain
\[
  \ddot{\hat{\boldsymbol{q}}}
  =
  \ddot{\boldsymbol{q}}_n
  +
  \hat{\boldsymbol{\sigma}}_2
  -
  C_1\dot{\hat{\boldsymbol{e}}}
  -
  C_2
  \operatorname{sgn}(\hat{\boldsymbol{\sigma}}_1)
  |\hat{\boldsymbol{\sigma}}_1|^{\nu_{\mathrm{asmc}}}.
\]
Since $\hat{\boldsymbol{\sigma}}_2$, $\dot{\hat{\boldsymbol{e}}}$, and
$\hat{\boldsymbol{\sigma}}_1$ are all bounded by
Theorem~\ref{thm:asmc_tracking_boundedness}, and $\ddot{\boldsymbol{q}}_n$
is bounded as established in the same proof, it follows that
$\ddot{\hat{\boldsymbol{q}}}$ is bounded.

As a result, the second-order branch
($\ddot{\tilde{\boldsymbol{q}}},\ddot{\hat{\boldsymbol{q}}}$) is bounded under
Theorems~\ref{thm:smdo_residual_boundedness}
and~\ref{thm:asmc_tracking_boundedness}.
For the first-order surrogate branch, the control inputs are algebraic,
given by Eqs.~\eqref{eq:u_smdo_first_order}--\eqref{eq:u_asmc_first_order}, and are assumed to satisfy
$\|\boldsymbol{u}_{\mathrm{smdo}}^{(1)}(t)\|\leq U_S^*$ and
$\|\boldsymbol{u}_{\mathrm{asmc}}^{(1)}(t)\|\leq U_A^*$ for all $t\geq0$,
reflecting magnitude saturation at the implementation level.
Under these admissibility conditions, the dynamics of
$\boldsymbol{\sigma}_1^{(1)}$ and $\boldsymbol{s}_1^{(1)}$ take the form
\begin{equation*}
  \dot{\boldsymbol{\sigma}}_1^{(1)} = \boldsymbol{\xi}_{\mathrm{asmc}}^{(1)}(t) + \boldsymbol{u}_{\mathrm{asmc}}^{(1)},
  \qquad
  \dot{\boldsymbol{s}}_1^{(1)} = \boldsymbol{\xi}_{\mathrm{smdo}}^{(1)}(t) + \boldsymbol{u}_{\mathrm{smdo}}^{(1)},
\end{equation*}
where $\boldsymbol{\xi}_{\mathrm{asmc}}^{(1)}$ and $\boldsymbol{\xi}_{\mathrm{smdo}}^{(1)}$
collect the remaining bounded terms.
Although $\boldsymbol{\sigma}_1^{(1)}$ and $\boldsymbol{s}_1^{(1)}$ are
first-order surfaces and do not share the nonlinear structure of
$\boldsymbol{z}_2$ in Eq.~\eqref{eq:canonical_z2}, their closed-loop
dynamics have the same additive form as
Eq.~\eqref{eq:z2_dynamics_generic}, with
$\boldsymbol{u}_{\mathrm{asmc}}^{(1)}$ and $\boldsymbol{u}_{\mathrm{smdo}}^{(1)}$
entering directly in place of $\dot{\boldsymbol{u}}$.
The gain update laws~\eqref{eq:K_first_order}--\eqref{eq:L_first_order}
and the control laws~\eqref{eq:u_asmc_first_order}--\eqref{eq:u_smdo_first_order}
match the canonical forms~\eqref{eq:canonical_gain_update}
and~\eqref{eq:canonical_adaptive_input} under the substitutions
$(P,P^*,\Delta,\kappa)\leftarrow(K^{(1)},K^*,\delta,\gamma)$
and
$(P,P^*,\Delta,\kappa)\leftarrow(L^{(1)},L^*,\epsilon,\eta)$,
respectively.
The same Lyapunov argument used to establish
$\limsup_{t\to\infty}\|\boldsymbol{z}_2(t)\|\leq\Delta$
in the proof of Lemma~\ref{lem:canonical_adaptive_boundary_layer}
therefore applies, and yields boundedness of
$\boldsymbol{\sigma}_1^{(1)}$ and $\boldsymbol{s}_1^{(1)}$.
Boundedness of $\ddot{\tilde{\boldsymbol{q}}}^{(1)}$ and
$\ddot{\hat{\boldsymbol{q}}}^{(1)}$ then follows directly from
Eqs.~\eqref{eq:first-order_branch}--\eqref{eq:first-order_aux_branch}
together with the bounded inputs and bounded remaining terms.

Hence the triangle inequality applied
to Eqs.~\eqref{eq:omega_s_def}--\eqref{eq:omega_sigma_def} yields
finite nonnegative constants $\bar{\omega}_s$ and $\bar{\omega}_\sigma$
such that
\begin{equation}
  \|\boldsymbol{\omega}_s(t)\|
  \leq
  \bar{\omega}_s,
  \qquad
  \|\boldsymbol{\omega}_\sigma(t)\|
  \leq
  \bar{\omega}_\sigma,
  \qquad
  \forall\,t\geq0.
  \label{eq:omega_bounds}
\end{equation}
The implemented second-order surfaces are
\begin{align}
  \hat{\boldsymbol{s}}_{2,\mathrm{imp}}
  &=
  \dot{\hat{\boldsymbol{s}}}_{1,\mathrm{sur}}
  +
  \lambda_2
  \operatorname{sgn}(\hat{\boldsymbol{s}}_1)
  |\hat{\boldsymbol{s}}_1|^{\nu_{\mathrm{smdo}}}
  =
  \hat{\boldsymbol{s}}_2
  +
  \boldsymbol{\omega}_s,
  \label{eq:s2_imp}\\
  \hat{\boldsymbol{\sigma}}_{2,\mathrm{imp}}
  &=
  \dot{\hat{\boldsymbol{\sigma}}}_{1,\mathrm{sur}}
  +
  C_2
  \operatorname{sgn}(\hat{\boldsymbol{\sigma}}_1)
  |\hat{\boldsymbol{\sigma}}_1|^{\nu_{\mathrm{asmc}}}
  =
  \hat{\boldsymbol{\sigma}}_2
  +
  \boldsymbol{\omega}_\sigma.
  \label{eq:sigma2_imp}
\end{align}
Applying the triangle inequality together with
Eqs.~\eqref{eq:smdo_s2_bound_theorem}, \eqref{eq:asmc_sigma2_bound_theorem},
and~\eqref{eq:omega_bounds},
\begin{align}
  \limsup_{t\to\infty}
  \|\hat{\boldsymbol{s}}_{2,\mathrm{imp}}(t)\|
  &\leq
  \epsilon + \bar{\omega}_s,
  \label{eq:s2imp_bound}\\
  \limsup_{t\to\infty}
  \|\hat{\boldsymbol{\sigma}}_{2,\mathrm{imp}}(t)\|
  &\leq
  \delta + \bar{\omega}_\sigma.
  \label{eq:sigma2imp_bound}
\end{align}
The implemented second-order surfaces remain bounded consistently with the exact-surface analysis of Theorems~\ref{thm:smdo_residual_boundedness}
and~\ref{thm:asmc_tracking_boundedness}, with enlarged practical bounds $\bar{\omega}_s$ and $\bar{\omega}_\sigma$. The downstream bounds on
$\hat{\boldsymbol{s}}_1$, $\tilde{\boldsymbol{q}}$,
$\hat{\boldsymbol{\sigma}}_1$, $\hat{\boldsymbol{e}}$, and their
time-derivatives are modified by replacing $\epsilon$ and $\delta$ with
$\epsilon+\bar{\omega}_s$ and $\delta+\bar{\omega}_\sigma$, respectively.
\end{corollary}

Corollary~\ref{cor:surrogate_derivative_implementation} confirms that the surrogate-derivative implementation does not compromise the stability conclusions of the analysis. Because the mismatches $\boldsymbol{\omega}_s$ and $\boldsymbol{\omega}_\sigma$ are bounded acceleration-level differences, the implemented second-order surfaces remain bounded, with the boundary-layer parameters $\epsilon$ and $\delta$ enlarged by the finite margins $\bar{\omega}_s$ and $\bar{\omega}_\sigma$.

\section{Implementation Summary of the Proposed Algorithm}
\label{sec:implementation_summary}

The proposed method is implemented as a two-phase algorithm that combines
analytic PCO insertion with the integrated SMDO--ASMC feedback
structure. The implementation procedure is summarized as follows.

\begin{enumerate}
\item \textit{Nominal trajectory design.}\;
      Specify the initial relative state $\boldsymbol{\xi}_0$, transfer time
      $t_f$, PCO radius $\rho$, and target mean motion $n$. Solve the quartic
      stationarity equation given in Eq.~\eqref{eq:phase_quartic_polynomial} to obtain the optimal PCO entry phase
      $\phi_0^\star$, and generate the nominal reference trajectory
      $\boldsymbol{q}_n(t)$, $\dot{\boldsymbol{q}}_n(t)$,
      $\ddot{\boldsymbol{q}}_n(t)$ using Eq.~\eqref{eq:qn_ref} and the feedforward input $\boldsymbol{u}_r(t)$ using Eq.~\eqref{eq:ur} for the Reaching Phase.

\item \textit{Initialization.}\;
      Set the auxiliary state $\hat{\boldsymbol{q}}(0)=\boldsymbol{q}(0)$, and select the initial
      adaptive gains $L(0)>0$ and $K(0)>0$.

\item \textit{SMDO update.}\;
      At each time step, propagate the auxiliary system
      Eq.~\eqref{eq:auxiliary_normalized_dynamics} and evaluate the observer
      surfaces $\hat{\boldsymbol{s}}_1$ and $\hat{\boldsymbol{s}}_2$.
      Update the adaptive gain $L(t)$ using Eq.~\eqref{eq:smdo_adaptive_gain_update} and integrate
      $\dot{\hat{\boldsymbol{u}}}_{\mathrm{smdo}}^{\mathrm{ad}}$ to obtain the total SMDO compensation
      $\hat{\boldsymbol{u}}_{\mathrm{smdo}}
      =\hat{\boldsymbol{u}}_{\mathrm{smdo}}^{\mathrm{can}}
      +\hat{\boldsymbol{u}}_{\mathrm{smdo}}^{\mathrm{ad}}$
      using Eq.~\eqref{eq:smdo_explicit}.

\item \textit{ASMC update.}\;
      Evaluate the estimated tracking surface $\hat{\boldsymbol{\sigma}}_1$
      using the reference-signal surrogate derivative, form
      $\hat{\boldsymbol{\sigma}}_2$, and update the adaptive gain $K(t)$ by Eq.~\eqref{eq:second_order_asmc_gain_update}
      to compute the total ASMC tracking command 
      $\hat{\boldsymbol{u}}_{\mathrm{asmc}}
      =\hat{\boldsymbol{u}}_{\mathrm{asmc}}^{\mathrm{can}}
      +\hat{\boldsymbol{u}}_{\mathrm{asmc}}^{\mathrm{ad}}$
      using Eq.~\eqref{eq:asmc_explicit}.

\item \textit{Control input application.}\;
      Apply the total input according to the active mission phase:
      \begin{equation}
        \boldsymbol{u}(t)
        =
        \begin{cases}
          \boldsymbol{u}_r(t)
          +\hat{\boldsymbol{u}}_{\mathrm{asmc}}(t)
          +\hat{\boldsymbol{u}}_{\mathrm{smdo}}(t),
          & 0\leq t\leq t_f,\\[4pt]
          \hat{\boldsymbol{u}}_{\mathrm{asmc}}(t)
          +\hat{\boldsymbol{u}}_{\mathrm{smdo}}(t),
          & t>t_f.
        \end{cases}
      \end{equation}
      Repeat steps 3--5 at each time step until the end of the mission.
\end{enumerate}

\section{Numerical Simulation Results}
\label{sec:numerical_results}

This section presents the numerical evaluation of the proposed optimal PCO-entry phase selection and the integrated SMDO--ASMC control framework. The simulation environment and disturbance model are described first, followed by validation of the analytic PCO-entry solution. The closed-loop behavior is then examined in terms of the chaser trajectory, control inputs, observer estimation, tracking-error boundedness, and adaptive gain evolution. Finally, an additional simulation under additive navigation measurement noise is presented to assess the robustness of the output-feedback structure against navigation uncertainty and to examine the degree to which measurement contamination is isolated within the observer layer.

\subsection{Simulation Setup}
\label{subsec:simulation_setup}

The total applied control input is
\begin{equation}
    \boldsymbol{u}_{\mathrm{tot}}
    =
    \boldsymbol{u}_r
    +
    \hat{\boldsymbol{u}}_{\mathrm{smdo}}
    +
    \hat{\boldsymbol{u}}_{\mathrm{asmc}},
\end{equation}
where $\boldsymbol{u}_r$ is the reaching-phase input,
$\hat{\boldsymbol{u}}_{\mathrm{smdo}}$ is the SMDO disturbance-compensation
term, and $\hat{\boldsymbol{u}}_{\mathrm{asmc}}$ is the ASMC tracking term.
The actual relative-motion dynamics are
\begin{equation}
    \ddot{\boldsymbol{q}}
    =
    \boldsymbol{a}
    +
    \boldsymbol{u}_{\mathrm{tot}}
    +
    \boldsymbol{d},
    \label{eq:simulation_full_dynamics}
\end{equation}
where $\boldsymbol{a}$ is the nominal relative acceleration from
Section~\ref{sec:system_modeling_nominal_design} and $\boldsymbol{d}$
is an external disturbance acceleration injected into the plant but
withheld from both the controller and the observer. The disturbance
takes the form
\begin{equation}
    \boldsymbol{d}(t)
    =
    1.8\times10^{-4}
    \begin{bmatrix}
        1 - 1.5\sin(3nt) \\
        0.5\sin(2nt) \\
        \sin(3nt)
    \end{bmatrix}
    \,\mathrm{m/s^2},
    \label{eq:simulation_disturbance}
\end{equation}
where $n$ is the target mean motion. This disturbance is time-varying,
bounded, and oscillatory at frequencies that are integer multiples of
$n$. It is applied only to the actual spacecraft dynamics and is not
provided to the controller or observer; the SMDO must reconstruct its
effect solely from the measured and estimated state information.

The SMDO supplies the estimated states $\hat{\boldsymbol{q}}$ and
$\dot{\hat{\boldsymbol{q}}}$ to the ASMC tracking layer, so the
controller operates on observer output without direct access to the
true state. The controller and observer parameters are listed in
Table~\ref{tab:simulation_parameters}. The initial value of $\hat{\boldsymbol{u}}_{\mathrm{asmc}}^{ad}$ was chosen to ensure
$\hat{\boldsymbol{u}}_{\mathrm{asmc}}(0)=\boldsymbol{0}$. All simulations were performed
in MATLAB/Simulink with a fixed time step of $0.001\,\mathrm{s}$ over
two orbital periods of the target.

\begin{table}[H]
\centering
\caption{Controller and observer parameters used in the numerical
simulations.}
\label{tab:simulation_parameters}
\begin{tabular}{cc|cc}
\toprule
\multicolumn{2}{c|}{SMDO} & \multicolumn{2}{c}{ASMC} \\
\cmidrule(lr){1-2}\cmidrule(lr){3-4}
Parameter & Value & Parameter & Value \\
\midrule
$\epsilon$ & $1.0\times10^{-5}$
& $\delta$ & $0.01$ \\
$\lambda_1$ & $2$
& $C_1$ & $2$ \\
$\eta$ & $0.01$
& $\gamma$ & $0.001$ \\
$L(0)$ & $0.01$
& $K(0)$ & $0.2$ \\
$L^{*}$ & $0.01$
& $K^{*}$ & $0.1$ \\
$\lambda_2$ & $3$
& $C_2$ & $3$ \\
$\nu_{\mathrm{smdo}}$ & $8/11$
& $\nu_{\mathrm{asmc}}$ & $9/11$ \\
$\hat{\boldsymbol{u}}_{\mathrm{smdo}}^{ad}(0)$ & 
$\begin{bmatrix}
        0 & 0 & 0
    \end{bmatrix}^{\top}$
& $\hat{\boldsymbol{u}}_{\mathrm{asmc}}^{ad}(0)$ & 
$10^{-4}\times\begin{bmatrix}
        14.7465 & -2.2136 & -13.4749
    \end{bmatrix}^{\top}$ \\
\bottomrule
\end{tabular}
\end{table}

\subsection{Numerical Validation of the Optimal PCO-Entry Phase}
\label{subsec:optimal_phi0_validation}

The analytic phase-selection procedure of
Section~\ref{subsec:pco_entry_phase_selection} was validated against a
brute-force grid search under the initial conditions
\begin{equation}
    \boldsymbol{\xi}_0 =
    \begin{bmatrix}
        100 & 100 & 1100 & 0.1 & 0.5 & 0.8
    \end{bmatrix}^{\top}\,\mathrm{m,\,m/s},
\end{equation}
with transfer time $t_f = 3600\,\mathrm{s}$, PCO radius
$\rho = 1000\,\mathrm{m}$, and target mean motion
$n = 1.1068\times10^{-3}\,\mathrm{rad/s}$.

The real roots of the quartic equation in
Eq.~\eqref{eq:phase_quartic_polynomial} were computed analytically, and
each admissible candidate was evaluated against
\begin{equation}
    J(\phi_0)
    =
    \frac{1}{2}
    G(\phi_0)^{\top}
    C_{\Phi}^{\top}
    S_f^{-1}
    C_{\Phi}
    G(\phi_0),
    \qquad
    G(\phi_0)
    =
    \Phi_f^{-1}\boldsymbol{\xi}_f(\phi_0)
    -
    \Phi_0^{-1}\boldsymbol{\xi}_0.
\end{equation}
The analytic optimum and the corresponding transfer cost were obtained as
\begin{equation}
    \phi_{0,\mathrm{analytic}}^{\star}
    =
    49.2878^{\circ},
    \qquad
    J_{\mathrm{analytic}}^{\star}
    =
    4.2236\times10^{-4}\,\mathrm{m^2/s^3},
\end{equation}
while a grid search over $\phi_0 \in [0^{\circ}, 360^{\circ})$ at a
resolution of $0.01^{\circ}$ produced
\begin{equation}
    \phi_{0,\mathrm{grid}}^{\star}
    =
    49.2900^{\circ},
    \qquad
    J_{\mathrm{grid}}^{\star}
    =
    4.2236\times10^{-4}\,\mathrm{m^2/s^3}.
\end{equation}
The phase discrepancy of $2.25\times10^{-3}\,\mathrm{deg}$ is smaller
than the grid resolution, and the corresponding cost difference of
$7.4\times10^{-13}\,\mathrm{m^2/s^3}$ is negligible. The cost curve
$J(\phi_0)$ is shown in Fig.~\ref{fig:optimal_phi0_cost_curve}; the
analytic minimizer coincides with the numerical minimum to within the
resolution of the sweep.

\begin{figure}[H]
    \centering
    \includegraphics[
        width=0.45\linewidth,
        trim={0 310 0 310},
        clip
    ]{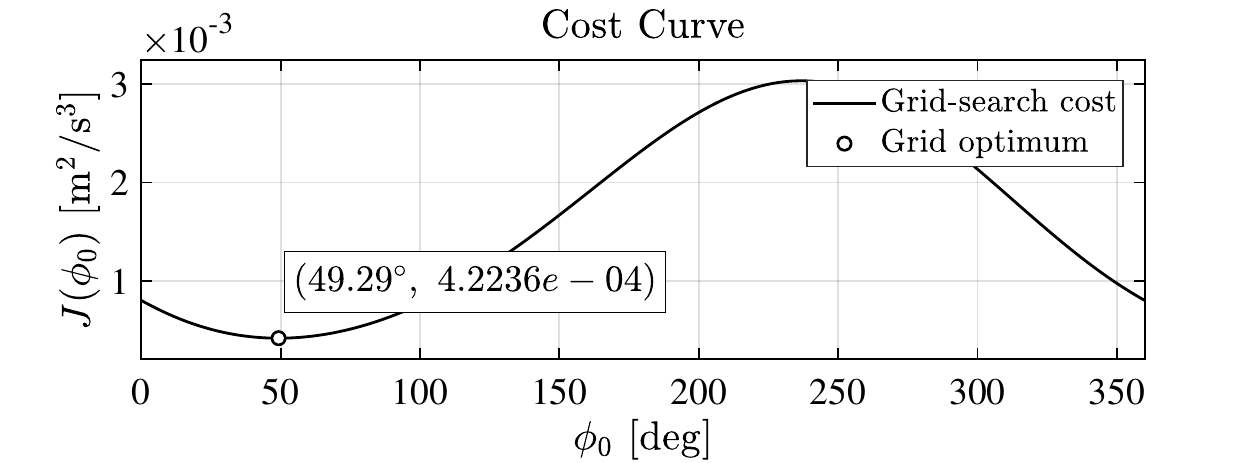}
    \caption{Transfer cost $J(\phi_0)$ as a function of the PCO entry
    phase.}
    \label{fig:optimal_phi0_cost_curve}
\end{figure}

As a terminal consistency check, the state produced by the optimal
reaching input at $t = t_f$ was compared with the desired PCO state:
\begin{equation}
    \left\|
        \boldsymbol{\xi}_n(t_f)
        -
        \boldsymbol{\xi}_f\!\left(\phi_{0,\mathrm{analytic}}^{\star}\right)
    \right\|
    =
    1.497\times10^{-11}.
\end{equation}
This residual establishes that the constructed nominal trajectory reaches
the intended PCO state at the prescribed terminal time. The resulting
trajectory is adopted as the tracking reference for all subsequent
closed-loop simulations.

\subsection{Baseline Closed-Loop Performance without Navigation Error}
\label{subsec:closed_loop_ideal}
Before examining the subsystem-level behavior in detail, the overall
closed-loop response is presented in terms of the chaser trajectory and
the three control input components. The subsequent analysis then proceeds
layer by layer, from the SMDO observer through the ASMC tracking
controller, following the boundedness hierarchy established in
Section~\ref{sec:boundedness_stability}.

\paragraph{Chaser trajectory and control inputs}
The trajectory of the chaser satellite in the relative frame is shown
in Fig.~\ref{fig:Trajectory_Chaser_Satellite}. The chaser is steered
from its initial offset toward the PCO and subsequently maintains the
formation geometry throughout the simulation. The transition between
the Reaching Phase and the PCO Phase is governed by the
interplay of the reaching input and the SMDO--ASMC feedback structure.

\begin{figure}[H]
    \centering
    \includegraphics[width=0.45\linewidth]{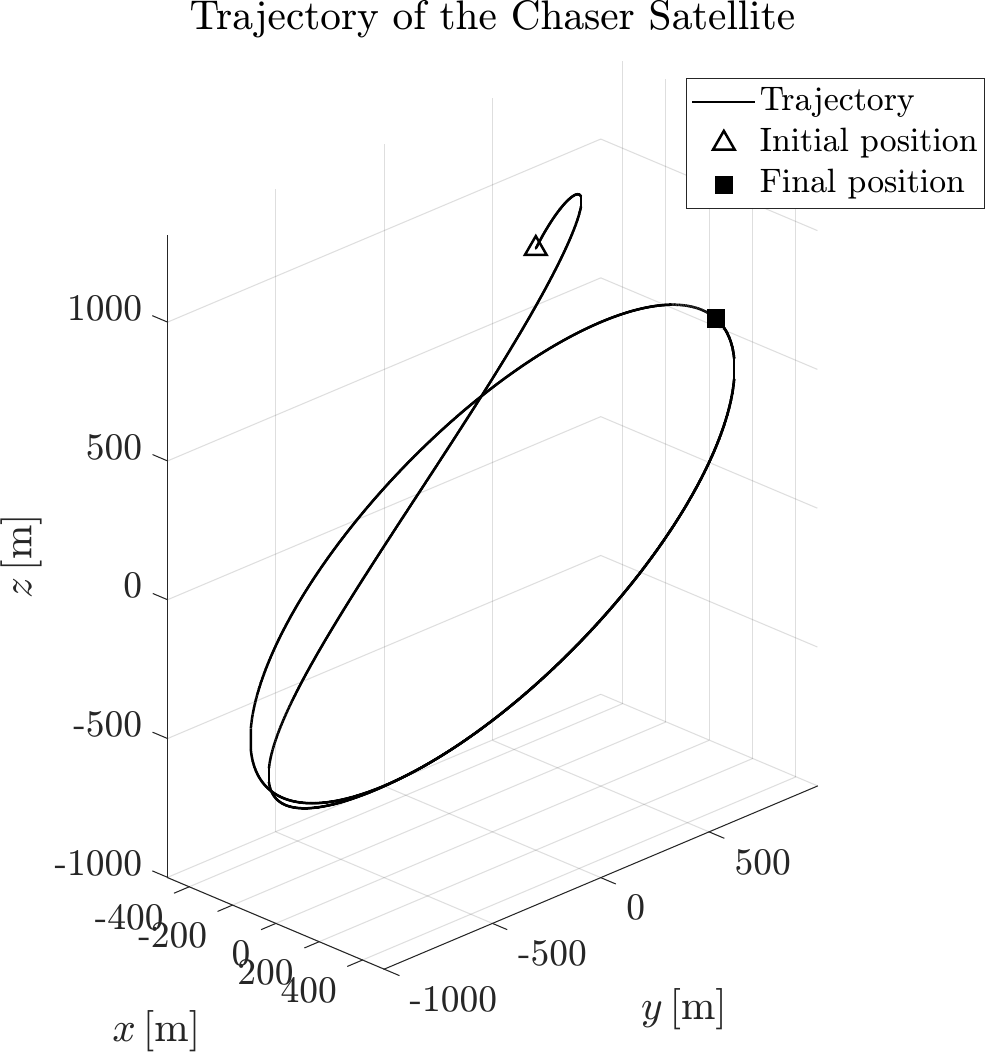}
    \caption{Trajectory of the chaser satellite in the relative frame.}
    \label{fig:Trajectory_Chaser_Satellite}
\end{figure}

The three control input components are shown in
Figs.~\ref{fig:reaching_input_second_order}--\ref{fig:u_asmc_second_order}.
The reaching input $\boldsymbol{u}_r$
(Fig.~\ref{fig:reaching_input_second_order}) is active during the
transfer interval and drives the nominal state toward the analytically
selected PCO terminal condition. Once the transfer is complete,
$\boldsymbol{u}_r$ becomes zero and the feedback layer assumes
full responsibility for maintaining the formation. The SMDO compensation
input $\hat{\boldsymbol{u}}_{\mathrm{smdo}}$
(Fig.~\ref{fig:u_smdo}) tracks the injected disturbance profile and
effectively reduces its effect on the plant before it reaches the tracking layer.
The ASMC tracking input $\hat{\boldsymbol{u}}_{\mathrm{asmc}}$
(Fig.~\ref{fig:u_asmc_second_order}) handles the residual that the
observer cannot reconstruct, keeping the estimated tracking error within
the prescribed boundary layer. All three components remain bounded
throughout the two-orbit simulation.

\begin{figure}[H]
    \centering
    \begin{subfigure}{0.32\linewidth}
        \centering
        \includegraphics[width=\linewidth]{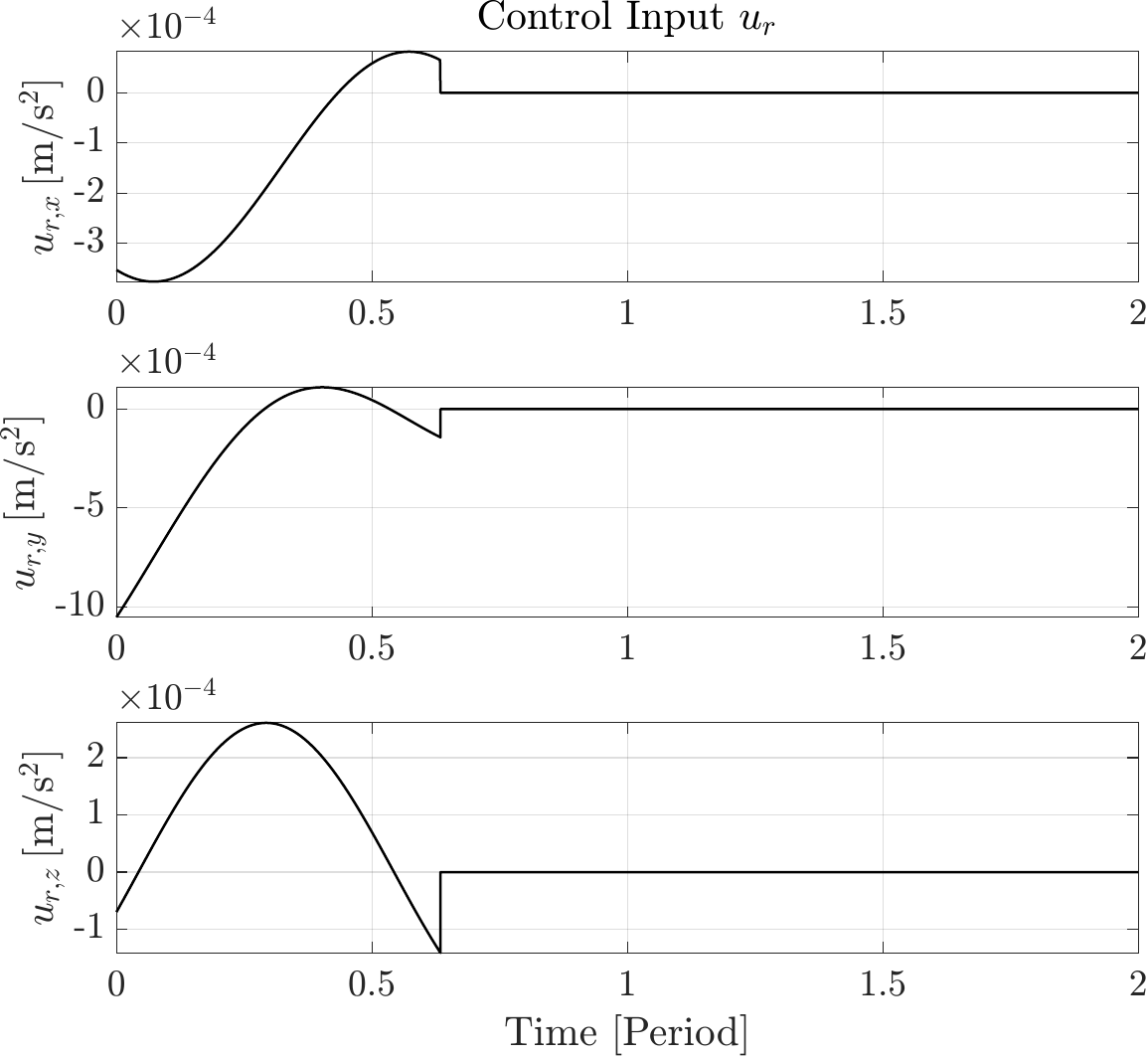}
        \caption{Reaching-phase input $\boldsymbol{u}_r$.}
        \label{fig:reaching_input_second_order}
    \end{subfigure}
    \hfill
    \begin{subfigure}{0.32\linewidth}
        \centering
        \includegraphics[width=\linewidth]{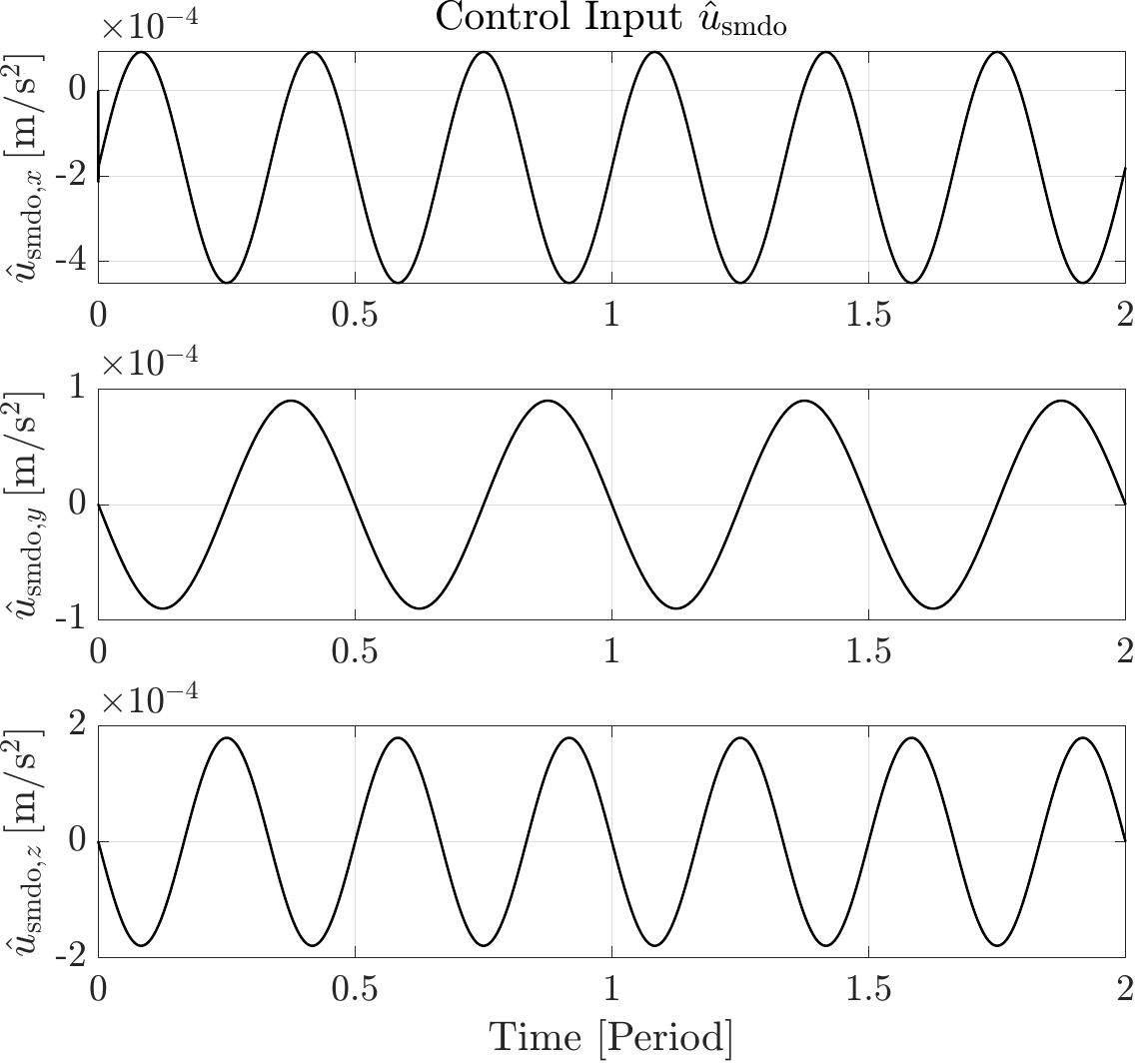}
        \caption{SMDO compensation input
        $\hat{\boldsymbol{u}}_{\mathrm{smdo}}$.}
        \label{fig:u_smdo}
    \end{subfigure}
    \hfill
    \begin{subfigure}{0.32\linewidth}
        \centering
        \includegraphics[width=\linewidth]{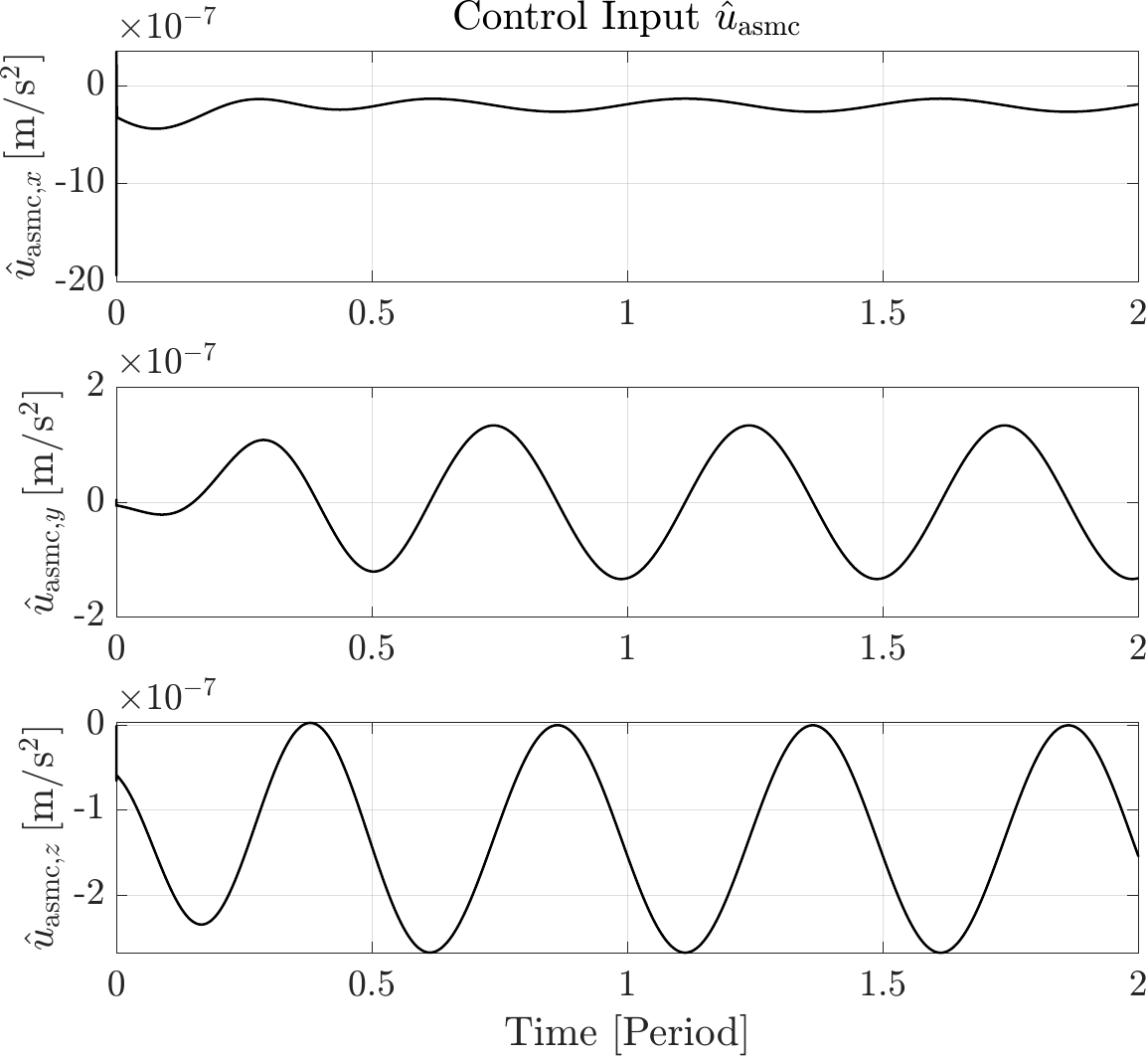}
        \caption{ASMC tracking input
        $\hat{\boldsymbol{u}}_{\mathrm{asmc}}$.}
        \label{fig:u_asmc_second_order}
    \end{subfigure}
    \caption{Control input components of the proposed second-order SMDO--ASMC structure.}
    \label{fig:control_input_components_second_order}
\end{figure}

The SMDO and ASMC inputs are further decomposed into their cancellation and
adaptive components in Figs.~\ref{fig:smdo_input_decomposition}
and~\ref{fig:asmc_input_decomposition}, respectively. For the SMDO,
the cancellation component
$\hat{\boldsymbol{u}}_{\mathrm{smdo}}^{\mathrm{can}}$
contains the known terms introduced to cancel the nominal part of the
observer-error dynamics, whereas the adaptive component
$\hat{\boldsymbol{u}}_{\mathrm{smdo}}^{\mathrm{ad}}$ is generated by the
adaptive law to compensate for the remaining unknown disturbance effect.
The bounded responses of both components indicate that the SMDO
compensation is achieved without unbounded growth in either the
cancellation term or the adaptive injection term.

\begin{figure}[H]
    \centering
    \begin{subfigure}{0.48\linewidth}
        \centering
        \includegraphics[width=\linewidth]{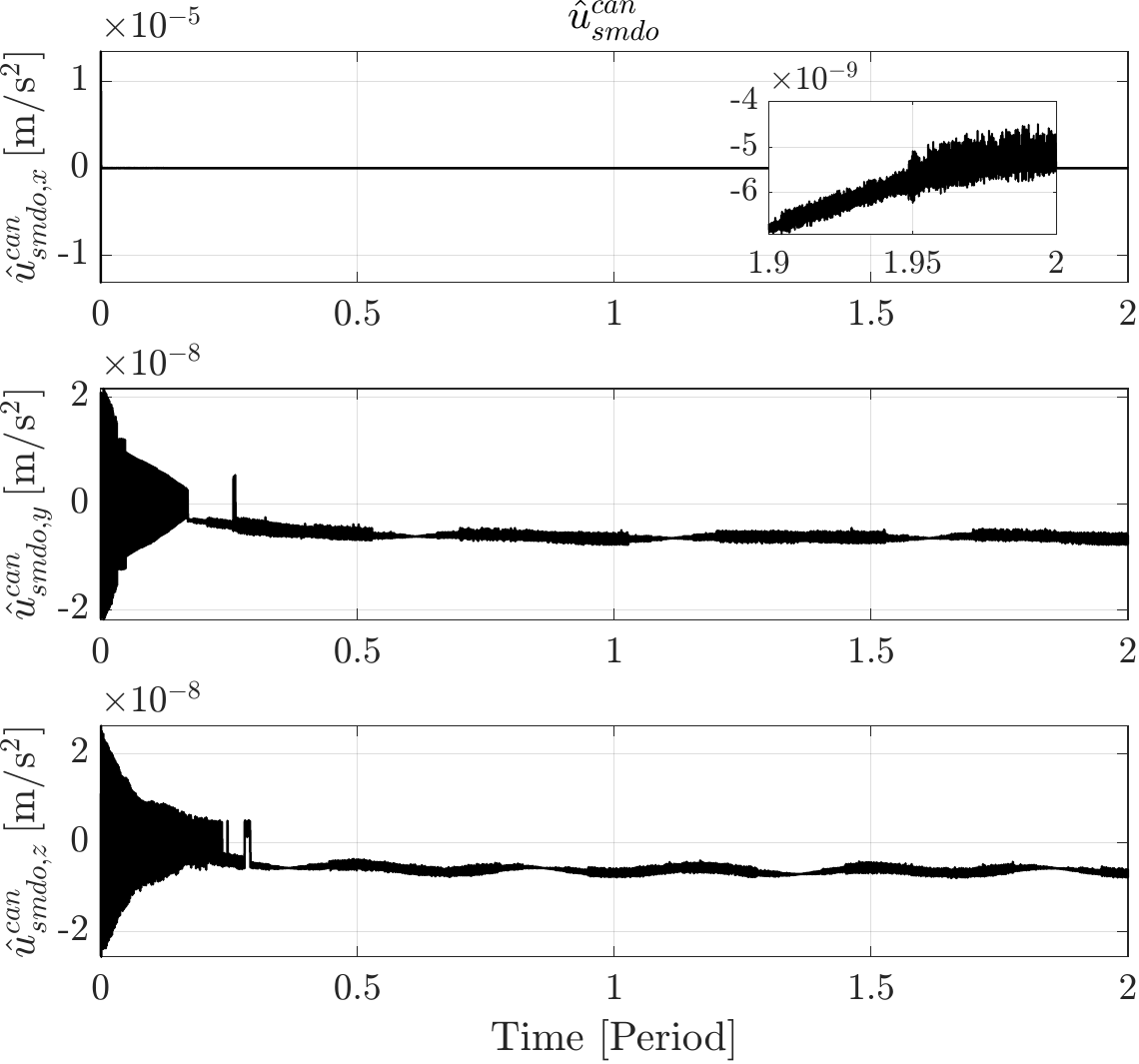}
        \caption{SMDO cancellation component
        $\hat{\boldsymbol{u}}_{\mathrm{smdo}}^{\mathrm{can}}$.}
        \label{fig:u_smdo_can}
    \end{subfigure}
    \hfill
    \begin{subfigure}{0.48\linewidth}
        \centering
        \includegraphics[width=\linewidth]{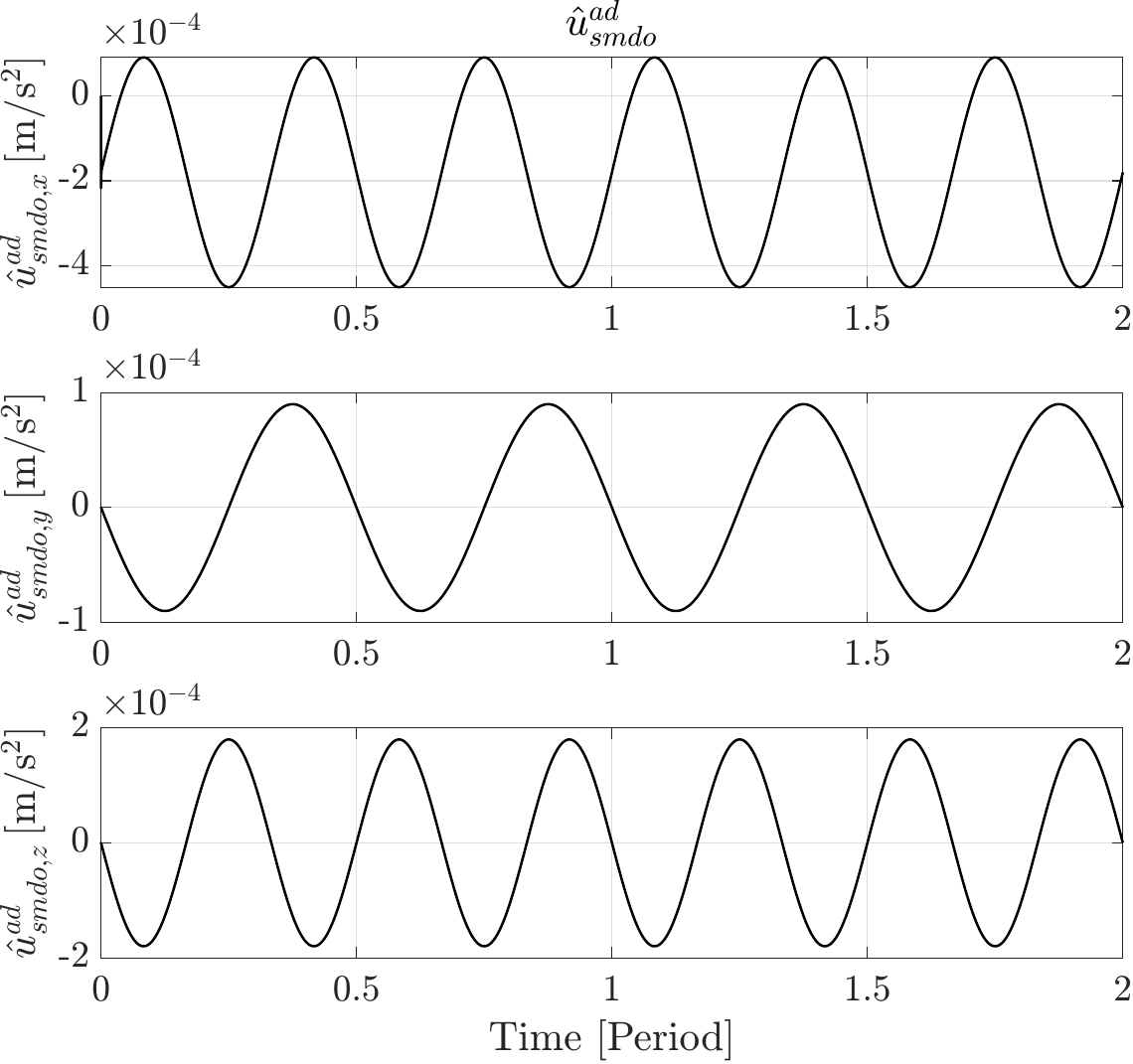}
        \caption{Adaptive SMDO component
        $\hat{\boldsymbol{u}}_{\mathrm{smdo}}^{\mathrm{ad}}$.}
        \label{fig:u_smdo_ad}
    \end{subfigure}
    \caption{Cancellation and adaptive components of the SMDO compensation input.}
    \label{fig:smdo_input_decomposition}
\end{figure}

For the ASMC input, the cancellation component
$\hat{\boldsymbol{u}}_{\mathrm{asmc}}^{\mathrm{can}}$ removes the known
surface-dependent terms in the estimated tracking-error dynamics, while
the adaptive component
$\hat{\boldsymbol{u}}_{\mathrm{asmc}}^{\mathrm{ad}}$ provides the remaining
robust correction against the lumped uncertainty in the reduced
second-order surface dynamics. The two components remain finite over the
entire simulation, which is consistent with the bounded behavior of the
ASMC surfaces and the estimated tracking error.

\begin{figure}[H]
    \centering
    \begin{subfigure}{0.48\linewidth}
        \centering
        \includegraphics[width=\linewidth]{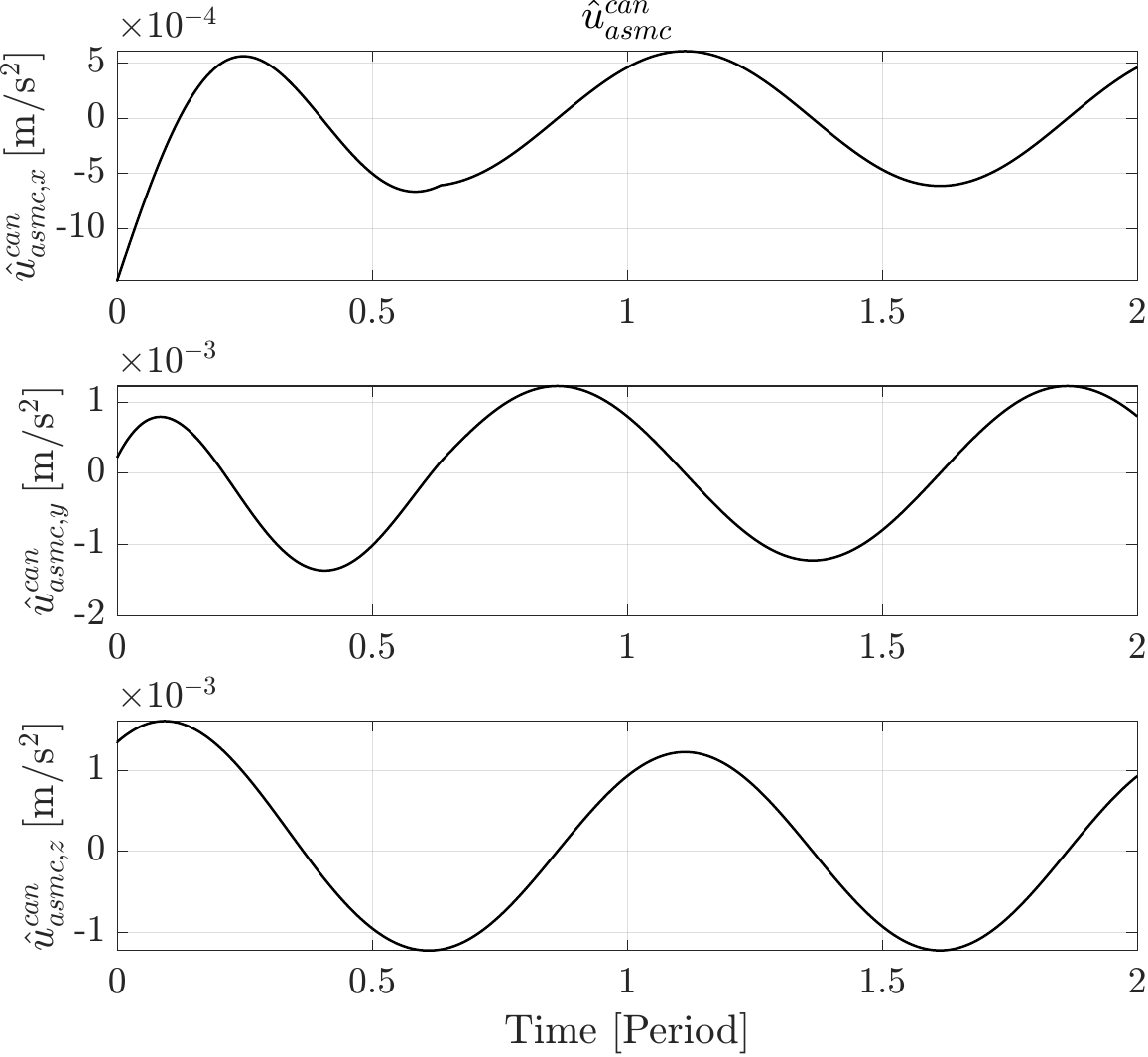}
        \caption{ASMC cancellation component
        $\hat{\boldsymbol{u}}_{\mathrm{asmc}}^{\mathrm{can}}$.}
        \label{fig:u_asmc_can}
    \end{subfigure}
    \hfill
    \begin{subfigure}{0.48\linewidth}
        \centering
        \includegraphics[width=\linewidth]{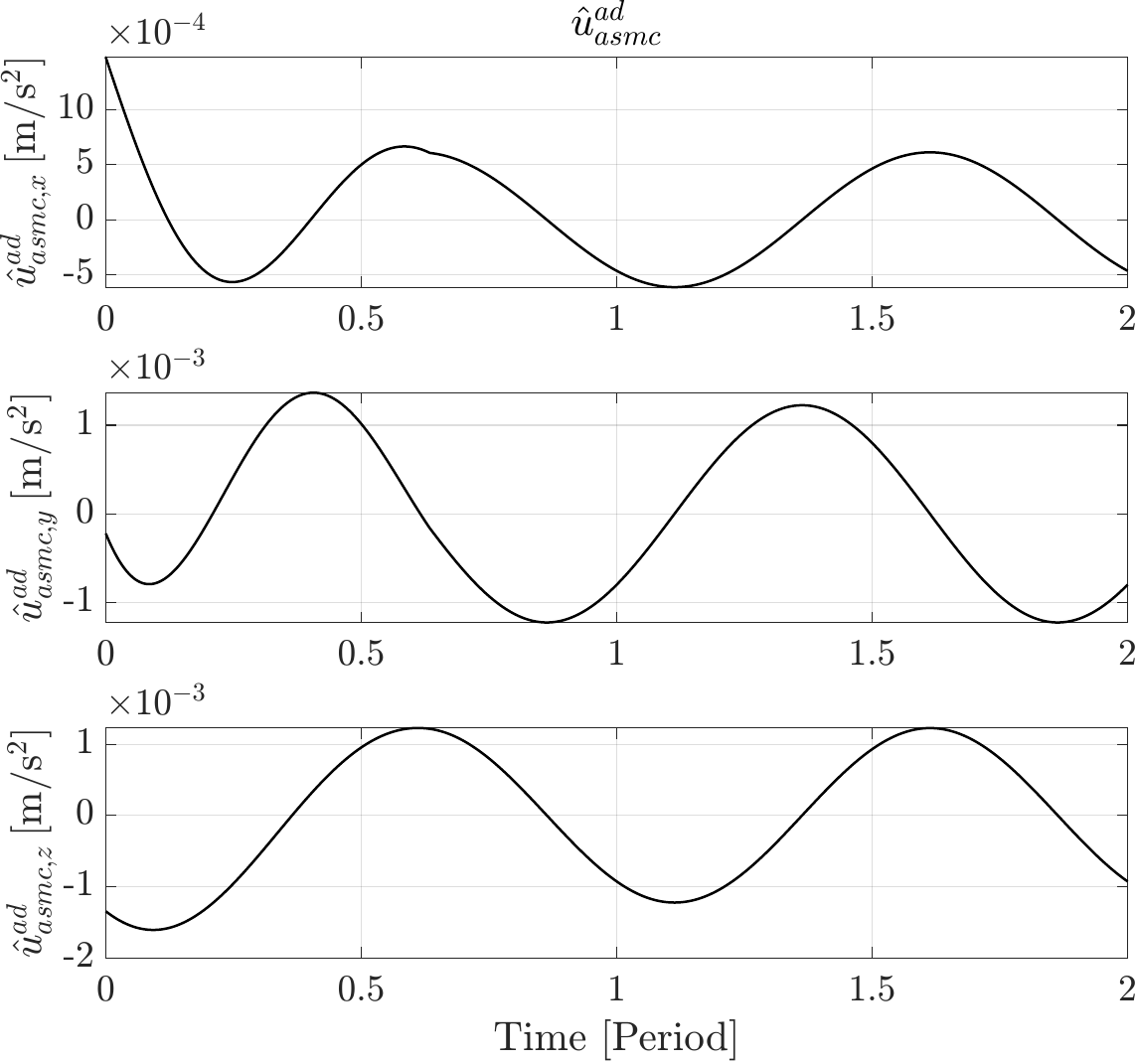}
        \caption{Adaptive ASMC component
        $\hat{\boldsymbol{u}}_{\mathrm{asmc}}^{\mathrm{ad}}$.}
        \label{fig:u_asmc_ad}
    \end{subfigure}
    \caption{Cancellation and adaptive components of the ASMC tracking input.}
    \label{fig:asmc_input_decomposition}
\end{figure}

With the control structure established, the observer and tracking layers are examined in turn.
\paragraph{SMDO Estimation Performance}

The observer layer is analyzed in terms of the SMDO sliding surfaces,
the resulting state-estimation errors, and the residual disturbance
after compensation. This sequence reflects the propagation of
boundedness established in Theorem~\ref{thm:smdo_residual_boundedness}:
once $\hat{\boldsymbol{s}}_2$ is confined within its boundary layer,
the downstream quantities $\hat{\boldsymbol{s}}_1$,
$\tilde{\boldsymbol{q}}$, $\dot{\tilde{\boldsymbol{q}}}$, and
$\tilde{\boldsymbol{d}}_{\mathrm{smdo}}$ are all ultimately bounded
as well.

With $\epsilon = 1.0\times10^{-5}$, $\lambda_1 = 2$,
$\lambda_2 = 3$, and $\nu_{\mathrm{smdo}} = 8/11$, the estimated ultimate bounds are
\begin{align}
  \limsup_{t\to\infty}\|\hat{\boldsymbol{s}}_2(t)\|
  &\leq
  1.0\times10^{-5}\ \mathrm{m/s^2},
  \label{eq:smdo_s2_numerical_bound}\\
  \limsup_{t\to\infty}\|\hat{\boldsymbol{s}}_1(t)\|_\infty
  &\leq
  2.9441\times10^{-8}\ \mathrm{m/s},
  \label{eq:smdo_s1_numerical_bound}\\
  \limsup_{t\to\infty}\|\tilde{\boldsymbol{q}}(t)\|_\infty
  &\leq
  1.4721\times10^{-8}\ \mathrm{m},
  \label{eq:smdo_qtilde_numerical_bound}\\
  \limsup_{t\to\infty}\|\dot{\tilde{\boldsymbol{q}}}(t)\|_\infty
  &\leq
  5.8883\times10^{-8}\ \mathrm{m/s},
  \label{eq:smdo_qdottilde_numerical_bound}\\
  \limsup_{t\to\infty}\|\tilde{\boldsymbol{d}}_{\mathrm{smdo}}(t)\|
  &\leq
  1.0\times10^{-5}\ \mathrm{m/s^2}.
  \label{eq:smdo_residual_numerical_bound}
\end{align}

Figure~\ref{fig:SMDO_surfaces} shows the two SMDO sliding surfaces.
The second-order surface $\hat{\boldsymbol{s}}_2$
(Fig.~\ref{fig:SMDO_s2}) is driven into the neighborhood determined
by $\epsilon$ after the initial transient, with all three components
remaining confined near the origin thereafter. The first-order surface
$\hat{\boldsymbol{s}}_1$ (Fig.~\ref{fig:SMDO_s1}), which represents
the filtered state-estimation error, inherits this boundedness: once
$\hat{\boldsymbol{s}}_2$ enters its layer, $\hat{\boldsymbol{s}}_1$
remains within the induced bound of
Eq.~\eqref{eq:smdo_s1_numerical_bound}.

\begin{figure}[H]
  \centering
  \begin{subfigure}{0.48\linewidth}
    \centering
    \includegraphics[width=\linewidth]{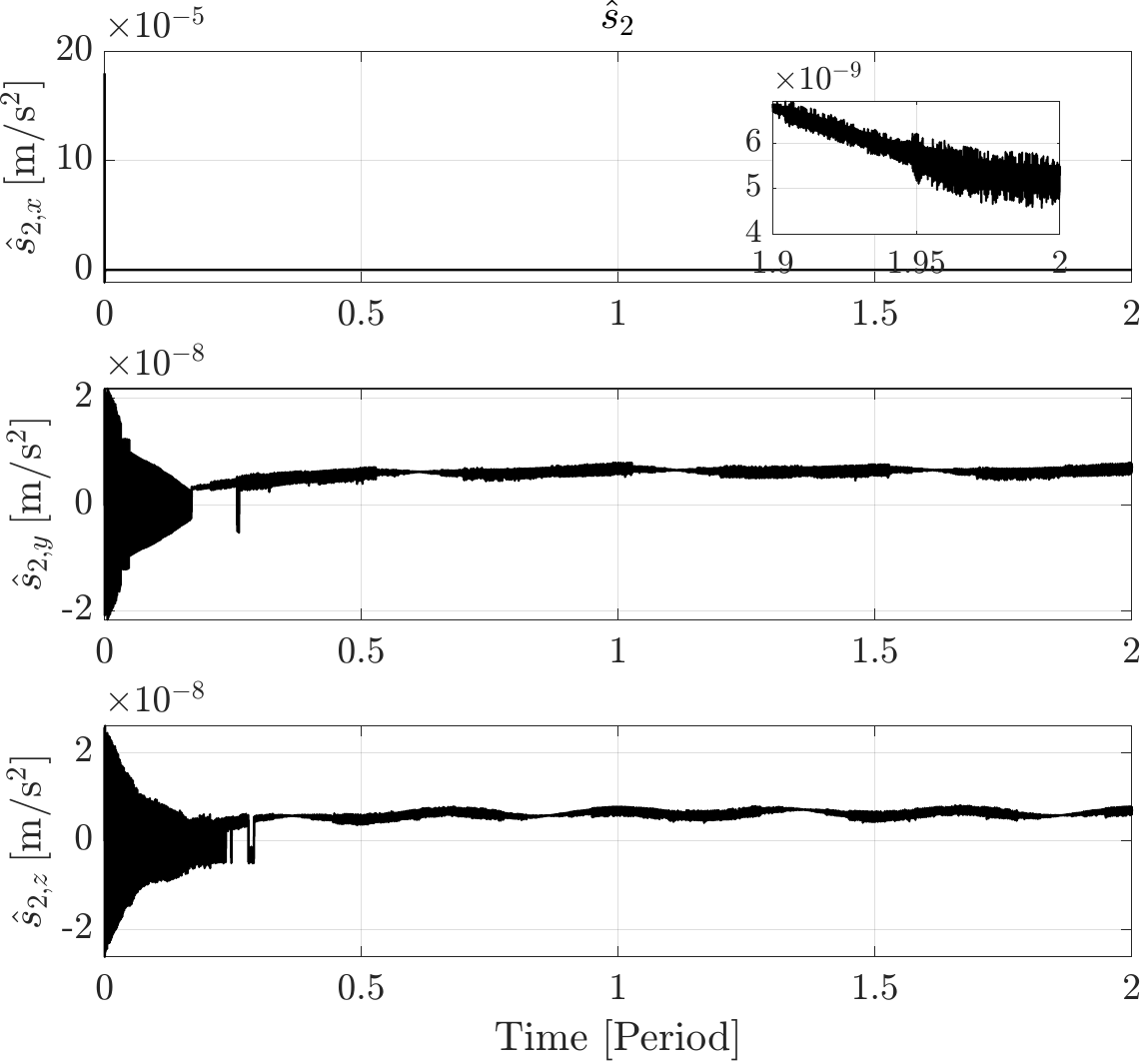}
    \caption{Second-order SMDO surface $\hat{\boldsymbol{s}}_2$.}
    \label{fig:SMDO_s2}
  \end{subfigure}
  \hfill
  \begin{subfigure}{0.48\linewidth}
    \centering
    \includegraphics[width=\linewidth]{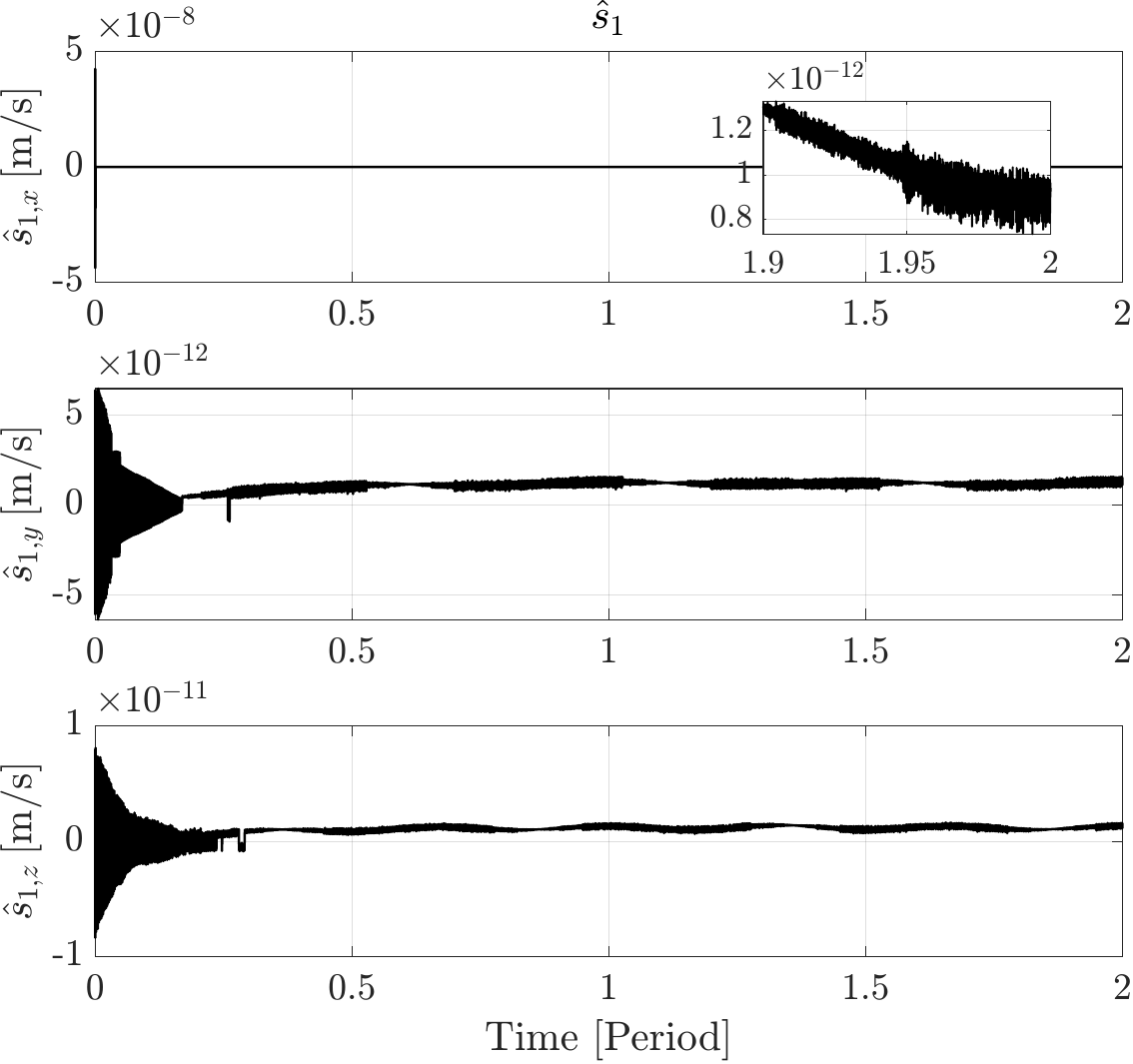}
    \caption{First-order SMDO surface $\hat{\boldsymbol{s}}_1$.}
    \label{fig:SMDO_s1}
  \end{subfigure}
  \caption{SMDO sliding surfaces.}
  \label{fig:SMDO_surfaces}
\end{figure}

The state-estimation errors are plotted in
Fig.~\ref{fig:state_estimation_errors}. The position error
$\tilde{\boldsymbol{q}} = \boldsymbol{q} - \hat{\boldsymbol{q}}$
(Fig.~\ref{fig:state_estimation_error_position}) settles within the
bound of Eq.~\eqref{eq:smdo_qtilde_numerical_bound}, and the velocity
error $\dot{\tilde{\boldsymbol{q}}}$
(Fig.~\ref{fig:state_estimation_error_velocity}) remains within the
bound of Eq.~\eqref{eq:smdo_qdottilde_numerical_bound}.

\begin{figure}[H]
  \centering
  \begin{subfigure}{0.48\linewidth}
    \centering
    \includegraphics[width=\linewidth]{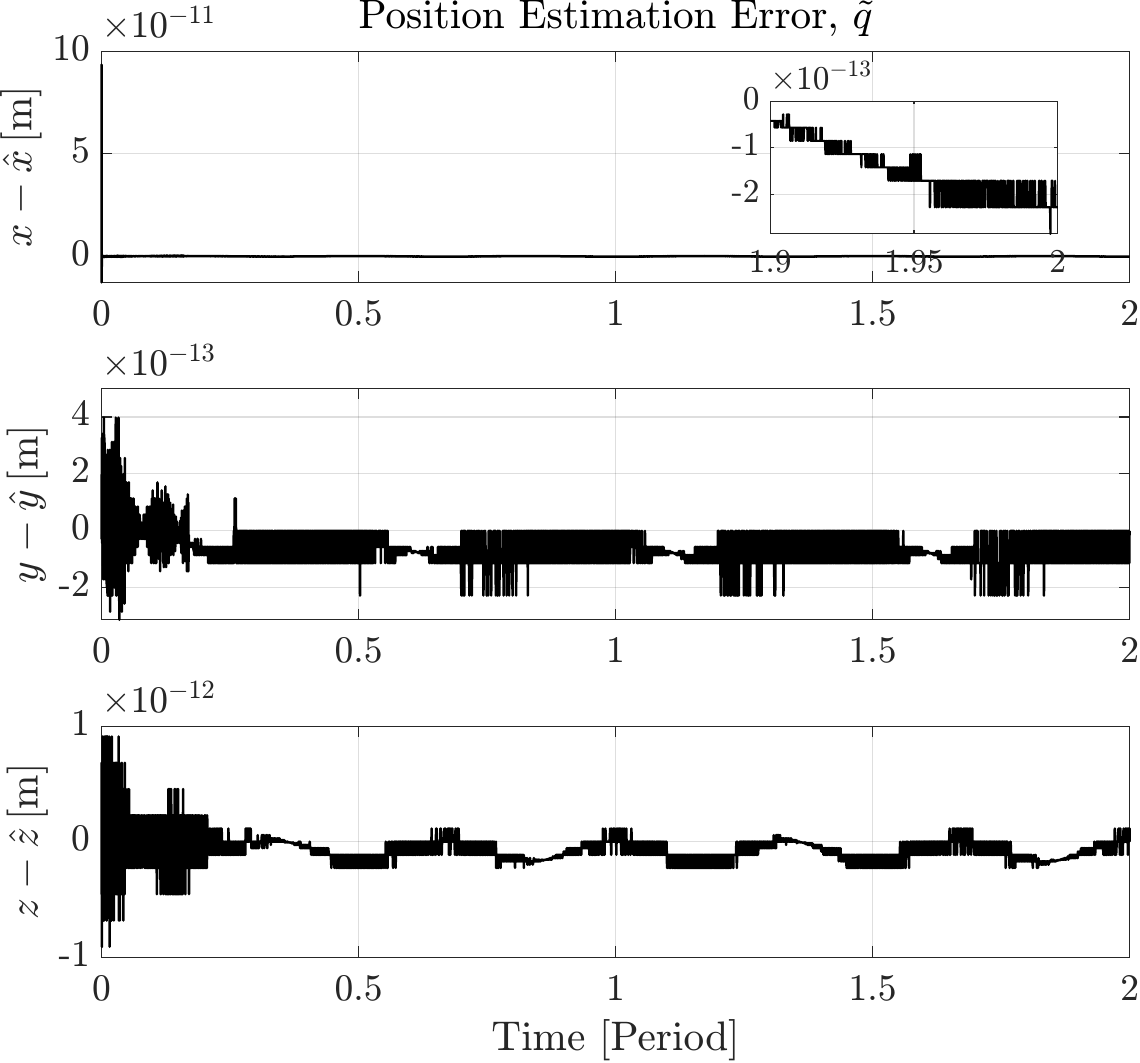}
    \caption{Position-estimation error
    $\tilde{\boldsymbol{q}} = \boldsymbol{q} - \hat{\boldsymbol{q}}$.}
    \label{fig:state_estimation_error_position}
  \end{subfigure}
  \hfill
  \begin{subfigure}{0.48\linewidth}
    \centering
    \includegraphics[width=\linewidth]{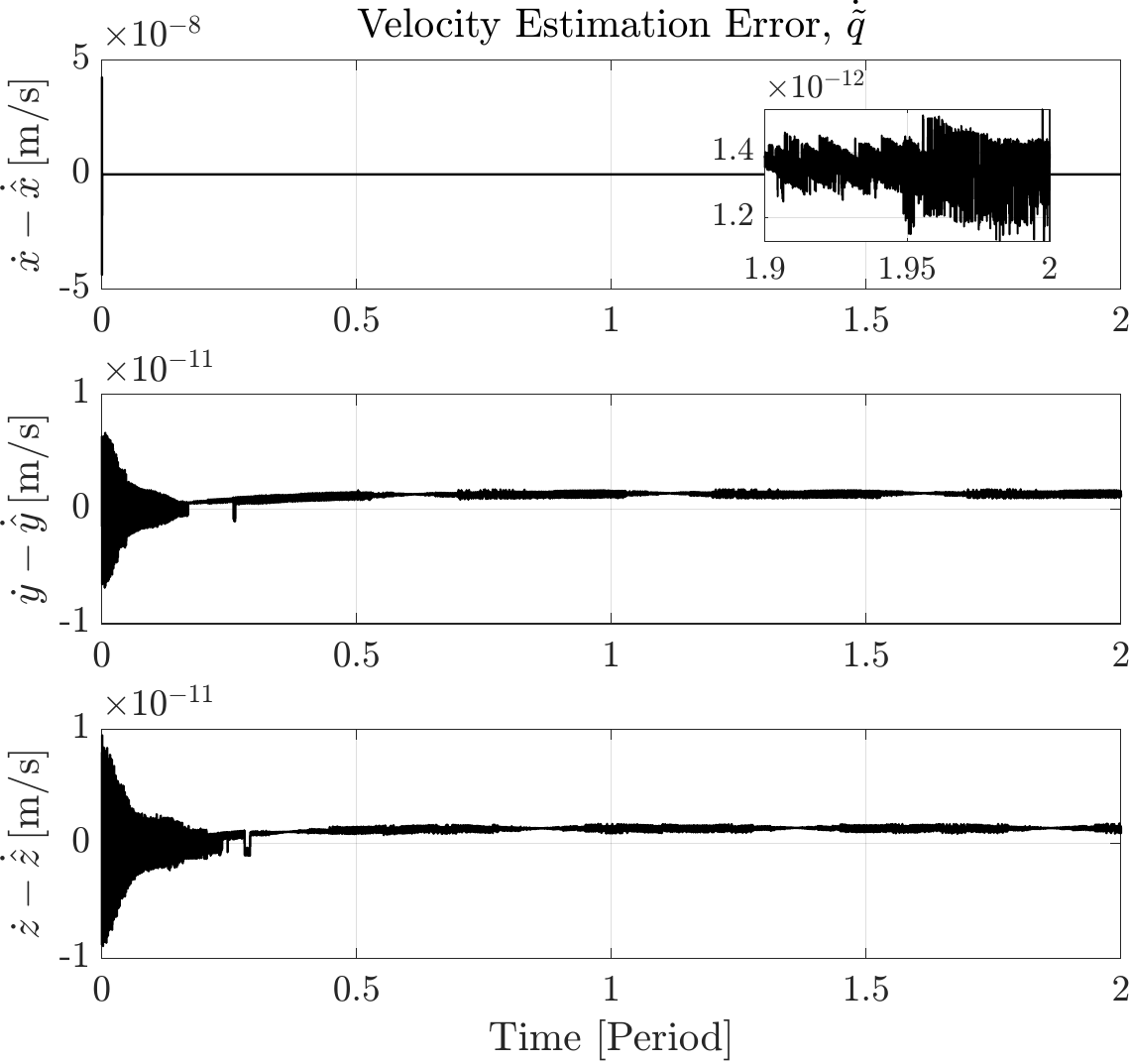}
    \caption{Velocity-estimation error
    $\dot{\tilde{\boldsymbol{q}}}
    = \dot{\boldsymbol{q}} - \dot{\hat{\boldsymbol{q}}}$.}
    \label{fig:state_estimation_error_velocity}
  \end{subfigure}
  \caption{State-estimation errors.}
  \label{fig:state_estimation_errors}
\end{figure}

The residual disturbance after SMDO compensation is shown in
Fig.~\ref{fig:Disturbance_Estimation_Error}. The compensation signal
$\hat{\boldsymbol{u}}_{\mathrm{smdo}}$, already presented in
Fig.~\ref{fig:u_smdo}, tracks the injected disturbance profile
throughout the simulation. The remaining residual
$\boldsymbol{d}+\hat{\boldsymbol{u}}_{\mathrm{smdo}}^{\mathrm{ad}}$ is reduced
to the neighborhood predicted by
Eq.~\eqref{eq:smdo_residual_numerical_bound} and is subsequently
handled by the ASMC layer as an unmodeled input.

\begin{figure}[H]
  \centering
  \includegraphics[width=0.45\linewidth]{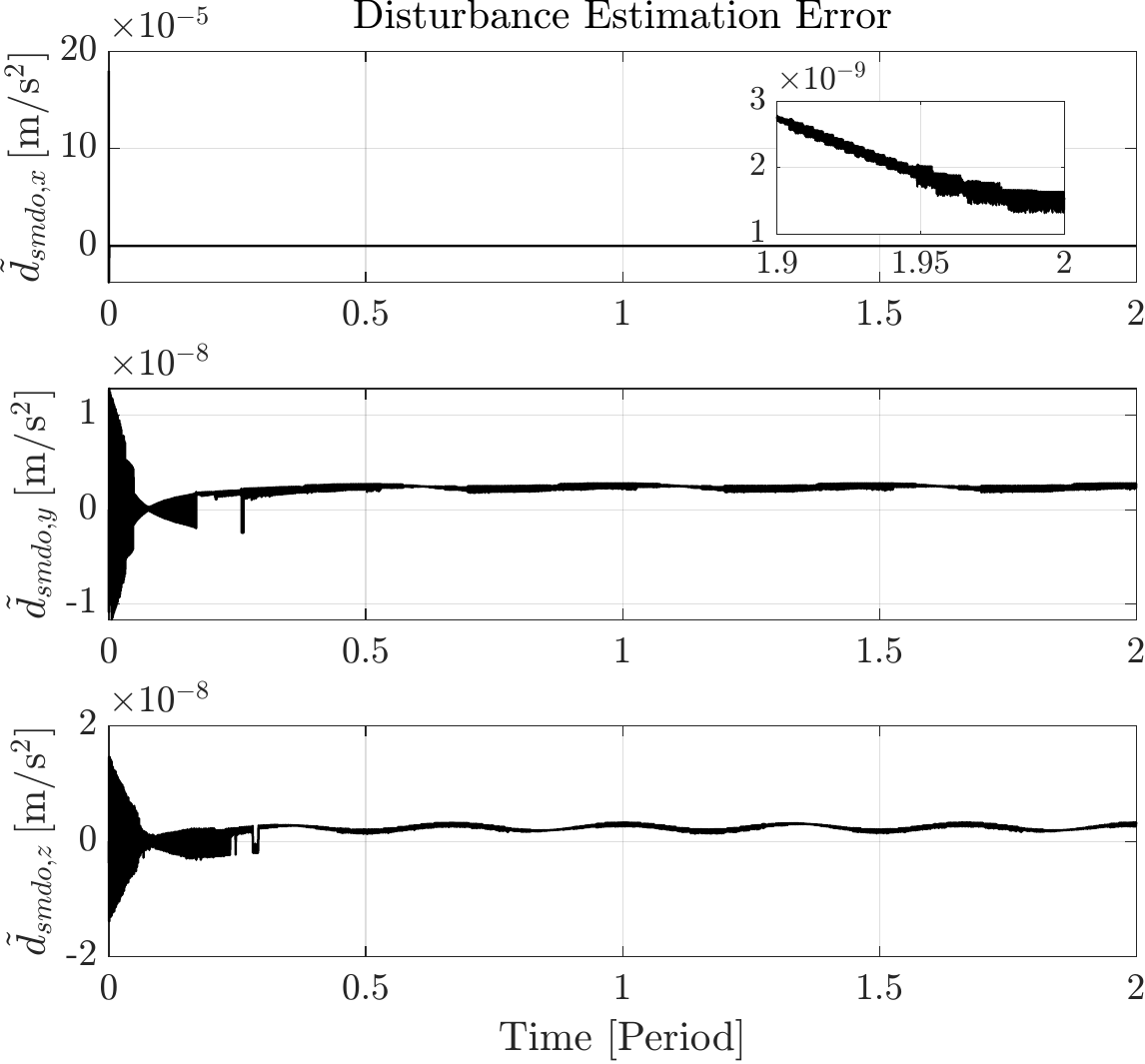}
  \caption{Residual disturbance-estimation error
  $\boldsymbol{d}
  +
  \hat{\boldsymbol{u}}_{\mathrm{smdo}}^{\mathrm{ad}}$.}
  \label{fig:Disturbance_Estimation_Error}
\end{figure}

\paragraph{ASMC Tracking Performance}

The ASMC layer receives the estimated states from the SMDO and generates
a tracking input that suppresses both the estimated tracking error and
the residual disturbance passed from the observer. The analysis proceeds
from the second-order sliding surface down to the actual tracking error,
mirroring the boundedness hierarchy of
Theorem~\ref{thm:asmc_tracking_boundedness}.

With $\delta = 0.01$, $C_1 = 2$, $C_2 = 3$, and
$\nu_{\mathrm{asmc}} = 9/11$, the predicted ultimate bounds are
\begin{align}
  \limsup_{t\to\infty}\|\hat{\boldsymbol{\sigma}}_2(t)\|
  &\leq
  0.01\ \mathrm{m/s^2},
  \label{eq:asmc_sigma2_numerical_bound}\\
  \limsup_{t\to\infty}\|\hat{\boldsymbol{\sigma}}_1(t)\|_\infty
  &\leq
  9.3844\times10^{-4}\ \mathrm{m/s},
  \label{eq:asmc_sigma1_numerical_bound}\\
  \limsup_{t\to\infty}\|\hat{\boldsymbol{e}}(t)\|_\infty
  &\leq
  4.6922\times10^{-4}\ \mathrm{m},
  \label{eq:asmc_ehat_numerical_bound}\\
  \limsup_{t\to\infty}\|\dot{\hat{\boldsymbol{e}}}(t)\|_\infty
  &\leq
  1.8769\times10^{-3}\ \mathrm{m/s}.
  \label{eq:asmc_ehatdot_numerical_bound}
\end{align}
Since the actual tracking error is the sum of the estimated tracking
error and the SMDO estimation error, the actual bounds follow as
\begin{align}
  \limsup_{t\to\infty}\|\boldsymbol{e}(t)\|_\infty
  &\leq
  4.6924\times10^{-4}\ \mathrm{m},
  \label{eq:actual_tracking_error_numerical_bound}\\
  \limsup_{t\to\infty}\|\dot{\boldsymbol{e}}(t)\|_\infty
  &\leq
  1.8769\times10^{-3}\ \mathrm{m/s}.
  \label{eq:actual_tracking_velocity_error_numerical_bound}
\end{align}
The SMDO estimation bounds are several orders of magnitude smaller than
the ASMC tracking bounds, so the actual tracking-error limits are
dominated by the ASMC boundary layer.

The ASMC sliding surfaces are plotted in Fig.~\ref{fig:ASMC_surfaces}.
The second-order surface $\hat{\boldsymbol{\sigma}}_2$
(Fig.~\ref{fig:ASMC_sigma2}) is driven into the $\delta = 0.01$
boundary layer after the initial transient, and the induced first-order
surface $\hat{\boldsymbol{\sigma}}_1$ (Fig.~\ref{fig:ASMC_sigma1})
remains within the bound of
Eq.~\eqref{eq:asmc_sigma1_numerical_bound}. The two-stage confinement
is the central mechanism through which the second-order ASMC propagates
its boundary-layer guarantee down to the tracking error level.

\begin{figure}[H]
  \centering
  \begin{subfigure}{0.48\linewidth}
    \centering
    \includegraphics[width=\linewidth]{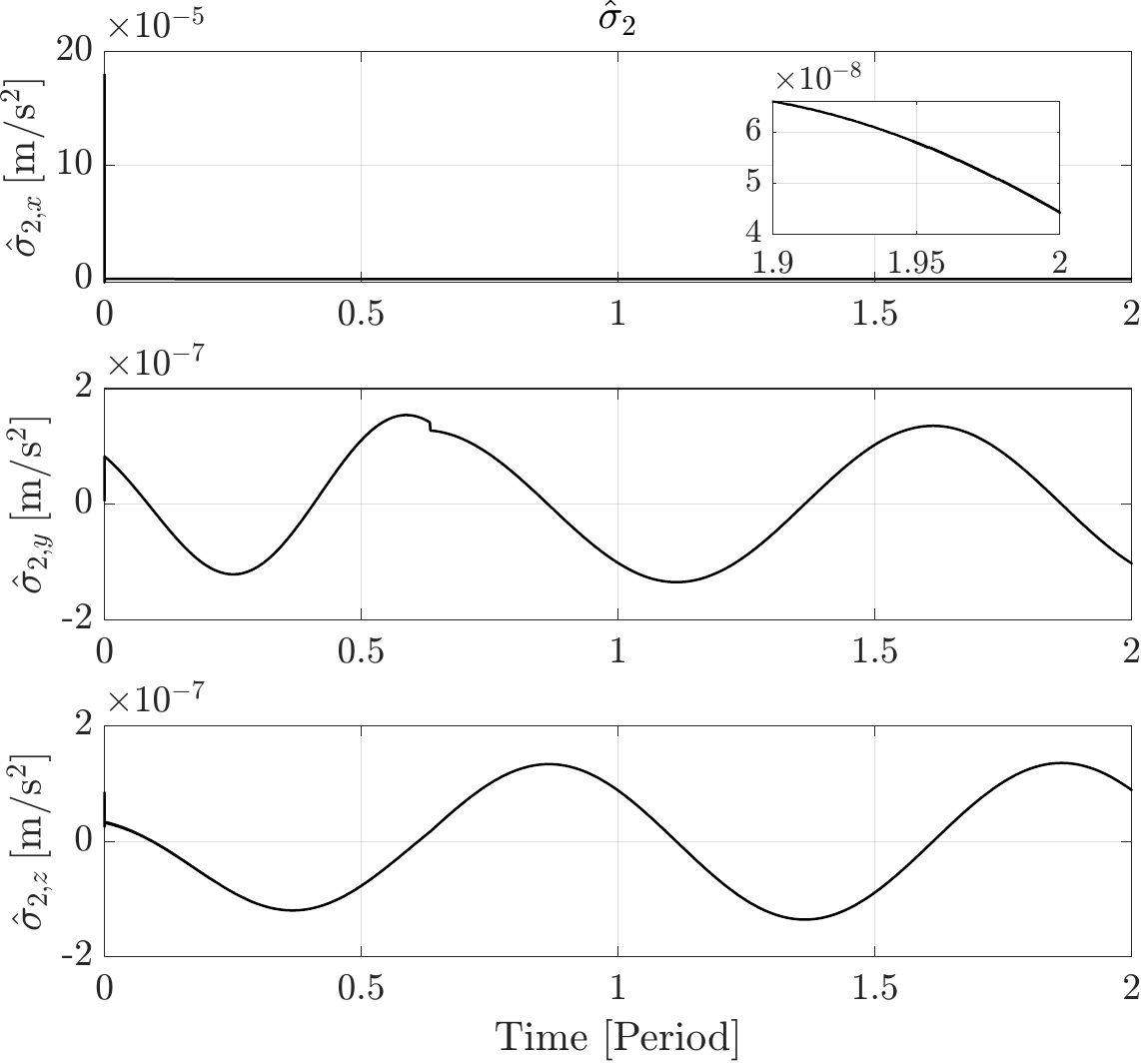}
    \caption{Second-order ASMC surface $\hat{\boldsymbol{\sigma}}_2$.}
    \label{fig:ASMC_sigma2}
  \end{subfigure}
  \hfill
  \begin{subfigure}{0.48\linewidth}
    \centering
    \includegraphics[width=\linewidth]{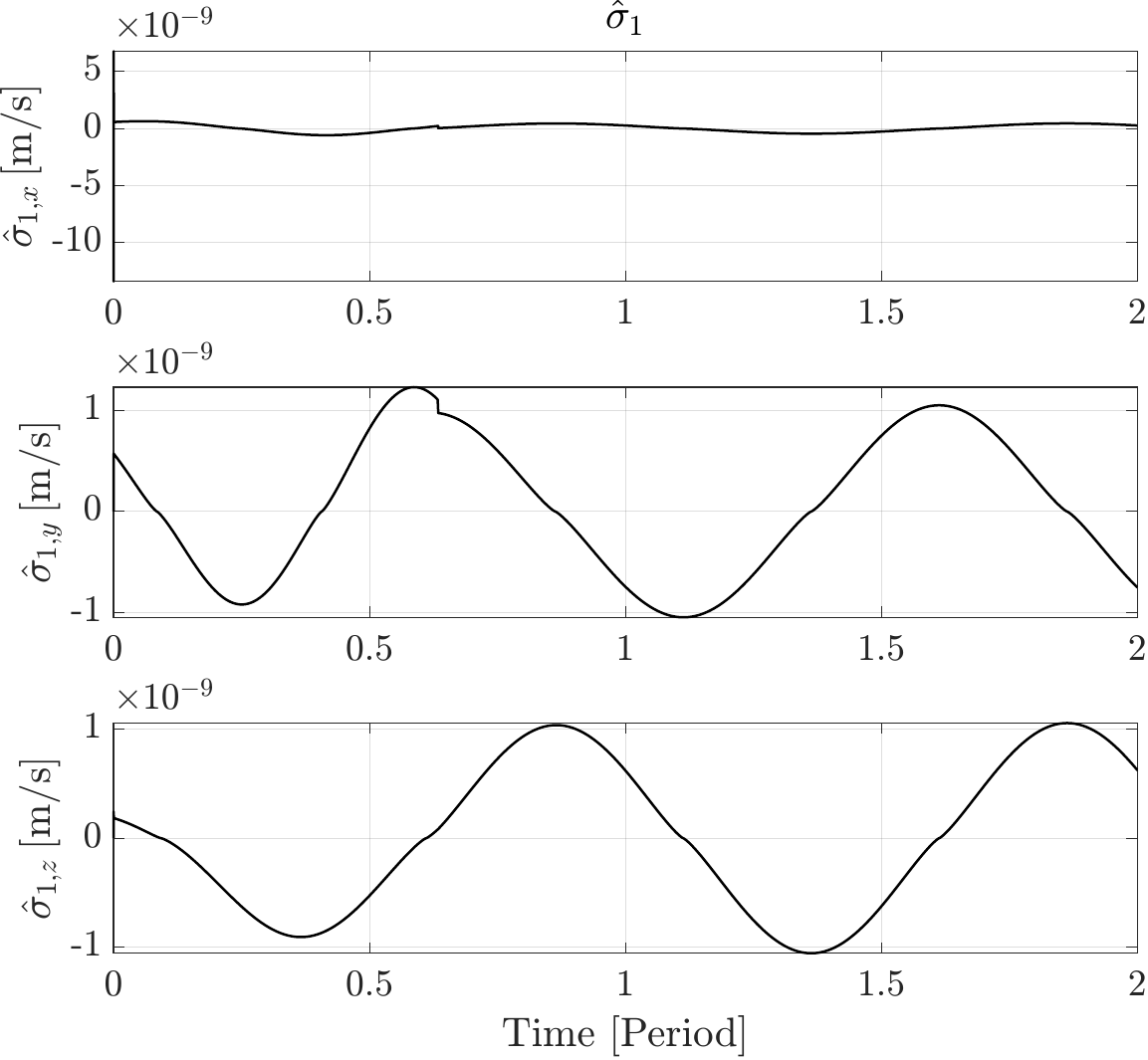}
    \caption{First-order ASMC surface $\hat{\boldsymbol{\sigma}}_1$.}
    \label{fig:ASMC_sigma1}
  \end{subfigure}
  \caption{ASMC sliding surfaces.}
  \label{fig:ASMC_surfaces}
\end{figure}

The estimated tracking errors are shown in
Fig.~\ref{fig:estimated_tracking_errors}. The position-level error
$\hat{\boldsymbol{e}} = \hat{\boldsymbol{q}} - \boldsymbol{q}_n$
(Fig.~\ref{fig:ehat}) enters the neighborhood predicted by
Eq.~\eqref{eq:asmc_ehat_numerical_bound} after the initial transient,
and the velocity-level error $\dot{\hat{\boldsymbol{e}}}$
(Fig.~\ref{fig:ehatdot}) remains within the bound of
Eq.~\eqref{eq:asmc_ehatdot_numerical_bound}. The ASMC tracking input
$\hat{\boldsymbol{u}}_{\mathrm{asmc}}$, presented in
Fig.~\ref{fig:u_asmc_second_order}, is the signal responsible for
driving these errors into their respective neighborhoods.

\begin{figure}[H]
  \centering
  \begin{subfigure}{0.48\linewidth}
    \centering
    \includegraphics[width=\linewidth]{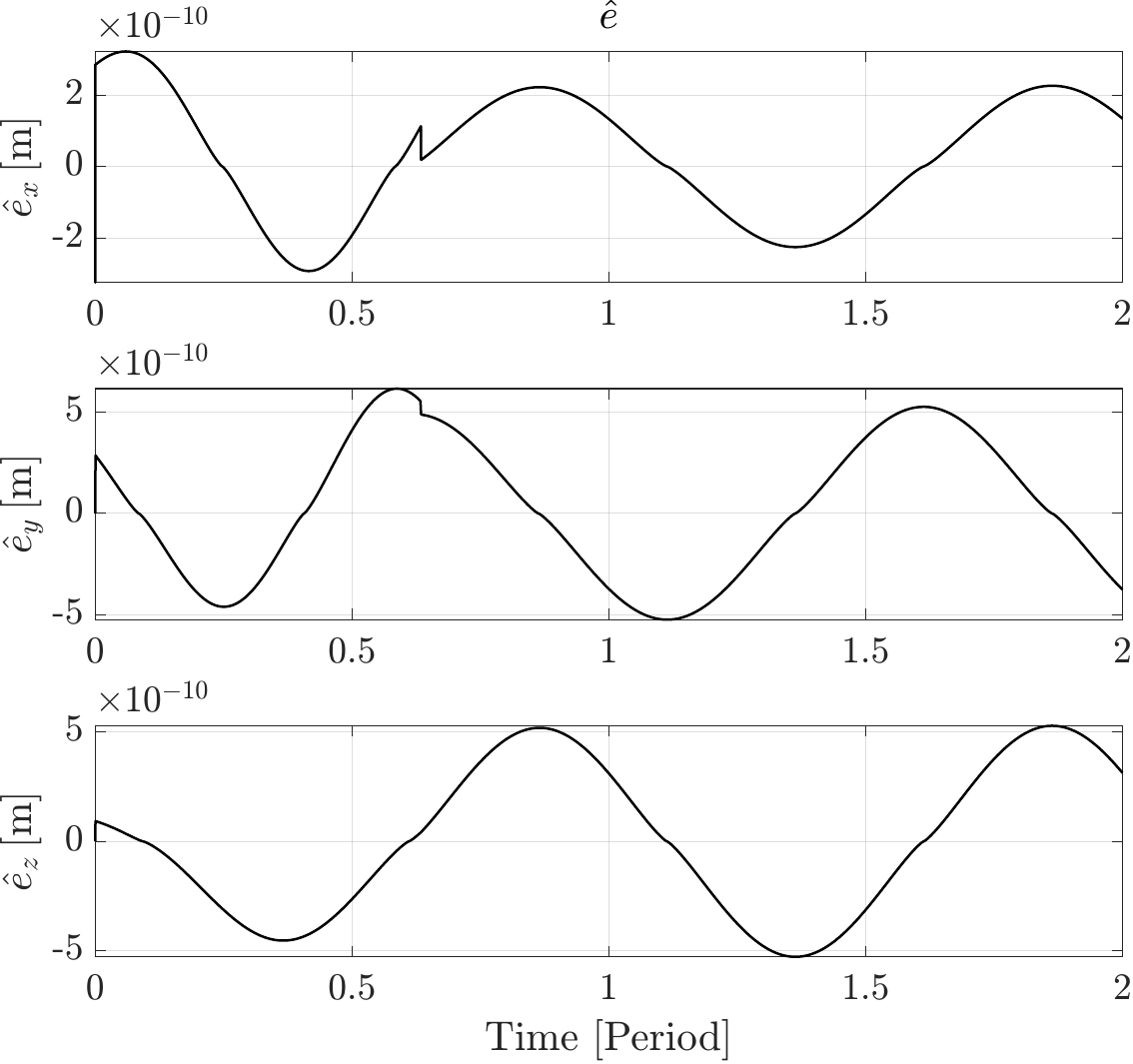}
    \caption{Estimated position tracking error
    $\hat{\boldsymbol{e}} = \hat{\boldsymbol{q}} - \boldsymbol{q}_n$.}
    \label{fig:ehat}
  \end{subfigure}
  \hfill
  \begin{subfigure}{0.48\linewidth}
    \centering
    \includegraphics[width=\linewidth]{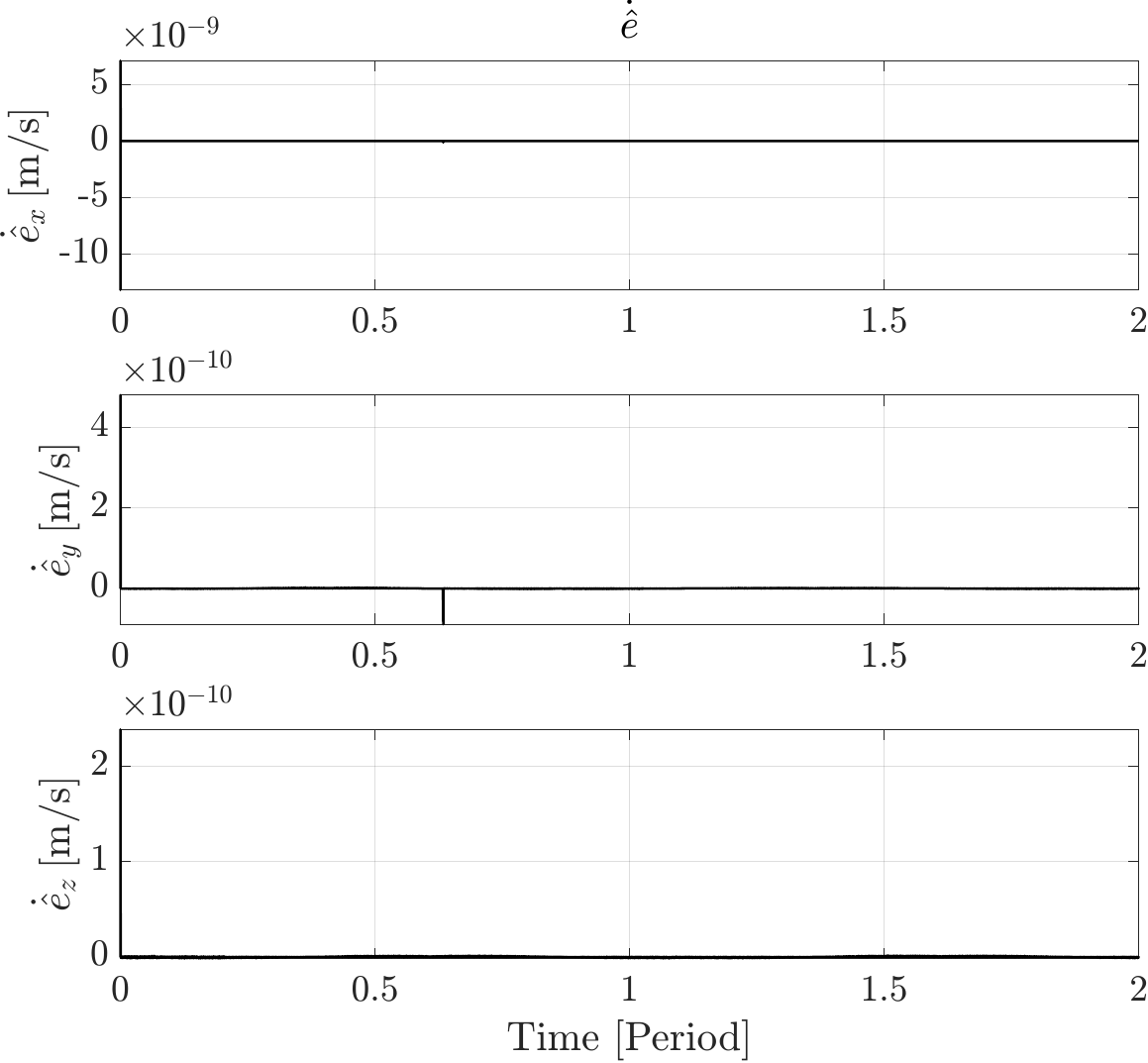}
    \caption{Estimated velocity tracking error
    $\dot{\hat{\boldsymbol{e}}} = \dot{\hat{\boldsymbol{q}}} - \dot{\boldsymbol{q}}_n$.}
    \label{fig:ehatdot}
  \end{subfigure}
  \caption{Estimated tracking errors.}
  \label{fig:estimated_tracking_errors}
\end{figure}

The actual tracking errors are shown in
Fig.~\ref{fig:actual_tracking_errors}. The position tracking error
$\boldsymbol{e} = \boldsymbol{q} - \boldsymbol{q}_n$
(Fig.~\ref{fig:Tracking_Error}) settles within the bound of
Eq.~\eqref{eq:actual_tracking_error_numerical_bound}, and the velocity
tracking error $\dot{\boldsymbol{e}}$
(Fig.~\ref{fig:edot}) remains within the bound of
Eq.~\eqref{eq:actual_tracking_velocity_error_numerical_bound}. The
SMDO estimation residual contributes a correction of only
$2\times10^{-10}\,\mathrm{m}$ to the position bound relative to the
estimated tracking error bound, indicating that the output-feedback
structure imposes no practical penalty on tracking accuracy.

\begin{figure}[H]
  \centering
  \begin{subfigure}{0.48\linewidth}
    \centering
    \includegraphics[width=\linewidth]{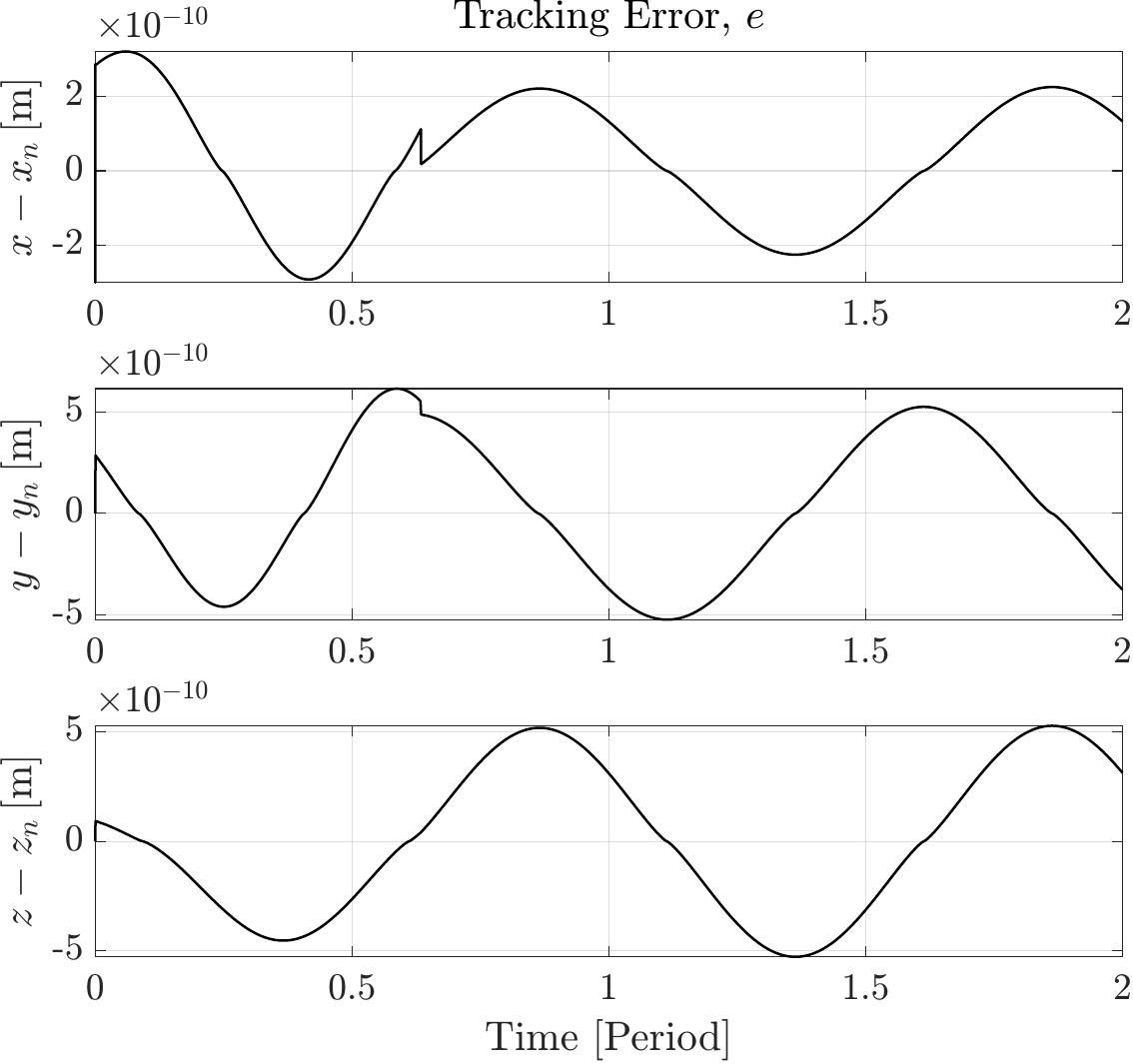}
    \caption{Actual position tracking error
    $\boldsymbol{e} = \boldsymbol{q} - \boldsymbol{q}_n$.}
    \label{fig:Tracking_Error}
  \end{subfigure}
  \hfill
  \begin{subfigure}{0.48\linewidth}
    \centering
    \includegraphics[width=\linewidth]{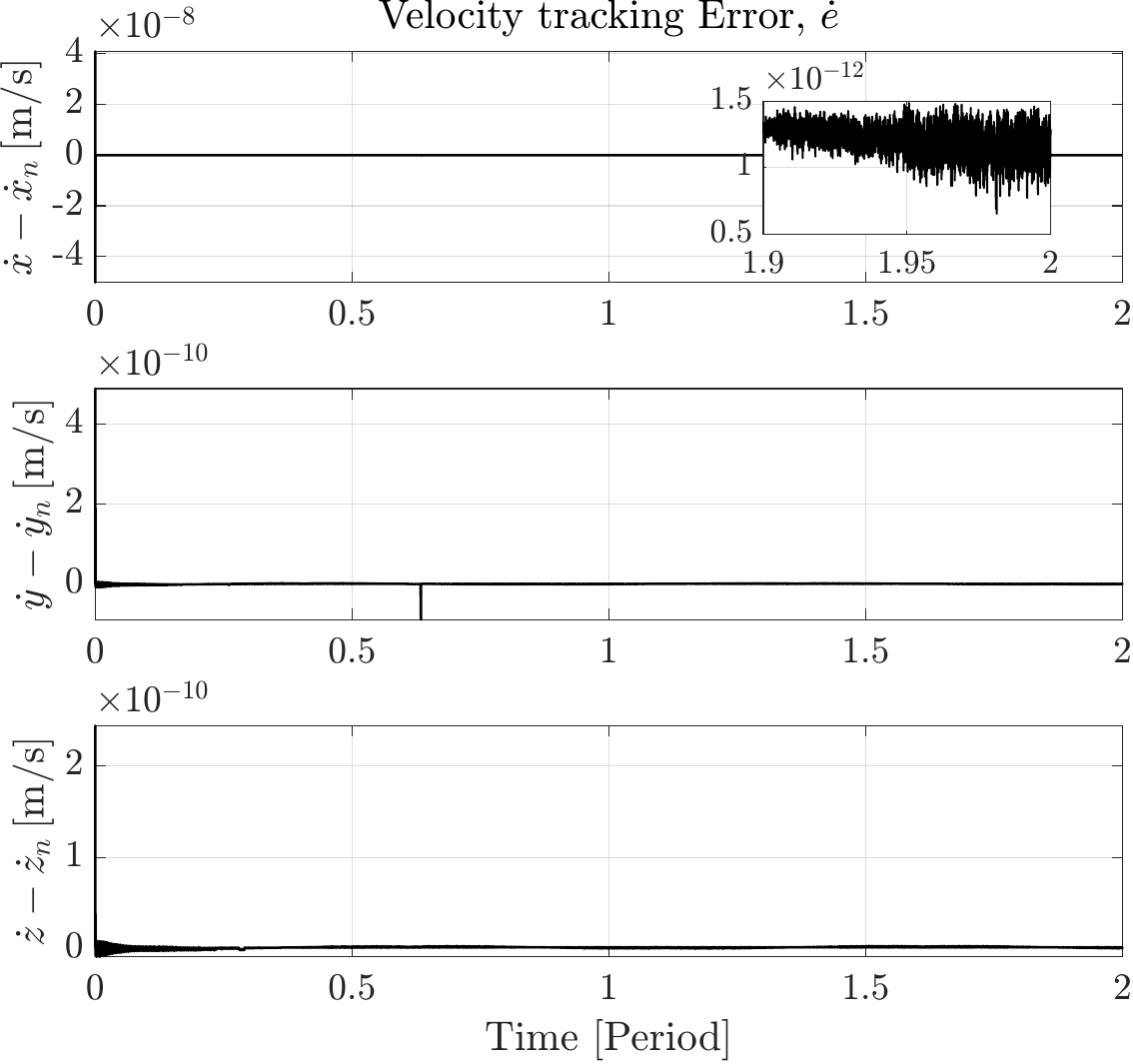}
    \caption{Actual velocity tracking error
    $\dot{\boldsymbol{e}} = \dot{\boldsymbol{q}} - \dot{\boldsymbol{q}}_n$.}
    \label{fig:edot}
  \end{subfigure}
  \caption{Actual tracking errors.}
  \label{fig:actual_tracking_errors}
\end{figure}

\paragraph{Adaptive Gains and Control Effort}

The adaptive SMDO gain $L(t)$ and ASMC gain $K(t)$ are plotted in
Fig.~\ref{fig:adaptive_gains}. Both gains remain finite throughout the
two-orbit simulation, which is consistent with the admissibility
conditions assumed in the theoretical analysis. The absence of unbounded
gain growth also implies that no unlimited actuation authority is
required to sustain the observed tracking performance.

\begin{figure}[H]
  \centering
  \begin{subfigure}{0.48\linewidth}
    \centering
    \includegraphics[width=\linewidth]{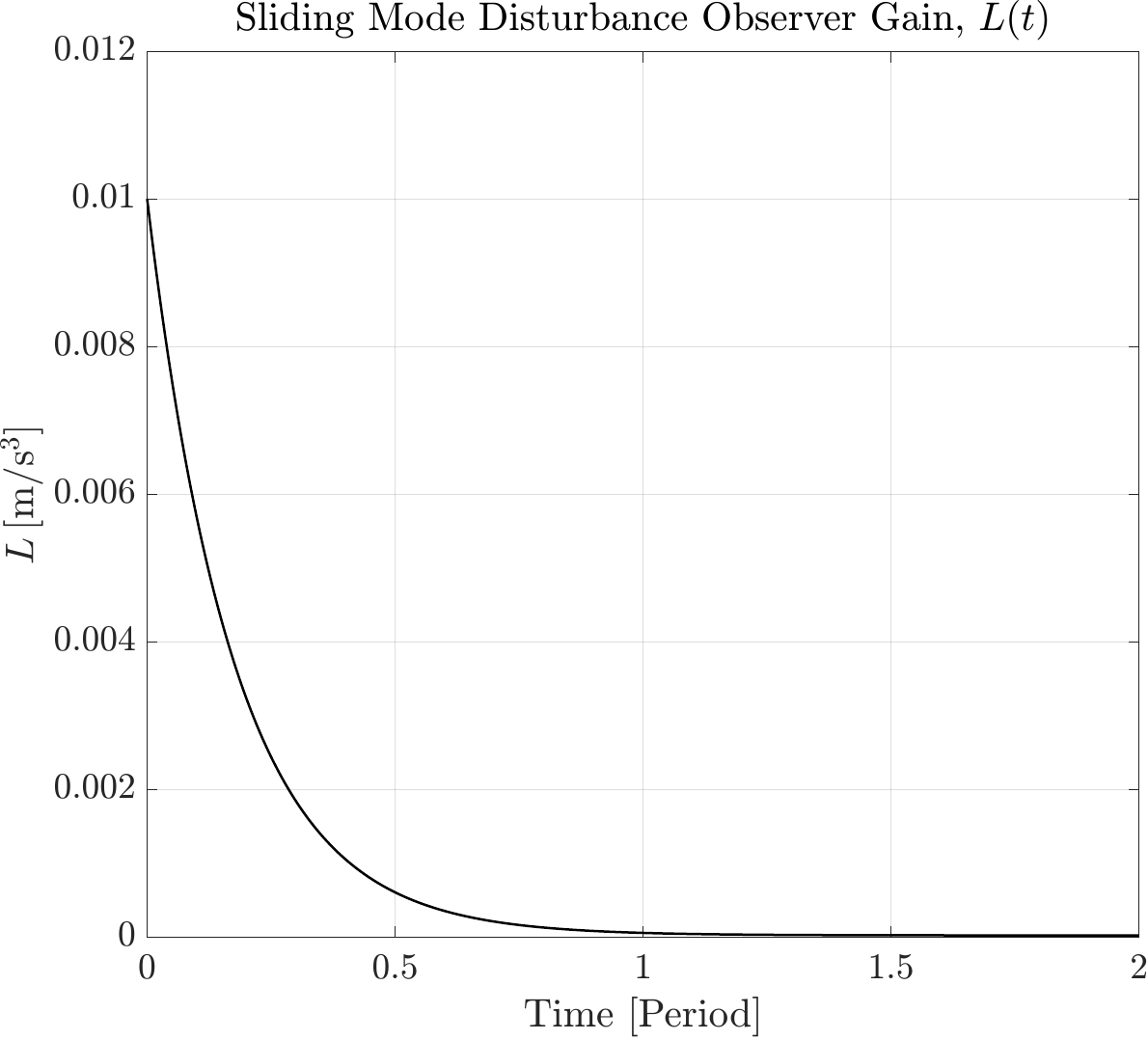}
    \caption{Adaptive SMDO gain $L(t)$.}
    \label{fig:smdo_gain_L}
  \end{subfigure}
  \hfill
  \begin{subfigure}{0.48\linewidth}
    \centering
    \includegraphics[width=\linewidth]{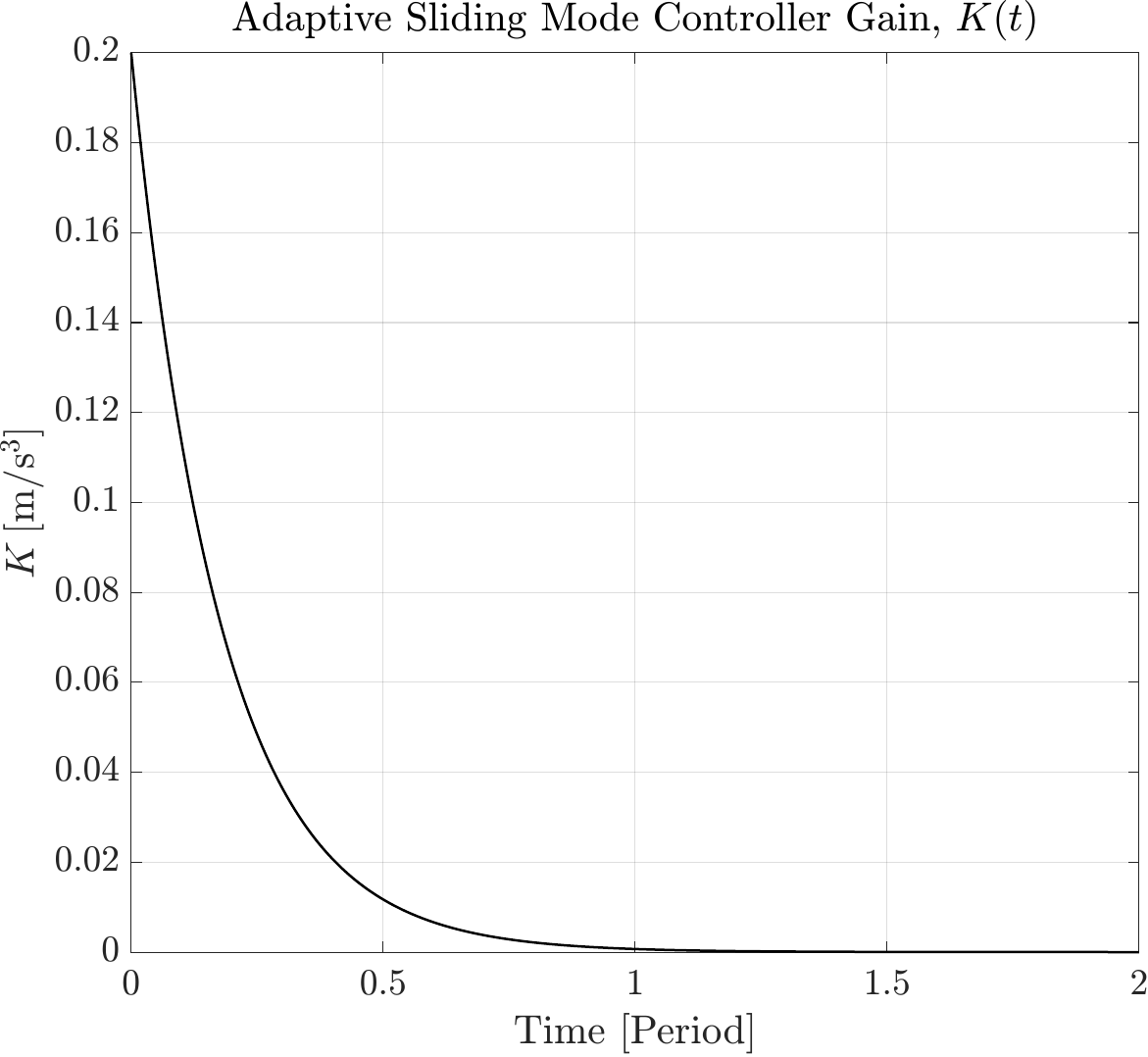}
    \caption{Adaptive ASMC gain $K(t)$.}
    \label{fig:asmc_gain_K}
  \end{subfigure}
  \caption{Adaptive gains of the SMDO and ASMC.}
  \label{fig:adaptive_gains}
\end{figure}

The three control components $\boldsymbol{u}_r$,
$\hat{\boldsymbol{u}}_{\mathrm{smdo}}$, and
$\hat{\boldsymbol{u}}_{\mathrm{asmc}}$, presented individually in
Figs.~\ref{fig:reaching_input_second_order}, \ref{fig:u_smdo},
and~\ref{fig:u_asmc_second_order}, combine to form the total control
input $\boldsymbol{u}_{\mathrm{tot}}$ shown in
Fig.~\ref{fig:total_input_second_order}. The total input remains
bounded throughout the simulation, and its magnitude is consistent
with the finite gain histories in Fig.~\ref{fig:adaptive_gains}.

\begin{figure}[H]
    \centering
    \includegraphics[width=0.45\linewidth]{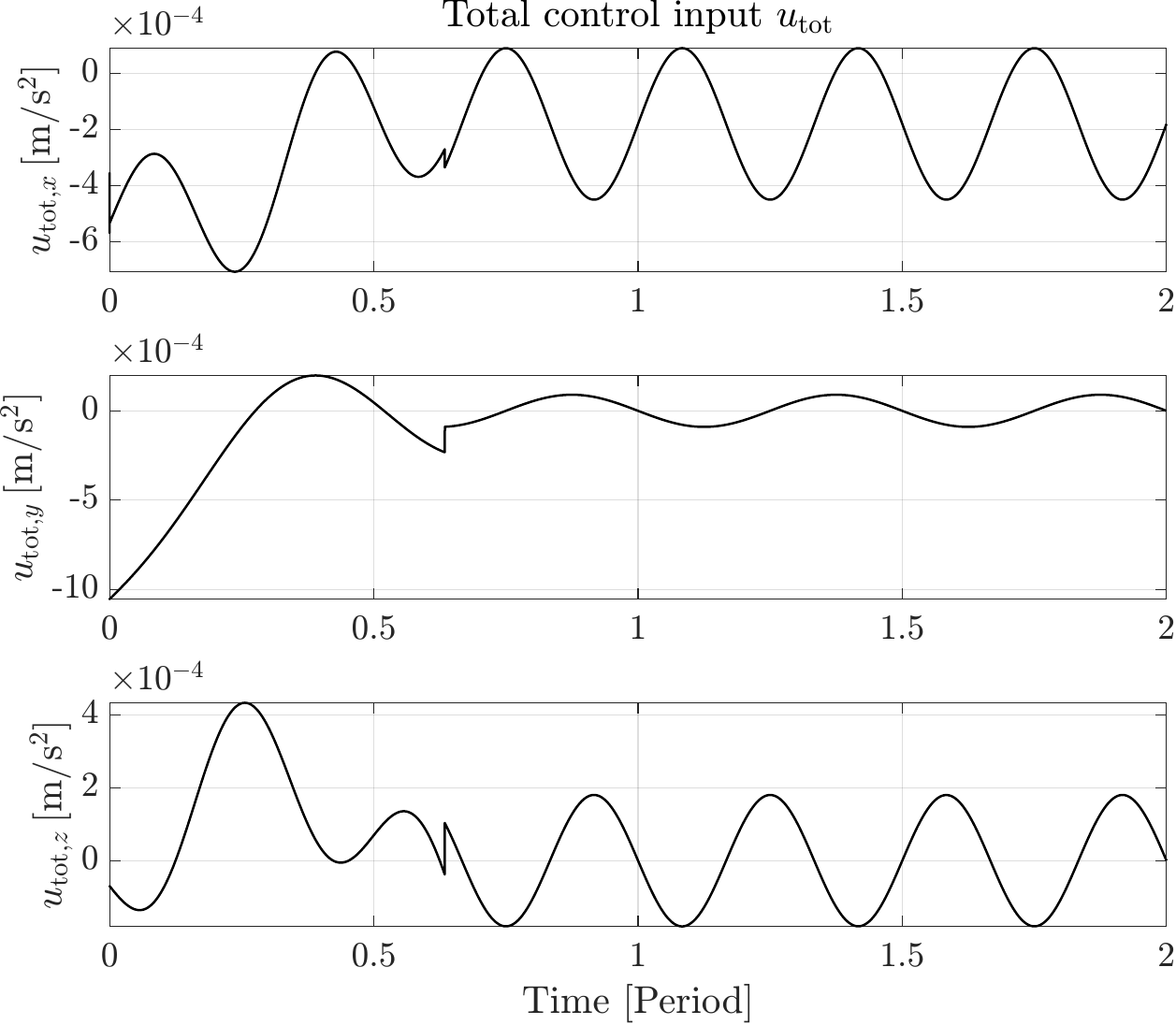}
    \caption{Total control input
    $\boldsymbol{u}_{\mathrm{tot}}
    = \boldsymbol{u}_r
    + \hat{\boldsymbol{u}}_{\mathrm{smdo}}
    + \hat{\boldsymbol{u}}_{\mathrm{asmc}}$.}
    \label{fig:total_input_second_order}
\end{figure}

To quantify the aggregate control expenditure, the integral of squared
control input (ISU), integral of absolute control input (IAU), and
integral of time-weighted absolute control input (ITAU) are evaluated:
\begin{align}
    \mathrm{ISU}\,[\mathrm{m^2/s^3}]
    &=
    \int_{0}^{T_{sim}}
    \left\|\boldsymbol{u}_{\mathrm{tot}}(t)\right\|^2 dt
    =
    0.0021,
    \\
    \mathrm{IAU}\,[\mathrm{m/s}]
    &=
    \int_{0}^{T_{sim}}
    \left\|\boldsymbol{u}_{\mathrm{tot}}(t)\right\| dt
    =
    5.8770,
    \\
    \mathrm{ITAU}\,[\mathrm{m}]
    &=
    \int_{0}^{T_{sim}}
    t\left\|\boldsymbol{u}_{\mathrm{tot}}(t)\right\| dt
    =
    2.6602\times10^4,
\end{align}
where $T_{sim}$ denotes the simulation duration, corresponding to two orbital periods of the target. ISU weights large-amplitude peaks more heavily, IAU accumulates the
overall control magnitude, and ITAU additionally penalizes effort that
persists later in the maneuver.

\graphicspath{{./figures_nav/}}

\subsection{Effect of Navigation Error}
\label{subsec:navigation_error_effect}
A key architectural property of the proposed output-feedback structure is
that navigation uncertainty does not enter the tracking controller
directly. To verify this claim, an additional simulation was conducted in
which additive Gaussian measurement noise was injected into the relative
position and velocity channels. The navigation-contaminated measurements
are expressed as
\begin{equation}
  \boldsymbol{q}_m
  =
  \boldsymbol{q}
  +
  \bar{q}_{\mathrm{nav}}\boldsymbol{r}_{q,\mathrm{nav}}(t),
  \qquad
  \dot{\boldsymbol{q}}_m
  =
  \dot{\boldsymbol{q}}
  +
  \bar{v}_{\mathrm{nav}}\boldsymbol{r}_{v,\mathrm{nav}}(t),
  \label{eq:navigation_error_model}
\end{equation}
where $\boldsymbol{q}_m$ and $\dot{\boldsymbol{q}}_m$ denote the
navigation-contaminated position and velocity measurements, respectively.
The random signals are dimensionless zero-mean Gaussian processes with
unit variance,
\begin{equation}
  \boldsymbol{r}_{q,\mathrm{nav}}(t)
  \sim
  \mathcal{N}(\boldsymbol{0},I),
  \qquad
  \boldsymbol{r}_{v,\mathrm{nav}}(t)
  \sim
  \mathcal{N}(\boldsymbol{0},I),
  \label{eq:navigation_random_signal}
\end{equation}
and the scaling constants were set to
\begin{equation}
  \bar{q}_{\mathrm{nav}}
  =
  0.1\,\mathrm{m},
  \qquad
  \bar{v}_{\mathrm{nav}}
  =
  0.001\,\mathrm{m/s}.
  \label{eq:navigation_error_constants}
\end{equation}

The SMDO and ASMC parameters used in this case are
\begin{align}
\epsilon &= 2,
&
\lambda_1 &= 2,
&
\eta &= 3,
&
L(0) &= 0.1,
&
L^* &= 0.1,
&
\lambda_2 &= 3,
&
\nu_{\mathrm{smdo}} &= \frac{8}{11},
\label{eq:navigation_smdo_parameters}
\end{align}
and
\begin{align}
\delta &= 0.5,
&
C_1 &= 2,
&
\gamma &= 0.01,
&
K(0) &= 2,
&
K^* &= 0.1,
&
C_2 &= 3,
&
\nu_{\mathrm{asmc}} &= \frac{9}{11}.
\label{eq:navigation_asmc_parameters}
\end{align}
Compared with the case without navigation error, the SMDO boundary-layer
parameter $\epsilon$ is increased to account for measurement contamination
caused by navigation errors.

\subsubsection*{SMDO response}

Because the SMDO reconstructs the mismatch between the actual measured
plant and the auxiliary observer model, measurement contamination enters
the observer channel directly. With the parameters above, the SMDO
boundary term evaluates to
\begin{equation}
\left(
\frac{\epsilon}{\lambda_2}
\right)^{1/\nu_{\mathrm{smdo}}}
=
\left(
\frac{2}{3}
\right)^{11/8}
=
0.5726,
\label{eq:nav_smdo_boundary_value}
\end{equation}
and Theorem~\ref{thm:smdo_residual_boundedness} gives the following
ultimate bounds:
\begin{align}
\limsup_{t\to\infty}
\|\hat{\boldsymbol{s}}_2(t)\|
&\leq
2\,\mathrm{m/s^2},
\label{eq:nav_smdo_s2_bound}\\
\limsup_{t\to\infty}
\|\hat{\boldsymbol{s}}_1(t)\|_\infty
&\leq
0.5726\,\mathrm{m/s},
\label{eq:nav_smdo_s1_bound}\\
\limsup_{t\to\infty}
\|\tilde{\boldsymbol{q}}(t)\|_\infty
&\leq
0.2863\,\mathrm{m},
\label{eq:nav_smdo_qtilde_bound}\\
\limsup_{t\to\infty}
\|\dot{\tilde{\boldsymbol{q}}}(t)\|_\infty
&\leq
1.1453\,\mathrm{m/s}.
\label{eq:nav_smdo_qdottilde_bound}
\end{align}
These bounds are intentionally loose because the SMDO operates on
navigation-contaminated measurements
. The visible effect appears in the
SMDO surfaces and state-estimation errors, which exhibit additional
fluctuations, as shown in
Figs.~\ref{fig:nav_smdo_surfaces} and~\ref{fig:nav_state_estimation_errors}.

\begin{figure}[H]
\centering
\begin{subfigure}{0.48\linewidth}
\centering
\includegraphics[width=\linewidth]{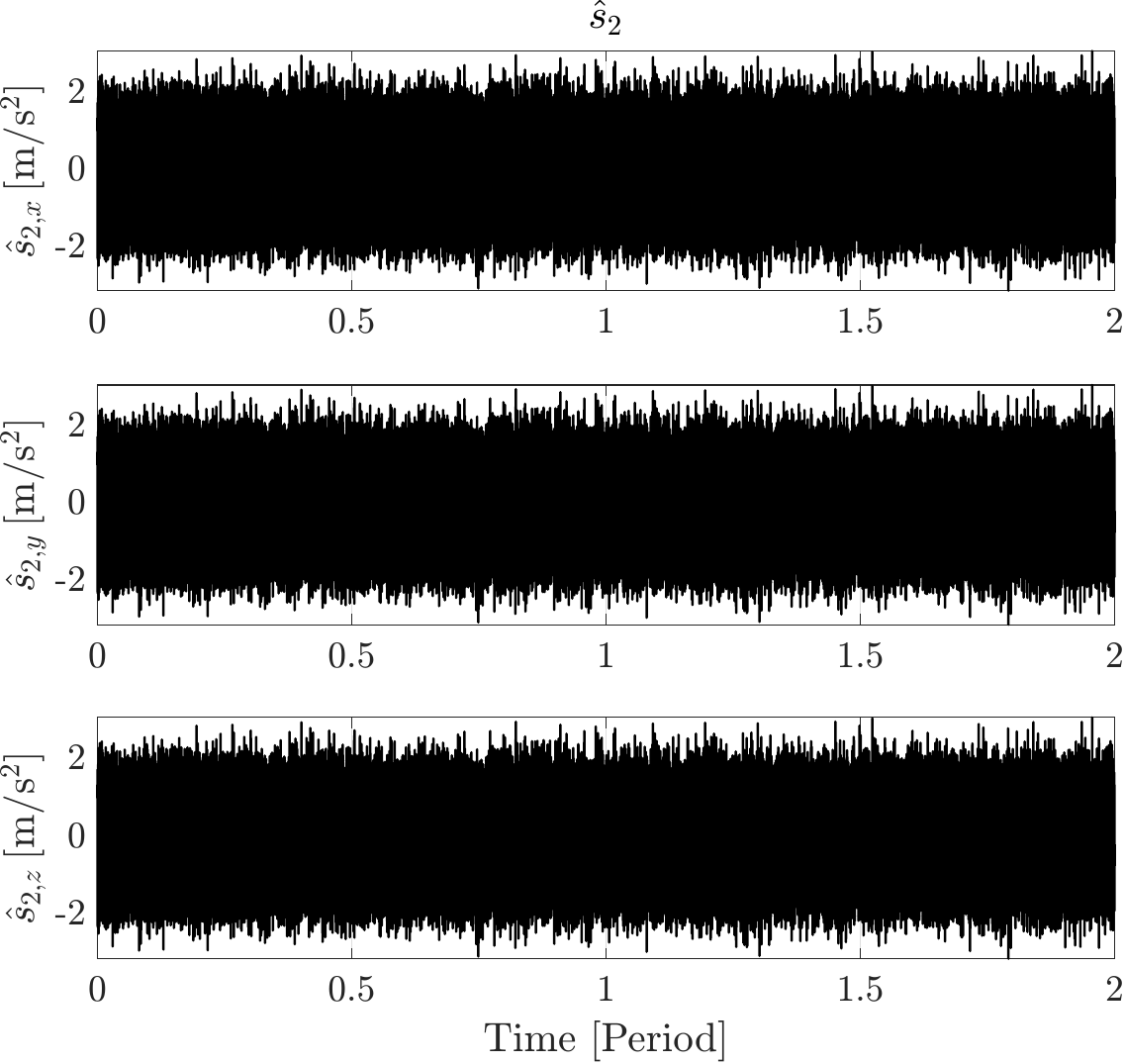}
\caption{Second-order SMDO surface $\hat{\boldsymbol{s}}_2$.}
\label{fig:nav_smdo_s2}
\end{subfigure}
\hfill
\begin{subfigure}{0.48\linewidth}
\centering
\includegraphics[width=\linewidth]{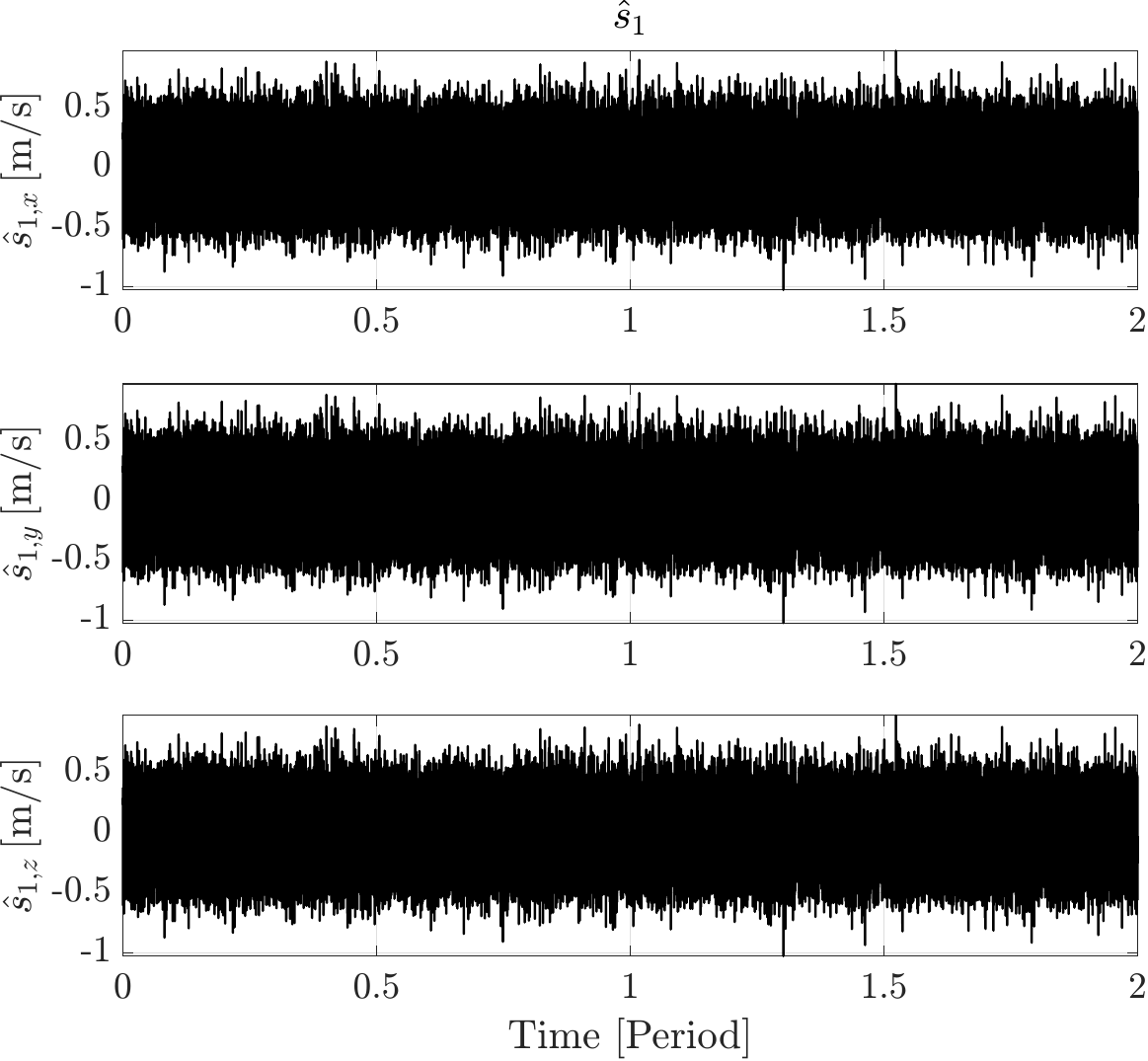}
\caption{First-order SMDO surface $\hat{\boldsymbol{s}}_1$.}
\label{fig:nav_smdo_s1}
\end{subfigure}
\caption{SMDO surfaces under navigation error.}
\label{fig:nav_smdo_surfaces}
\end{figure}

\begin{figure}[H]
  \centering
  \begin{subfigure}{0.48\linewidth}
    \centering
    \includegraphics[width=\linewidth]{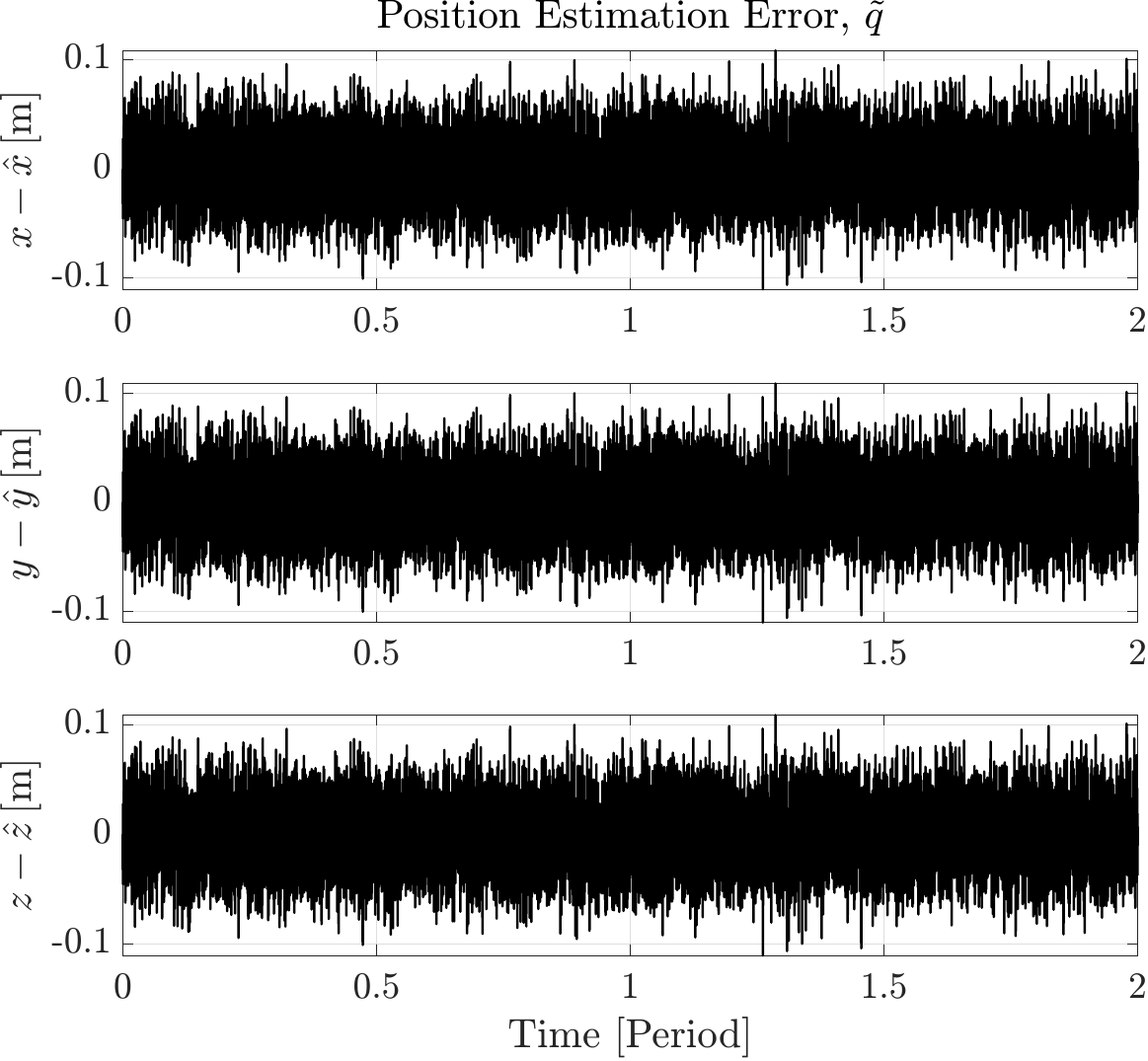}
    \caption{Position-estimation error
    $\tilde{\boldsymbol{q}} = \boldsymbol{q} - \hat{\boldsymbol{q}}$.}
    \label{fig:nav_state_estimation_error_position}
  \end{subfigure}
  \hfill
  \begin{subfigure}{0.48\linewidth}
    \centering
    \includegraphics[width=\linewidth]{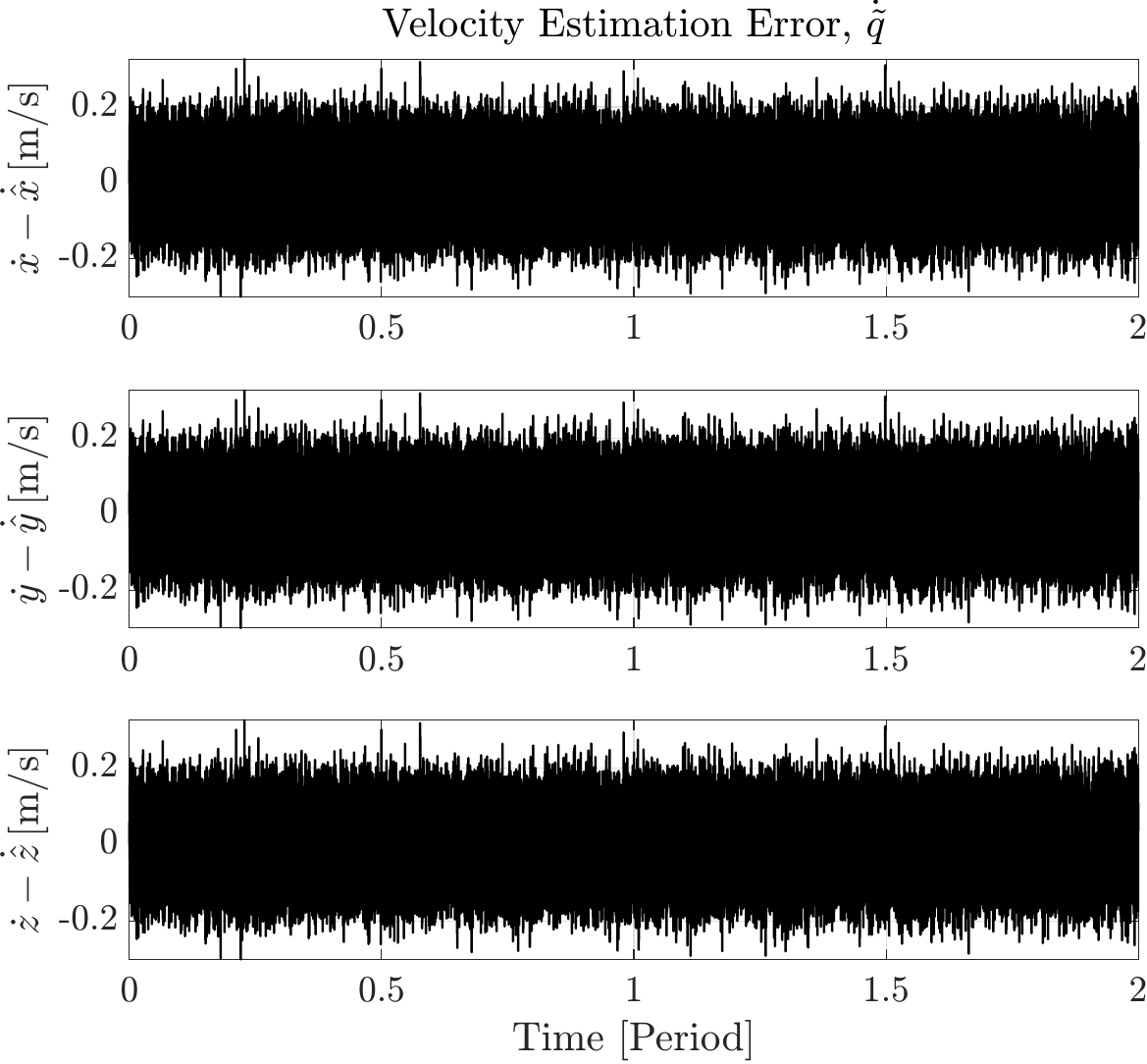}
    \caption{Velocity-estimation error
    $\dot{\tilde{\boldsymbol{q}}}
    = \dot{\boldsymbol{q}} - \dot{\hat{\boldsymbol{q}}}$.}
    \label{fig:nav_state_estimation_error_velocity}
  \end{subfigure}
  \caption{State-estimation errors under navigation error.}
  \label{fig:nav_state_estimation_errors}
\end{figure}

\subsubsection*{ASMC response and tracking performance isolation}

The ASMC operates exclusively on the estimated state
$\hat{\boldsymbol{q}}$ and $\dot{\hat{\boldsymbol{q}}}$ produced by the
SMDO, rather than on the raw navigation-corrupted measurements. As a
result, navigation error does not enter the tracking control loop
directly; the SMDO acts as an intermediary that absorbs the measurement
contamination before it can propagate to the controller. This separation
yields a practically valuable property: even when navigation quality
degrades, the tracking computation remains robust and the ASMC gains
adapt based on smooth estimated signals.

This isolation is also reflected in the adaptive gain histories,
as shown in Fig.~\ref{fig:nav_gains}.
The SMDO gain $L(t)$ exhibits large fluctuations in response to
the navigation-contaminated measurements, whereas the ASMC gain
$K(t)$ remains smooth and well-regulated throughout the simulation.
This contrast further confirms that the measurement contamination
does not propagate into the adaptive tracking layer.

\begin{figure}[H]
\centering
\begin{subfigure}{0.48\linewidth}
\centering
\includegraphics[width=\linewidth]{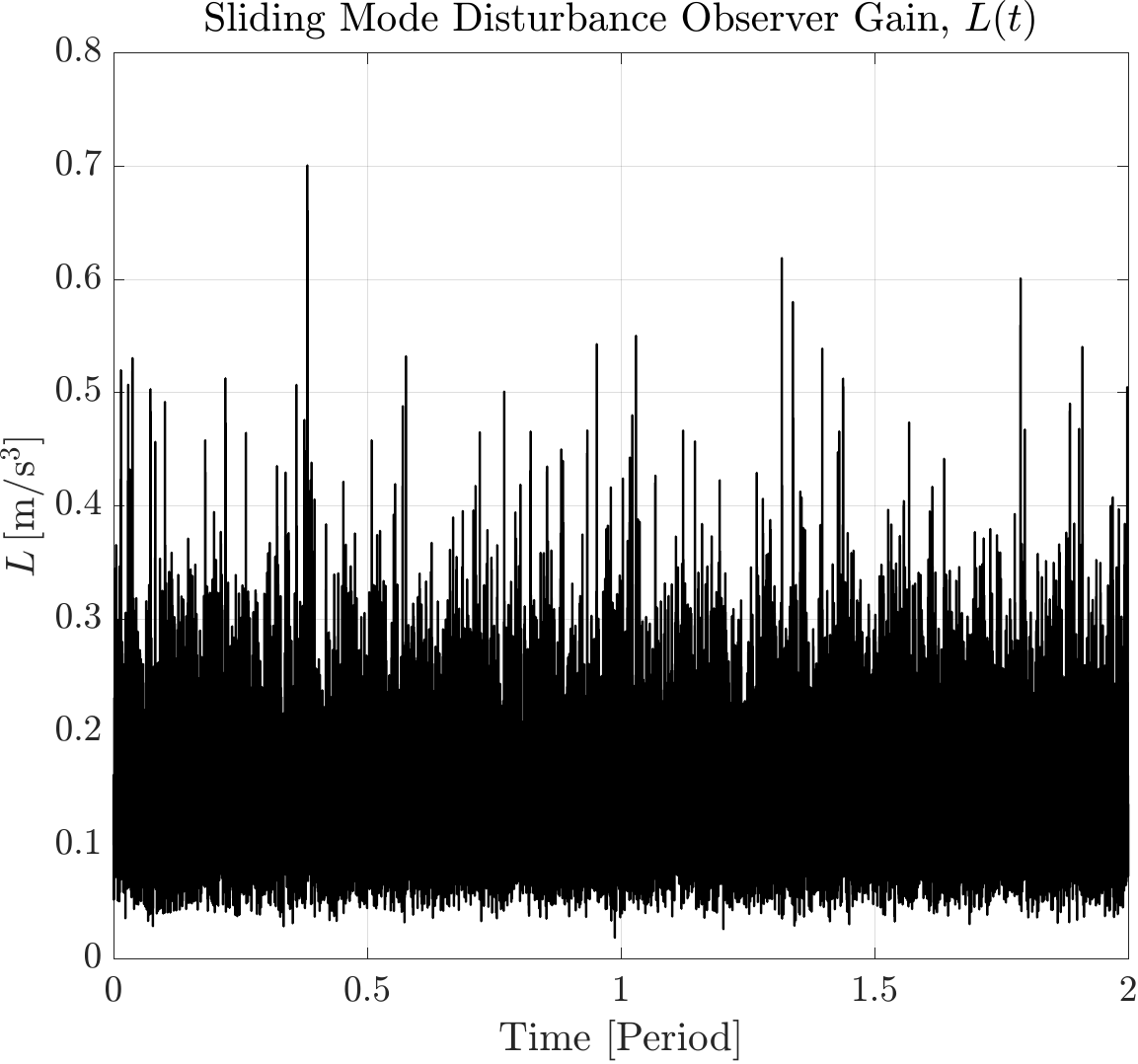}
\caption{SMDO adaptive gain $L(t)$.}
\label{fig:nav_smdo_gain}
\end{subfigure}
\hfill
\begin{subfigure}{0.48\linewidth}
\centering
\includegraphics[width=\linewidth]{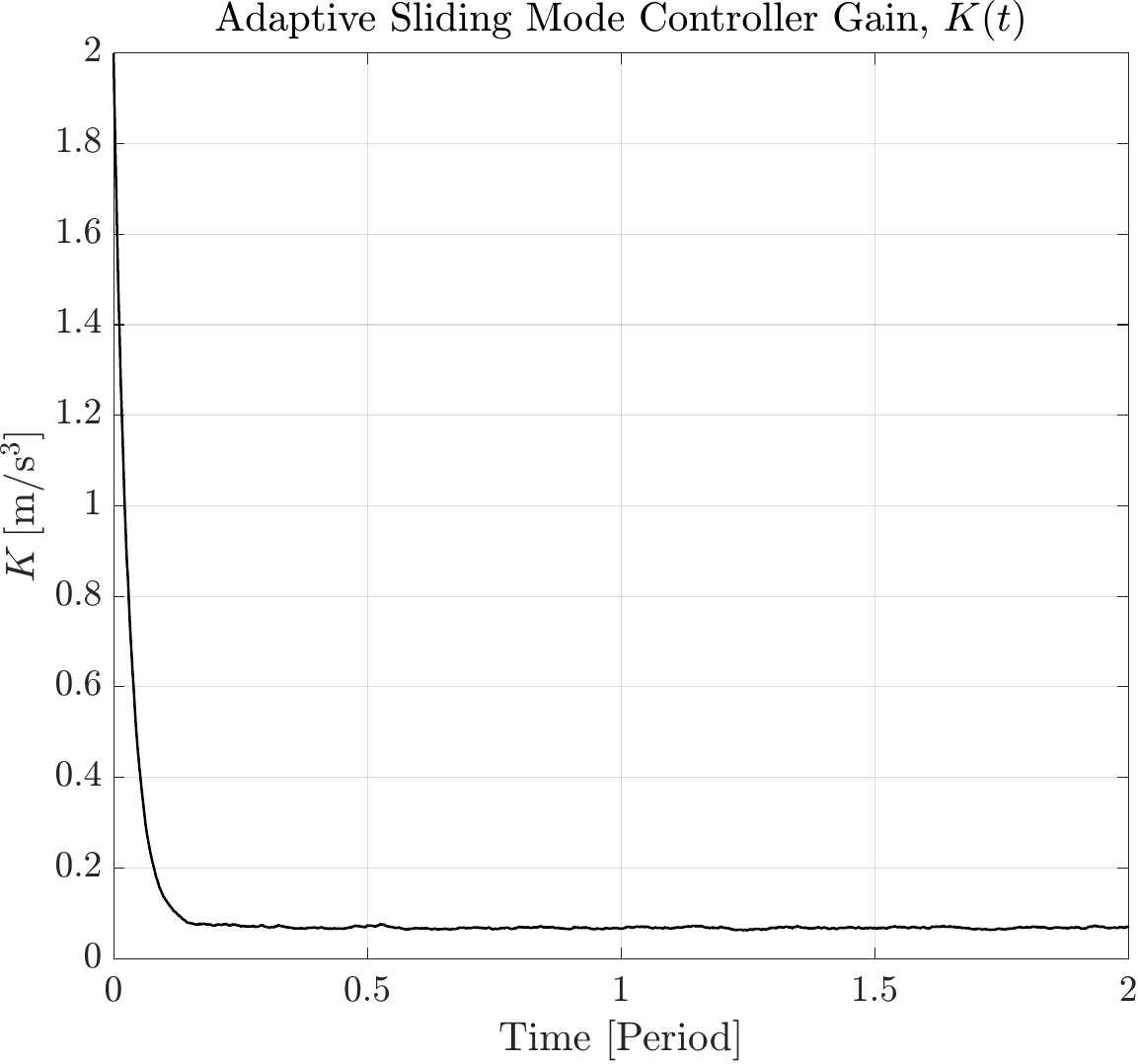}
\caption{ASMC adaptive gain $K(t)$.}
\label{fig:nav_asmc_gain}
\end{subfigure}
\caption{Adaptive gain histories under navigation error.}
\label{fig:nav_gains}
\end{figure}

The ASMC boundary term evaluates to
\begin{equation}
\left(
\frac{\delta}{C_2}
\right)^{1/\nu_{\mathrm{asmc}}}
=
\left(
\frac{0.5}{3}
\right)^{11/9}
=
0.1667,
\label{eq:nav_asmc_boundary_value}
\end{equation}
and the corresponding ASMC ultimate bounds are
\begin{align}
\limsup_{t\to\infty}
\|\hat{\boldsymbol{\sigma}}_2(t)\|
&\leq
0.5\,\mathrm{m/s^2},
\label{eq:nav_asmc_sigma2_bound}\\
\limsup_{t\to\infty}
\|\hat{\boldsymbol{\sigma}}_1(t)\|_\infty
&\leq
0.1667\,\mathrm{m/s},
\label{eq:nav_asmc_sigma1_bound}\\
\limsup_{t\to\infty}
\|\hat{\boldsymbol{e}}(t)\|_\infty
&\leq
0.0560\,\mathrm{m},
\label{eq:nav_asmc_ehat_bound}\\
\limsup_{t\to\infty}
\|\dot{\hat{\boldsymbol{e}}}(t)\|_\infty
&\leq
0.2238\,\mathrm{m/s}.
\label{eq:nav_asmc_ehatdot_bound}
\end{align}

\begin{figure}[H]
\centering
\begin{subfigure}{0.48\linewidth}
\centering
\includegraphics[width=\linewidth]{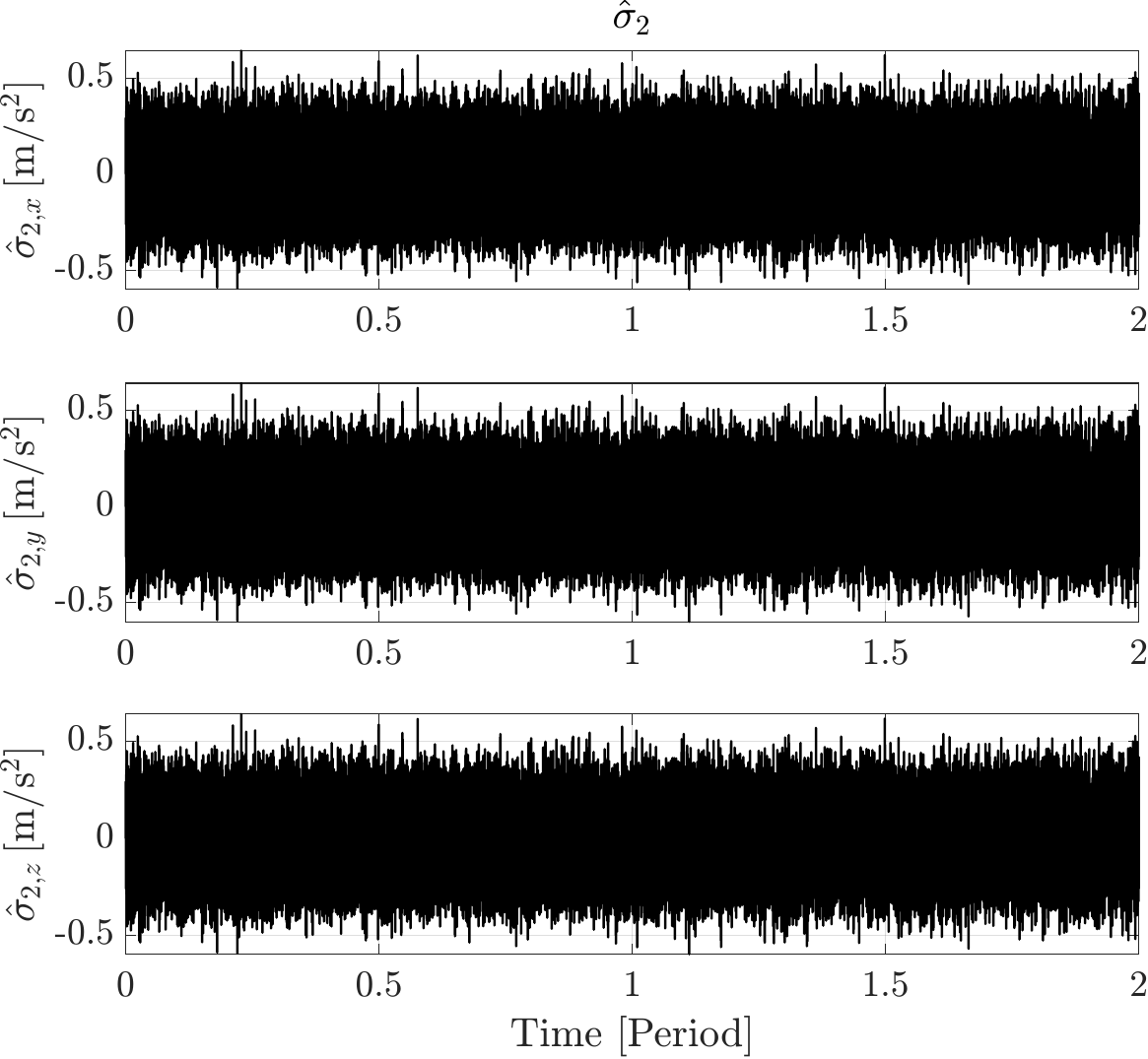}
\caption{Second-order ASMC surface $\hat{\boldsymbol{\sigma}}_2$.}
\label{fig:nav_asmc_sigma2}
\end{subfigure}
\hfill
\begin{subfigure}{0.48\linewidth}
\centering
\includegraphics[width=\linewidth]{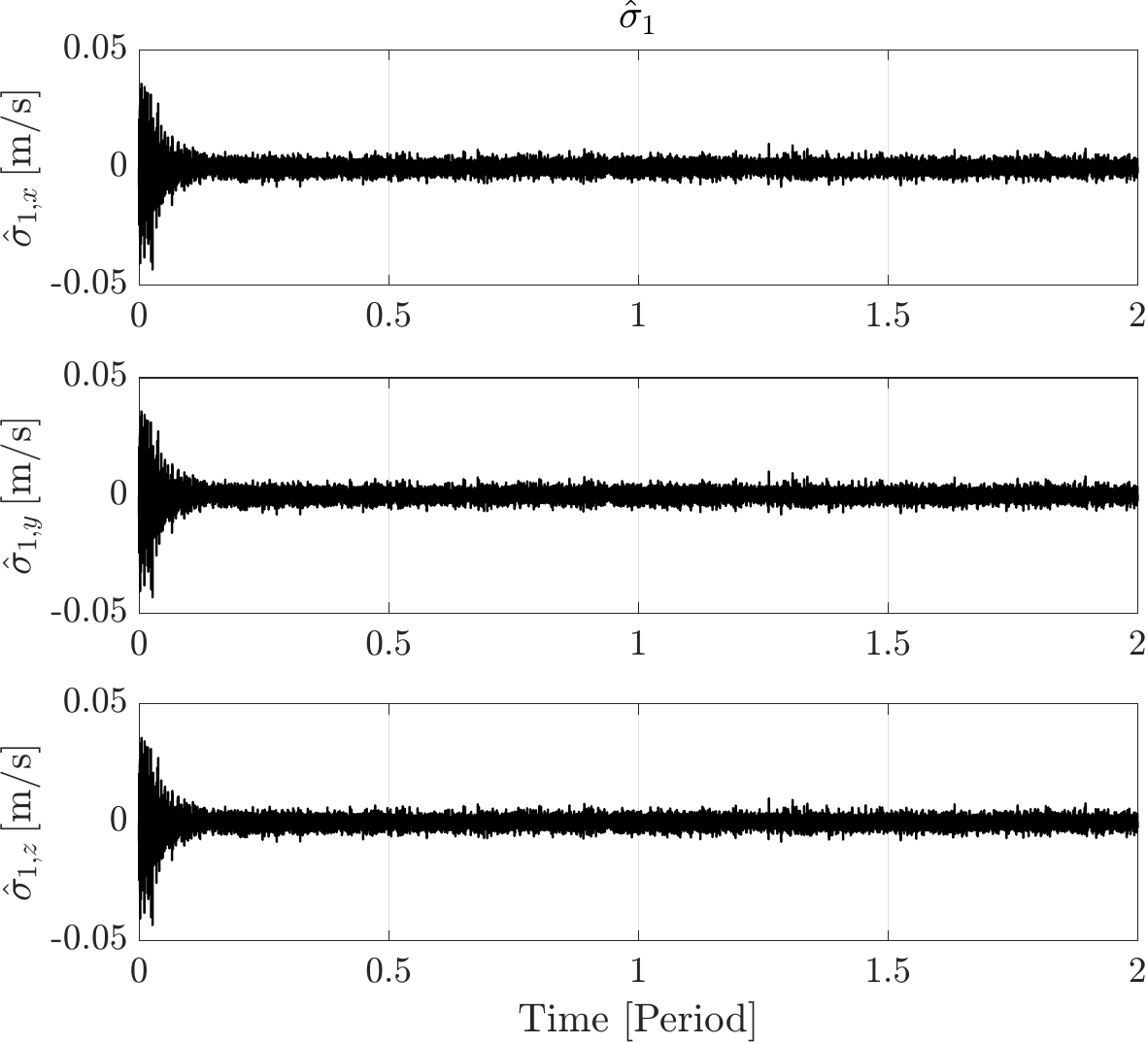}
\caption{First-order ASMC surface $\hat{\boldsymbol{\sigma}}_1$.}
\label{fig:nav_asmc_sigma1}
\end{subfigure}
\caption{ASMC surfaces under navigation error.}
\label{fig:nav_asmc_surfaces}
\end{figure}

\begin{figure}[H]
\centering
\begin{subfigure}{0.48\linewidth}
\centering
\includegraphics[width=\linewidth]{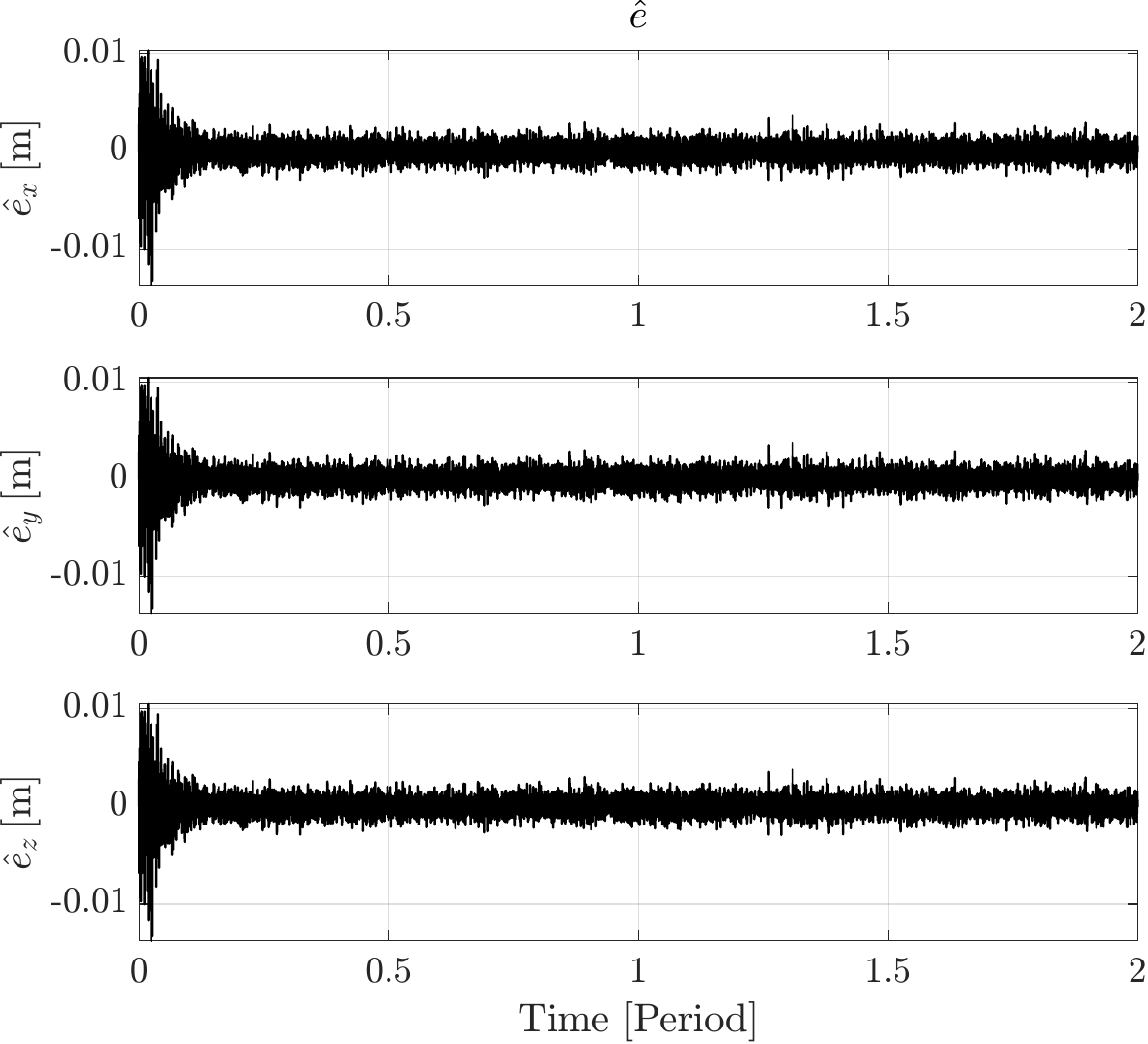}
\caption{Estimated tracking position error $\hat{\boldsymbol{e}}$.}
\label{fig:nav_ehat}
\end{subfigure}
\hfill
\begin{subfigure}{0.48\linewidth}
\centering
\includegraphics[width=\linewidth]{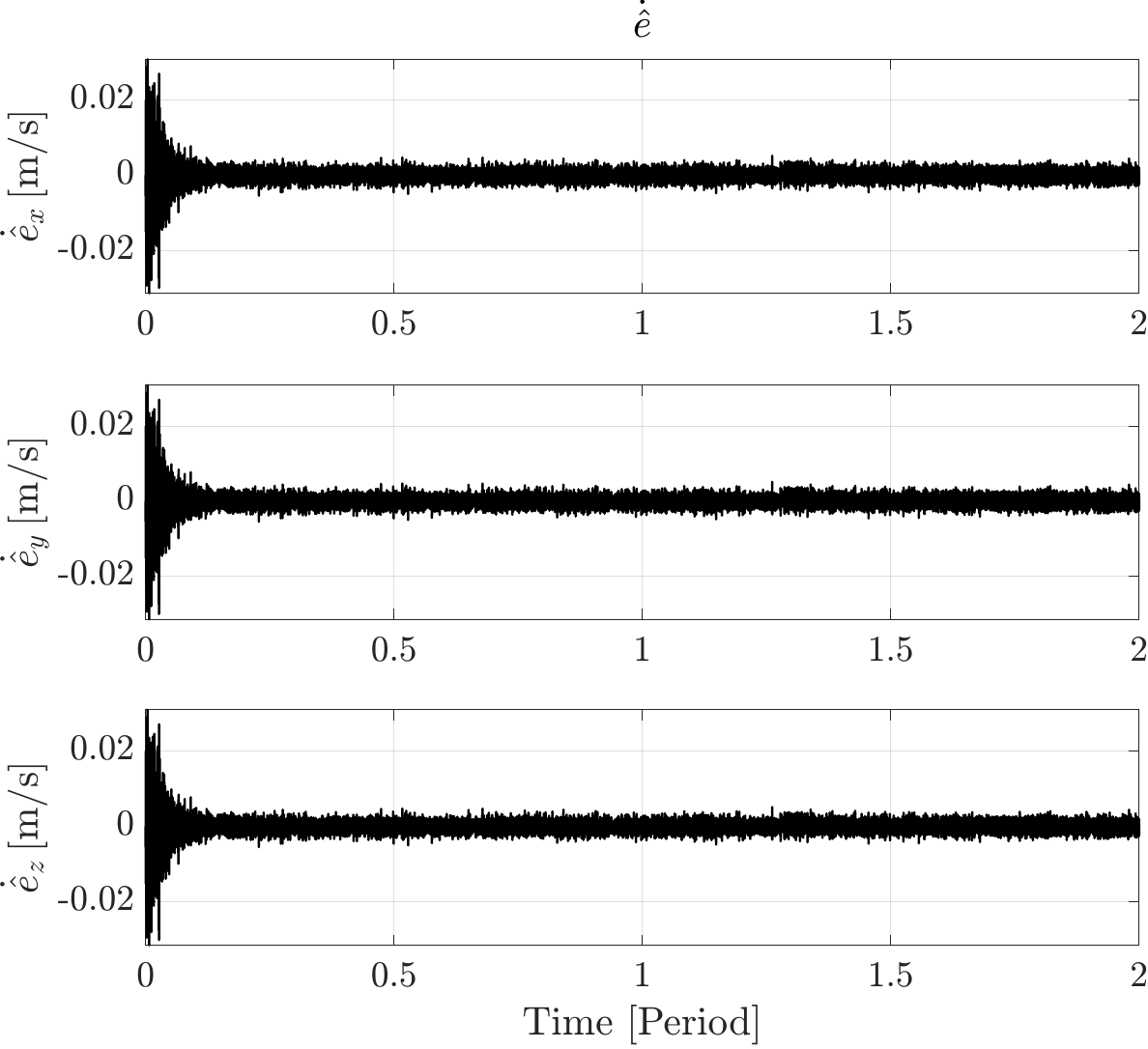}
\caption{Estimated tracking velocity error
$\dot{\hat{\boldsymbol{e}}}$.}
\label{fig:nav_ehatdot}
\end{subfigure}
\caption{Estimated tracking errors under navigation error.}
\label{fig:nav_estimated_tracking_errors}
\end{figure}

\subsubsection*{Actual tracking errors}

The actual tracking-error bounds combine the ASMC estimated-error bounds
with the SMDO state-estimation contribution:
\begin{align}
\limsup_{t\to\infty}
\|\boldsymbol{e}(t)\|_\infty
&\leq
0.3422\,\mathrm{m},
\label{eq:nav_actual_e_bound}\\
\limsup_{t\to\infty}
\|\dot{\boldsymbol{e}}(t)\|_\infty
&\leq
1.3687\,\mathrm{m/s}.
\label{eq:nav_actual_edot_bound}
\end{align}
These values are conservative, as they are obtained by directly summing
the SMDO ultimate bounds with the estimated ASMC tracking bounds. The
simulation confirms a more structured behavior: the navigation error
manifests primarily as noise in the SMDO estimation and compensation
channels, while the ASMC tracking variables remain essentially unaffected,
as shown in Fig.~\ref{fig:actual_tracking_errors_noise}.

\begin{figure}[H]
  \centering
  \begin{subfigure}{0.48\linewidth}
    \centering
    \includegraphics[width=\linewidth]{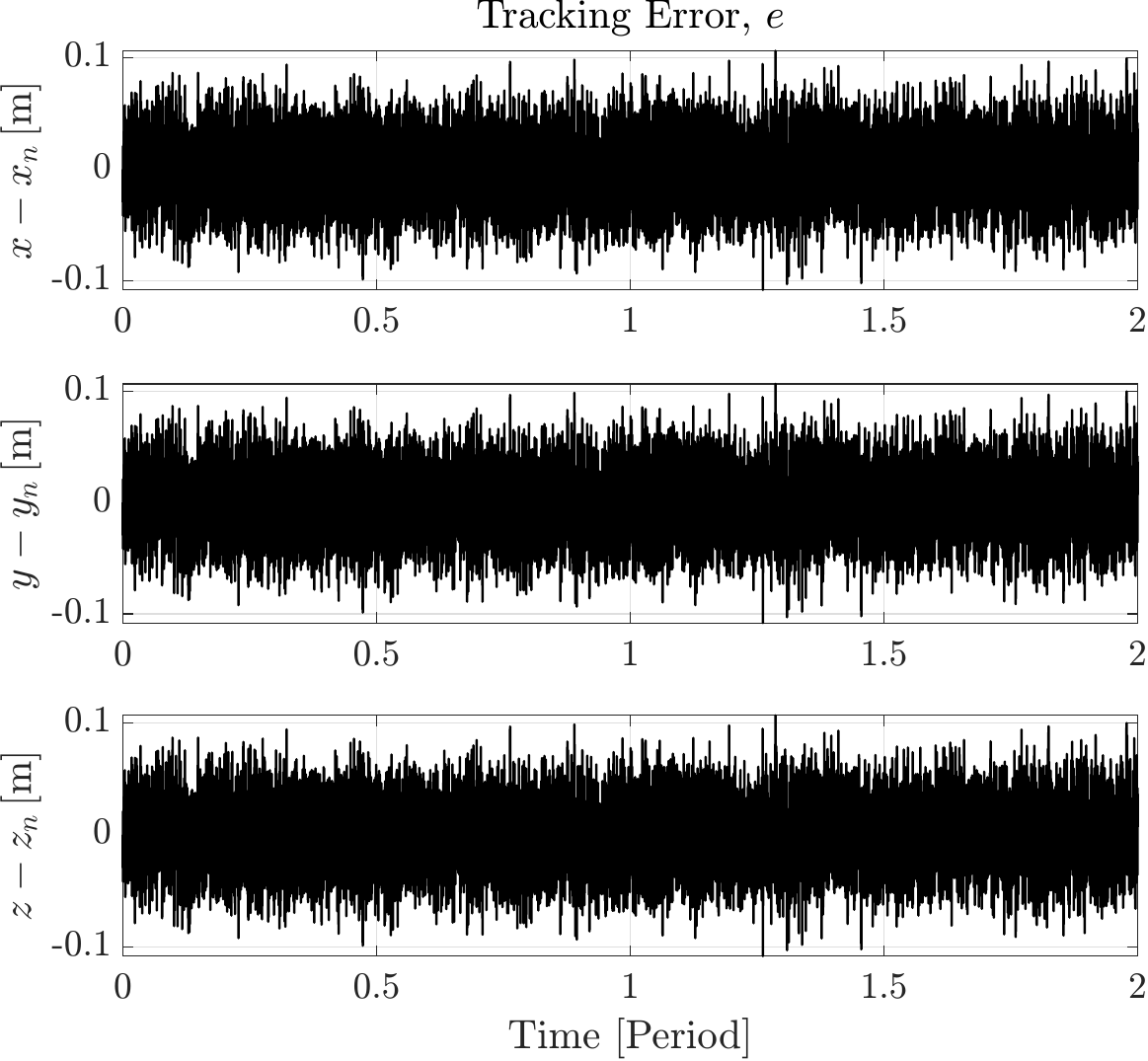}
    \caption{Actual tracking error
    $\boldsymbol{e} = \boldsymbol{q} - \boldsymbol{q}_n$.}
    \label{fig:Tracking_Error_nav}
  \end{subfigure}
  \hfill
  \begin{subfigure}{0.48\linewidth}
    \centering
    \includegraphics[width=\linewidth]{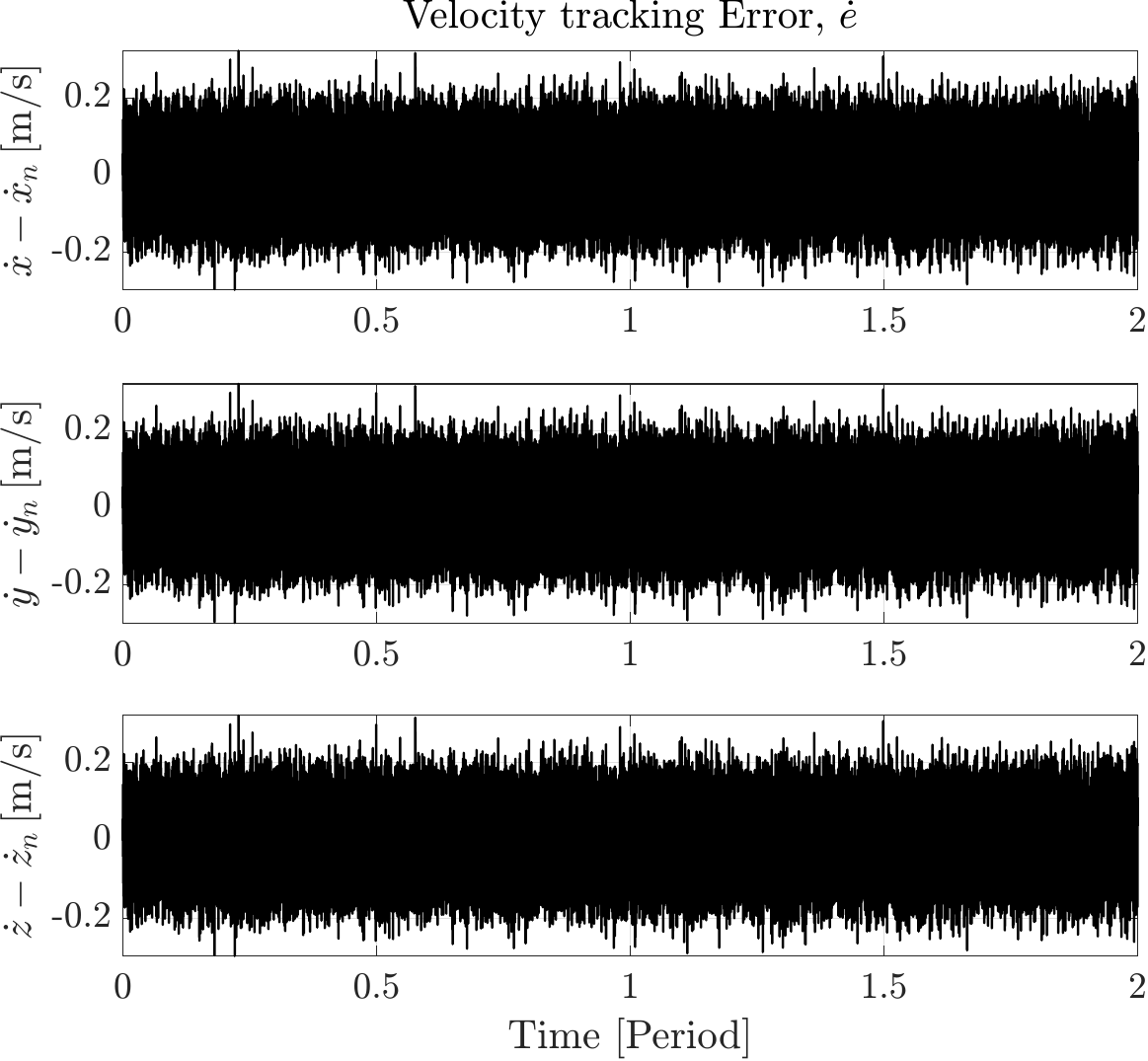}
    \caption{Actual velocity tracking error
    $\dot{\boldsymbol{e}} = \dot{\boldsymbol{q}} - \dot{\boldsymbol{q}}_n$.}
    \label{fig:edot_nav}
  \end{subfigure}
  \caption{Actual tracking errors under navigation error.}
  \label{fig:actual_tracking_errors_noise}
\end{figure}

Overall, this case demonstrates a structural advantage of the proposed
output-feedback architecture: the SMDO confines the effect of navigation
error to the observer layer, preventing its propagation into the tracking
control loop. Consequently, the ASMC achieves tracking performance that is
largely decoupled from navigation quality, a property that is directly
attributable to the estimated-state feedback structure.

\section{Conclusion}
\label{sec:conclusion}
This paper presented a two-phase relative orbit control framework for spacecraft formation reconfiguration by combining analytic minimum-energy PCO insertion with an auxiliary-state second-order SMDO-ASMC feedback architecture. In the reaching phase, the PCO entry phase was treated as an explicit design variable rather than a prescribed terminal condition. By parameterizing the transfer cost with respect to the PCO phase angle, the stationarity condition was reduced to a quartic polynomial, whose real roots provide all candidate entry phases. This enables analytic selection of the minimum-energy PCO insertion point without numerical phase sweeping.

For robust closed-loop operation, the nominal transfer and PCO references were augmented by an estimated-state feedback structure consisting of a second-order sliding mode disturbance observer and a second-order adaptive sliding mode controller. The SMDO reduces the effective lumped disturbance to a bounded residual, while the ASMC regulates the estimated tracking error with respect to the phase-dependent nominal reference. The auxiliary-state formulation induces a triangular observer-controller structure, in which the observer error and estimated tracking error are bounded separately and the actual tracking error is recovered algebraically from their sum. The resulting analysis establishes practical ultimate boundedness, with the final tracking-error bound determined by the SMDO and ASMC boundary-layer parameters.

A surrogate derivative generator was also introduced to implement the second-order sliding surfaces without direct finite differencing or algebraic-loop computation. The parallel first-order surrogate branch generates loop-free acceleration-level signals only for derivative evaluation, and its commands are not applied to the actual closed-loop plant.

Numerical simulations demonstrated that the proposed framework achieves energy-efficient PCO insertion, accurate tracking during both the reaching and PCO phases, bounded adaptive gains, and bounded control effort. Simulations with navigation-error-contaminated measurements further showed that the estimated-state architecture limits the direct influence of measurement errors on the ASMC tracking loop. Future work will extend the proposed method to higher-fidelity orbital dynamics, navigation-filter coupling, and hardware-in-the-loop validation.

\appendix

\section{Derivation of Cost-Optimal Nominal Trajectory}
\label{sec:appendix_nominal_trajectory}

This appendix documents two items that support
Section~\ref{sec:system_modeling_nominal_design}: the constant
skew-symmetry of $C_{\Phi}$, which is asserted without proof in
Section~\ref{subsec:nominal_reaching_design}, and the complete derivation
of the optimal PCO entry phase, including the handling of degenerate cases
not covered in the main text.

\subsection{Properties of \texorpdfstring{$C_{\Phi}$}{C\_Phi}}
\label{subsec:appendix_Cphi}

Substituting the explicit CW state transition matrix
(\cite{vallado2001fundamentals}) into Eq.~\eqref{eq:Cphi_def} shows by
direct computation that $C_{\Phi}$ is independent of $t$ and satisfies
$C_{\Phi}=-C_{\Phi}^{\top}$. Both properties are used in
Section~\ref{subsec:nominal_reaching_design} and the derivation of
Eqs.~\eqref{eq:ur}--\eqref{eq:xi_n} follows \cite{cho2009analytic}.

\subsection{Candidate Selection and Special Cases for the Optimal Entry Phase}
\label{subsec:appendix_cases}

The Weierstrass substitution $\tau = \tan(\phi_0/2)$ applied to
Eq.~\eqref{eq:stationarity} yields the quartic
Eq.~\eqref{eq:phase_quartic_polynomial}, whose finite real roots map back
to phase candidates via
\begin{equation}
\phi_{0,i}=\operatorname{mod}\!\left(2\arctan(\tau_i),\,2\pi\right).
\end{equation}
The point $\phi_0=\pi$ corresponds to $\tau\to\infty$ and is therefore not
represented by any finite root. Substituting $\phi_0=\pi$ directly into
Eq.~\eqref{eq:stationarity} gives
\begin{equation}
A_q\sin 2\pi+2B_q\cos 2\pi+2C_q\cos\pi-2D_q\sin\pi
=2(B_q-C_q),
\end{equation}
so $\phi_0=\pi$ is a stationary candidate if and only if $B_q-C_q=0$.
When all four coefficients $A_q,B_q,C_q,D_q$ vanish simultaneously,
$q_\phi$ is identically zero and the cost is constant over $[0,2\pi)$.
The three resulting cases are summarized in
Table~\ref{tab:phi0_case_summary}.

\begin{table}[H]
\centering
\small
\caption{Case-dependent candidate sets for selecting the optimal PCO entry
phase $\phi_0^\star$.}
\label{tab:phi0_case_summary}
\begin{tabular}{p{0.25\textwidth}p{0.43\textwidth}p{0.22\textwidth}}
\toprule
Case & Candidate set $\mathcal{C}$ & Selection rule \\
\midrule
$B_q-C_q\neq0$ &
$\displaystyle
\left\{
\operatorname{mod}(2\arctan\tau_i,2\pi)
\mid \tau_i\in\mathbb{R},\ q_\phi(\tau_i)=0
\right\}$ &
$\displaystyle \phi_0^\star=\arg\min_{\phi_0\in\mathcal{C}}J(\phi_0)$ \\
\addlinespace
$B_q-C_q=0$,\ $q_\phi\not\equiv0$ &
$\displaystyle
\left\{
\operatorname{mod}(2\arctan\tau_i,2\pi)
\mid \tau_i\in\mathbb{R},\ q_\phi(\tau_i)=0
\right\}\cup\{\pi\}$ &
$\displaystyle \phi_0^\star=\arg\min_{\phi_0\in\mathcal{C}}J(\phi_0)$ \\
\addlinespace
$q_\phi\equiv0$ &
$\mathcal{C}=[0,2\pi)$ &
$J(\phi_0)=\mathrm{const}$; any $\phi_0$ is optimal \\
\bottomrule
\end{tabular}
\end{table}

In every non-degenerate case the polynomial roots are only stationary
candidates, not optimality certificates; the global minimizer is
identified by directly evaluating $J(\phi_0)$ over $\mathcal{C}$ as in
Eq.~\eqref{eq:phi0_star}.

\graphicspath{{./figures_diff/}}

\section{Comparison of Differentiation Methods}
\label{sec:comparison}

\subsection*{Finite Differencing}

The simplest option is the Euler approximation:
\begin{equation}
    \widehat{\dot{u}}(t) = \frac{u(t) - u(T_{\mathrm{prev}})}{t -
    T_{\mathrm{prev}}}
    \bigg|_{t > T_{\mathrm{prev}}},
    \label{eq:euler}
\end{equation}
where $\tau := t - T_{\mathrm{prev}}$ is the sampling interval.
If each measurement carries independent noise with variance $\sigma^2$,
error propagation yields $\mathrm{Var}(\widehat{\dot{u}}) = 2\sigma^2/\tau^2$,
so high-frequency measurement noise is amplified by the differentiation step.
When the estimate of $\dot{\hat{\boldsymbol{\sigma}}}_1$ enters the control law
through $1/\delta$, even modest noise levels can produce severe chattering.

\subsection*{Levant's Differentiator}

Levant (\cite{levant1998robust}) proposed a robust exact differentiator based on second-order sliding modes. Let $f(t)$ be the signal to be differentiated, which is assumed to consist of a base signal with a derivative whose Lipschitz constant is bounded by a known constant $C>0$, plus a bounded measurement noise. The differentiator introduces an auxiliary state $z_0(t)$ that tracks $f(t)$, and produces an estimate $z_1(t)$ of its time derivative $\dot{f}(t)$:
\begin{align}
    \dot{z}_0 &= -\lambda|z_0 - f(t)|^{1/2}\operatorname{sgn}(z_0 - f(t)) + z_1,
    \label{eq:levant_z0} \\
    \dot{z}_1 &= -\alpha\operatorname{sgn}(z_0 - f(t)),
    \label{eq:levant_z1}
\end{align}
where $\lambda>0$ and $\alpha>0$ are design parameters satisfying $\alpha > C$
and $\lambda^2 \geq 4C(\alpha+C)/(\alpha-C)^2$. In the absence of measurement
noise, $z_0\to f(t)$ and $z_1\to\dot{f}(t)$ in finite time; under bounded noise
of magnitude $\varepsilon$, the differentiation error satisfies
$|z_1-\dot{f}|\leq\mu\sqrt{C\varepsilon}$ for some constant $\mu>0$.

In the present context, $f(t)$ is identified with $\hat{\boldsymbol{\sigma}}_1(t)$
and $z_1(t)$ serves as the estimate of $\dot{\hat{\boldsymbol{\sigma}}}_1$
required by the second-order ASMC surface
Eq.~\eqref{eq:second_order_tracking_surface}, applied componentwise.
The parameters $\lambda$ and $\alpha$ must be tuned to the scale of
$\dot{\hat{\boldsymbol{\sigma}}}_1$, whose magnitude is not known a priori;
mistuning degrades convergence or induces oscillation.

\subsection*{Proposed Surrogate Derivative Generator}

The proposed method is given by
Eqs.~\eqref{eq:smdo_surrogate_derivative}--\eqref{eq:asmc_surrogate_derivative},
with the first-order branch described in
Eqs.~\eqref{eq:first-order_branch}--\eqref{eq:L_first_order}.
No additional tuning parameters are required, and measurement noise is not
further amplified.

\subsection{Noise-Free Comparison}

Figures~\ref{fig:finite_diff}--\ref{fig:noisefreetotal} compare tracking error,
disturbance estimation error, and total control input for all three methods
without sensor noise. Under these conditions the proposed method outperforms
both alternatives by roughly one order of magnitude in tracking accuracy.

\begin{figure}[H]
    \centering
    \begin{subfigure}[b]{0.45\textwidth}
        \centering
        \includegraphics[width=\textwidth]{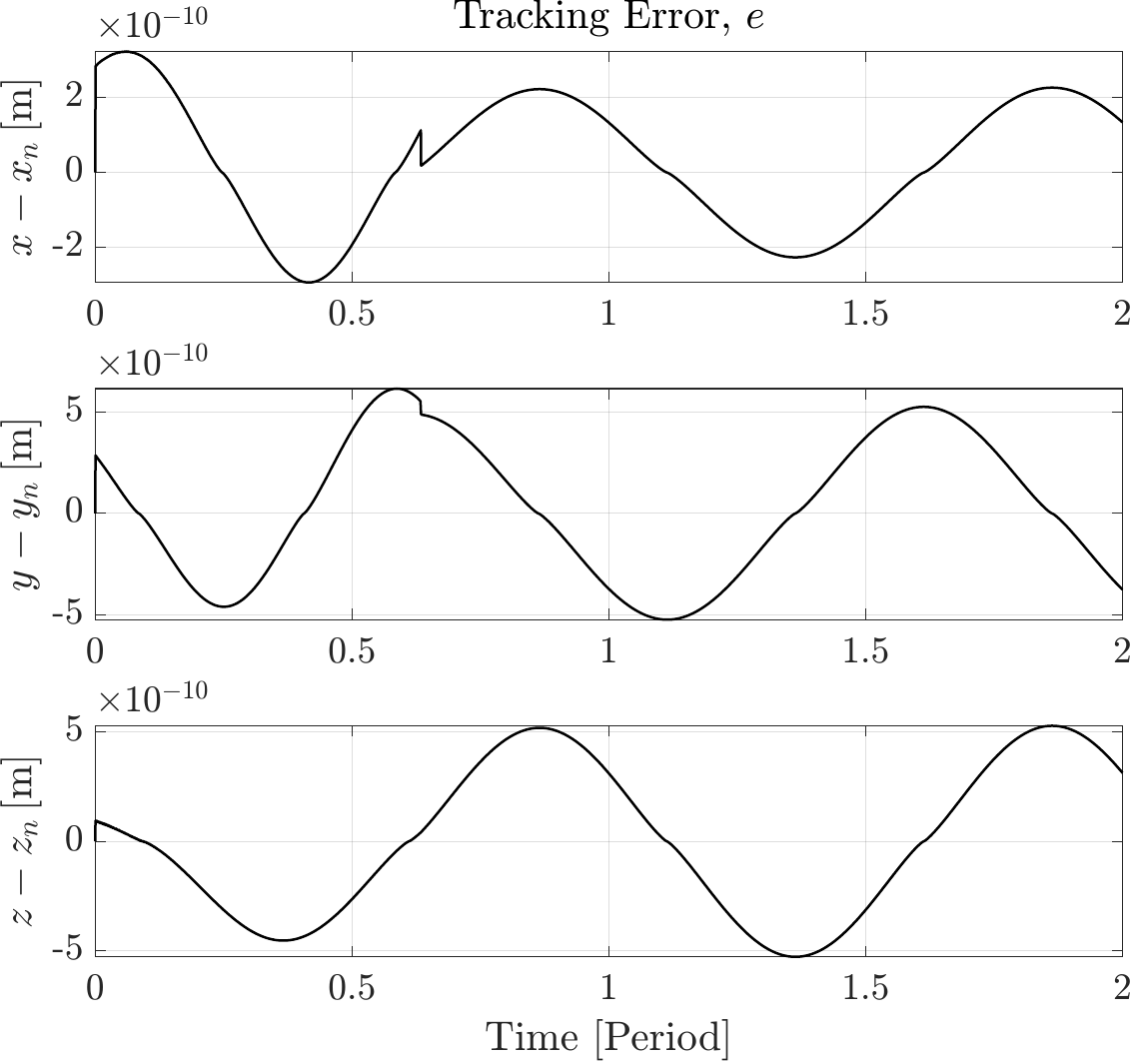}
        \caption{Tracking error}
    \end{subfigure}
    \hfill
    \begin{subfigure}[b]{0.45\textwidth}
        \centering
        \includegraphics[width=\textwidth]{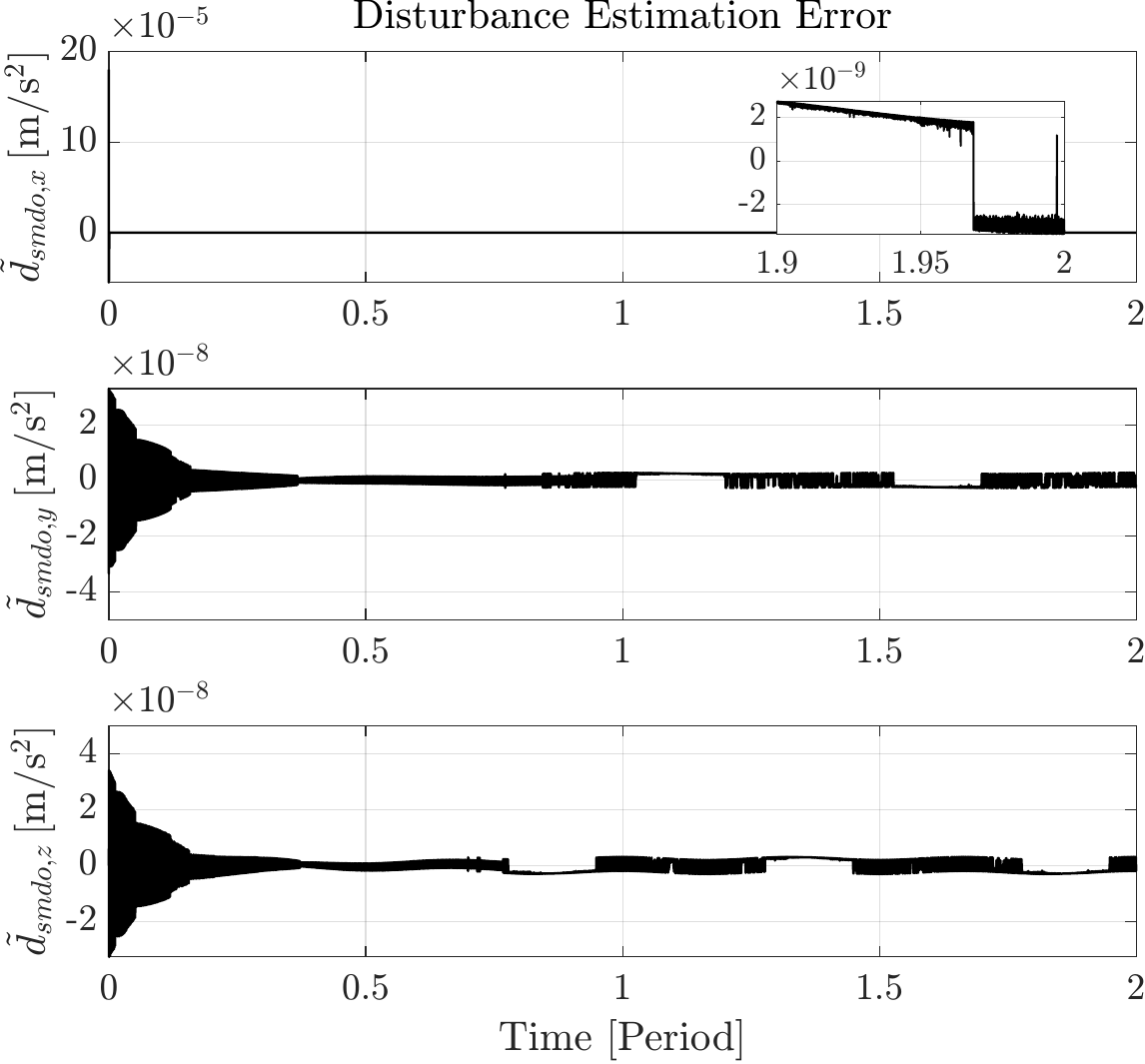}
        \caption{Disturbance estimation error}
    \end{subfigure}
    \caption{Noise-free performance: finite differencing.}
    \label{fig:finite_diff}
\end{figure}

\begin{figure}[H]
    \centering
    \begin{subfigure}[b]{0.45\textwidth}
        \centering
        \includegraphics[width=\textwidth]{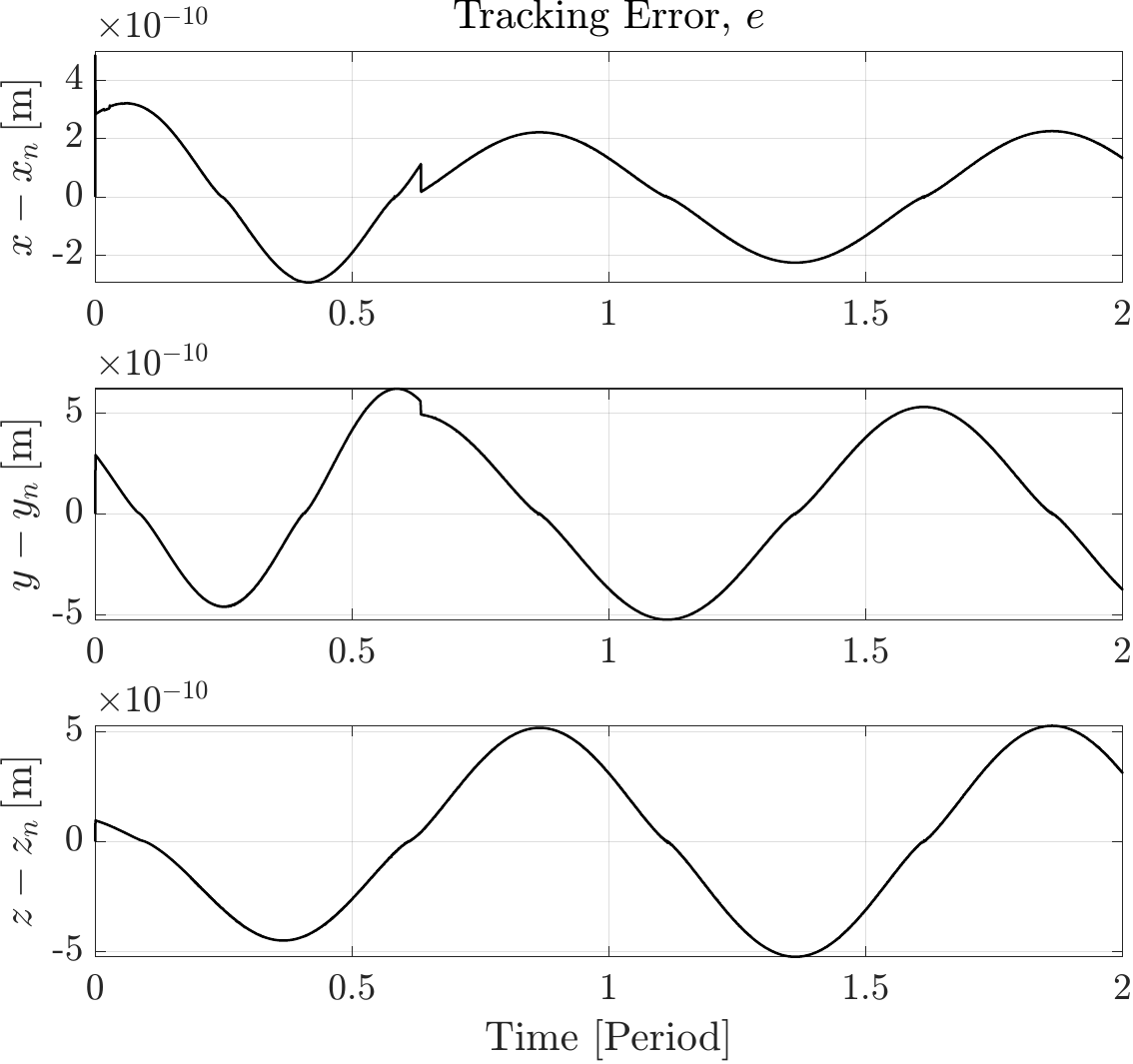}
        \caption{Tracking error}
    \end{subfigure}
    \hfill
    \begin{subfigure}[b]{0.45\textwidth}
        \centering
        \includegraphics[width=\textwidth]{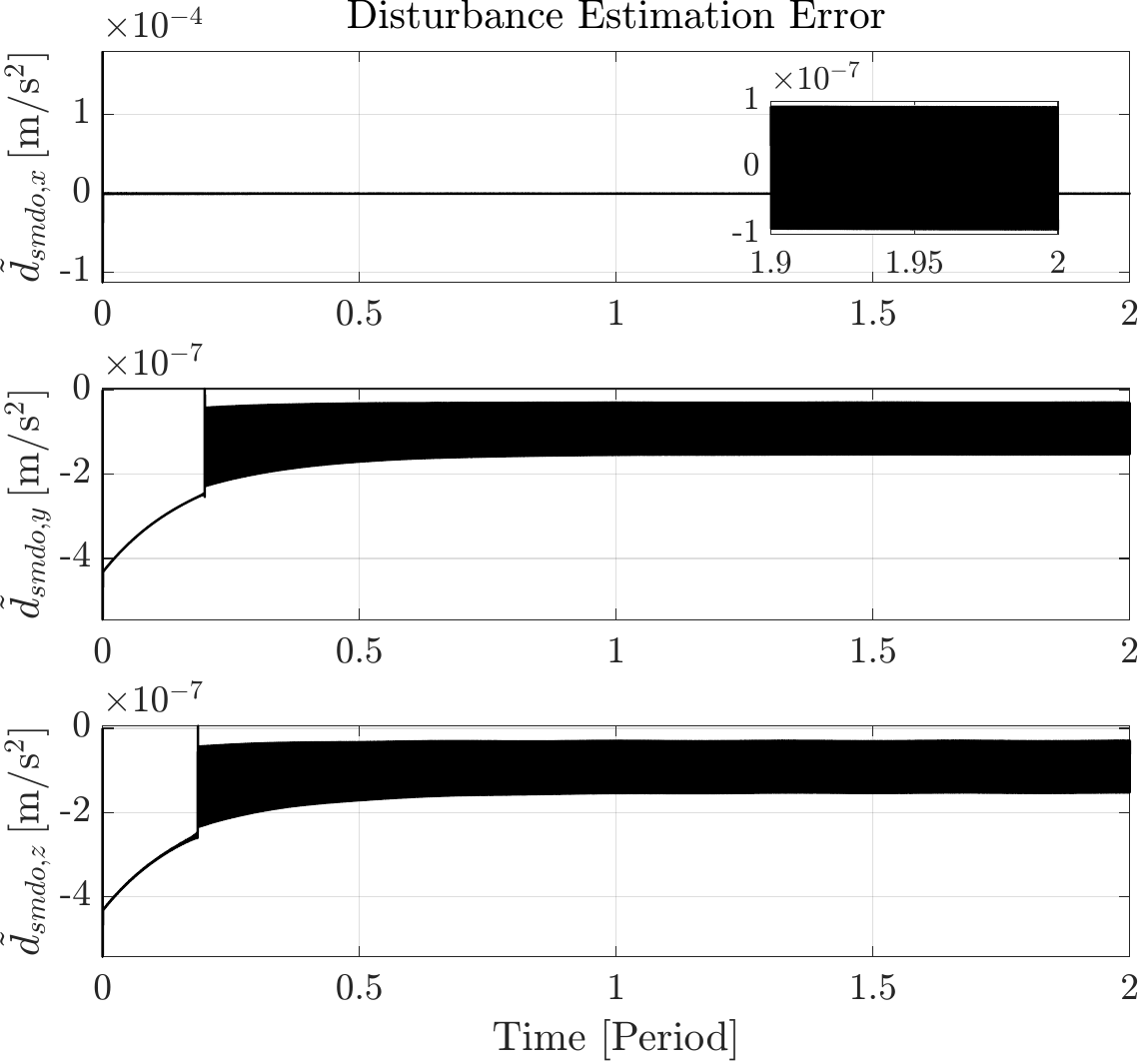}
        \caption{Disturbance estimation error}
    \end{subfigure}
    \caption{Noise-free performance: Levant's differentiator
             ($C = 0.5$, $\alpha = 0.55$, $\lambda = \sqrt{0.5}$).}
    \label{fig:levant_diff}
\end{figure}

\begin{figure}[H]
    \centering
    \begin{subfigure}[b]{0.45\textwidth}
        \centering
        \includegraphics[width=\textwidth]{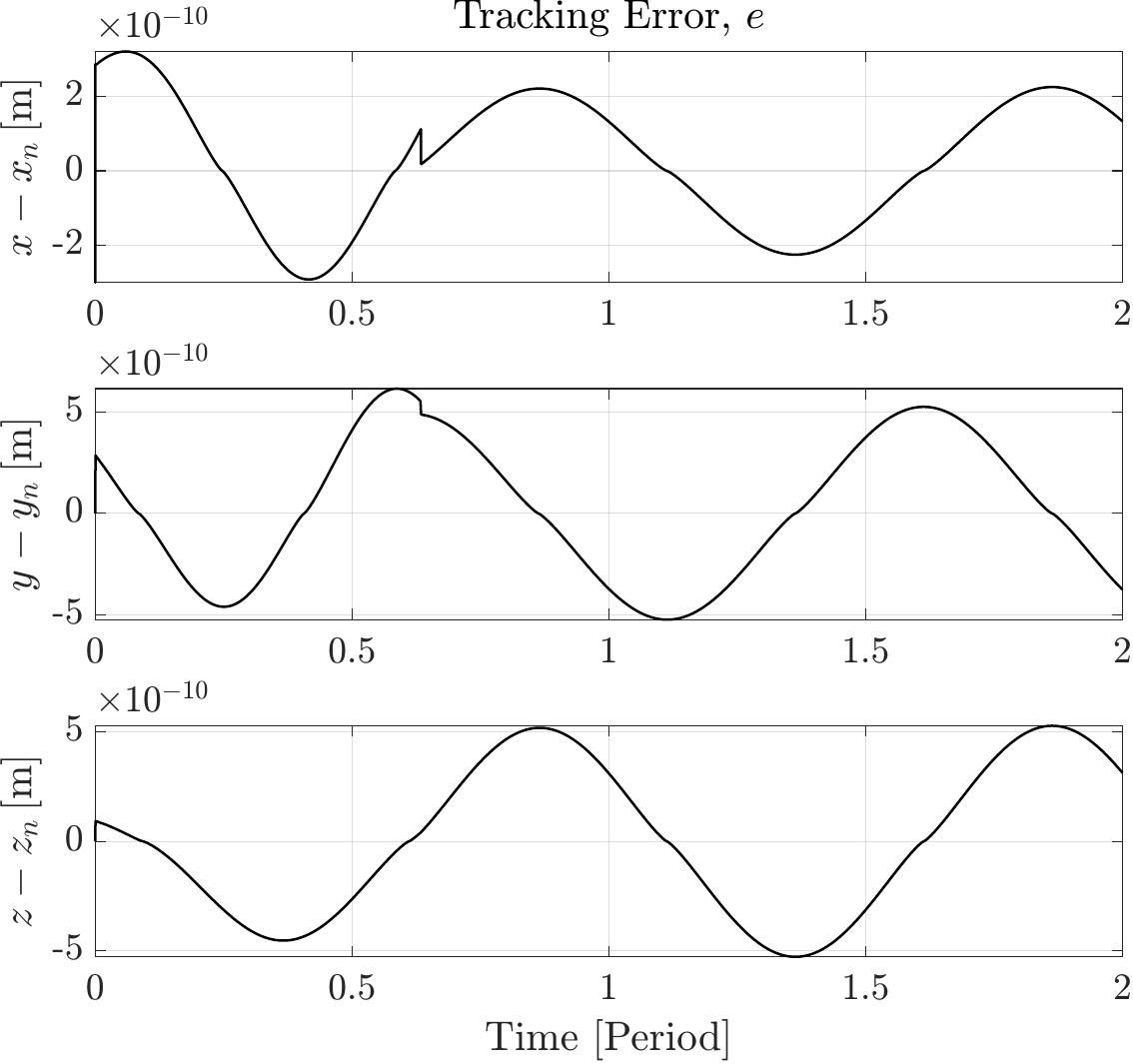}
        \caption{Tracking error}
    \end{subfigure}
    \hfill
    \begin{subfigure}[b]{0.45\textwidth}
        \centering
        \includegraphics[width=\textwidth]{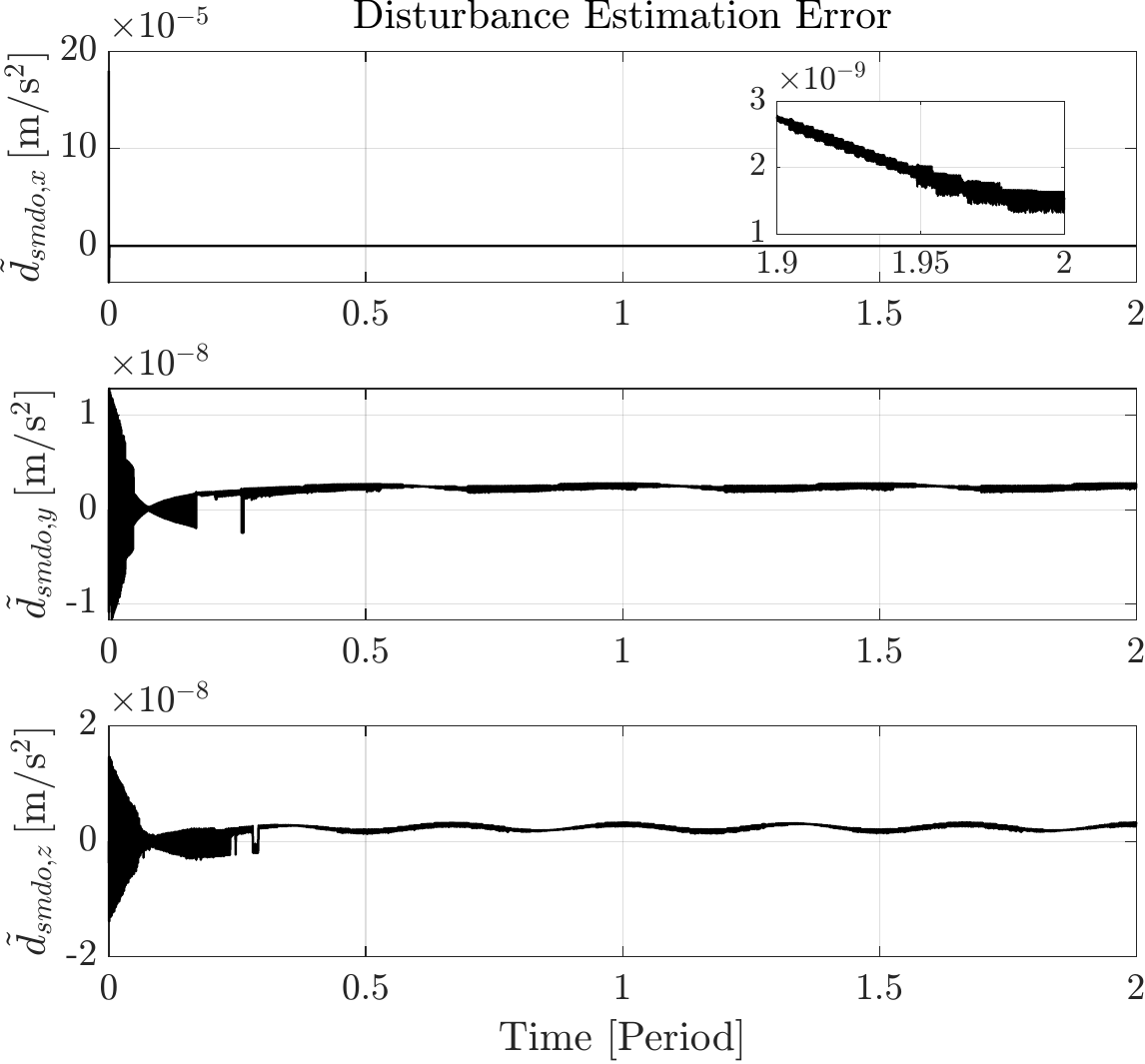}
        \caption{Disturbance estimation error}
    \end{subfigure}
    \caption{Noise-free performance: proposed surrogate derivative generator.}
    \label{fig:firstref_diff}
\end{figure}

\begin{figure}[H]
    \centering
    \begin{subfigure}[b]{0.3\textwidth}
        \centering
        \includegraphics[width=\textwidth]{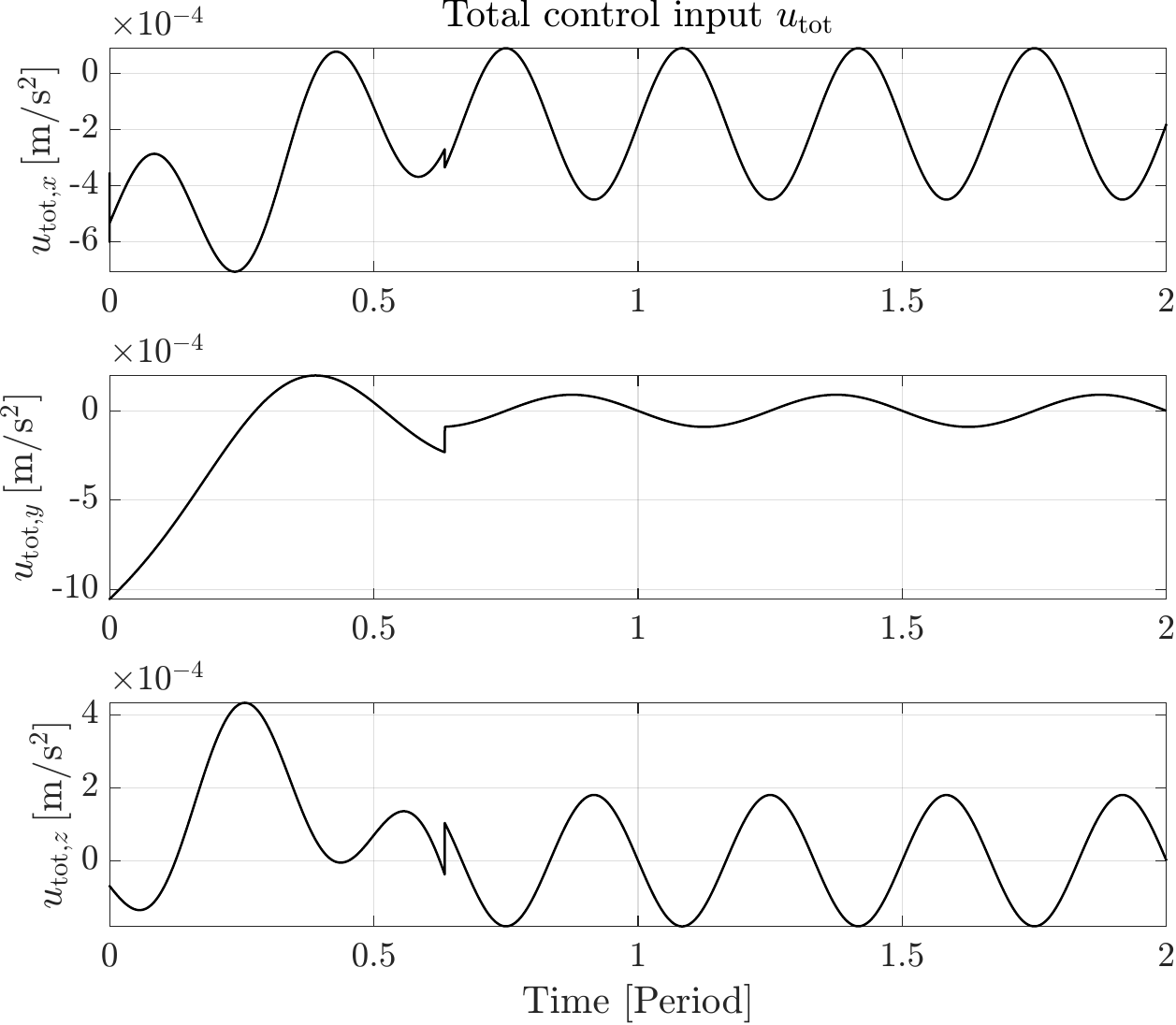}
        \caption{Finite differencing.}
    \end{subfigure}
    \hfill
    \begin{subfigure}[b]{0.3\textwidth}
        \centering
        \includegraphics[width=\textwidth]{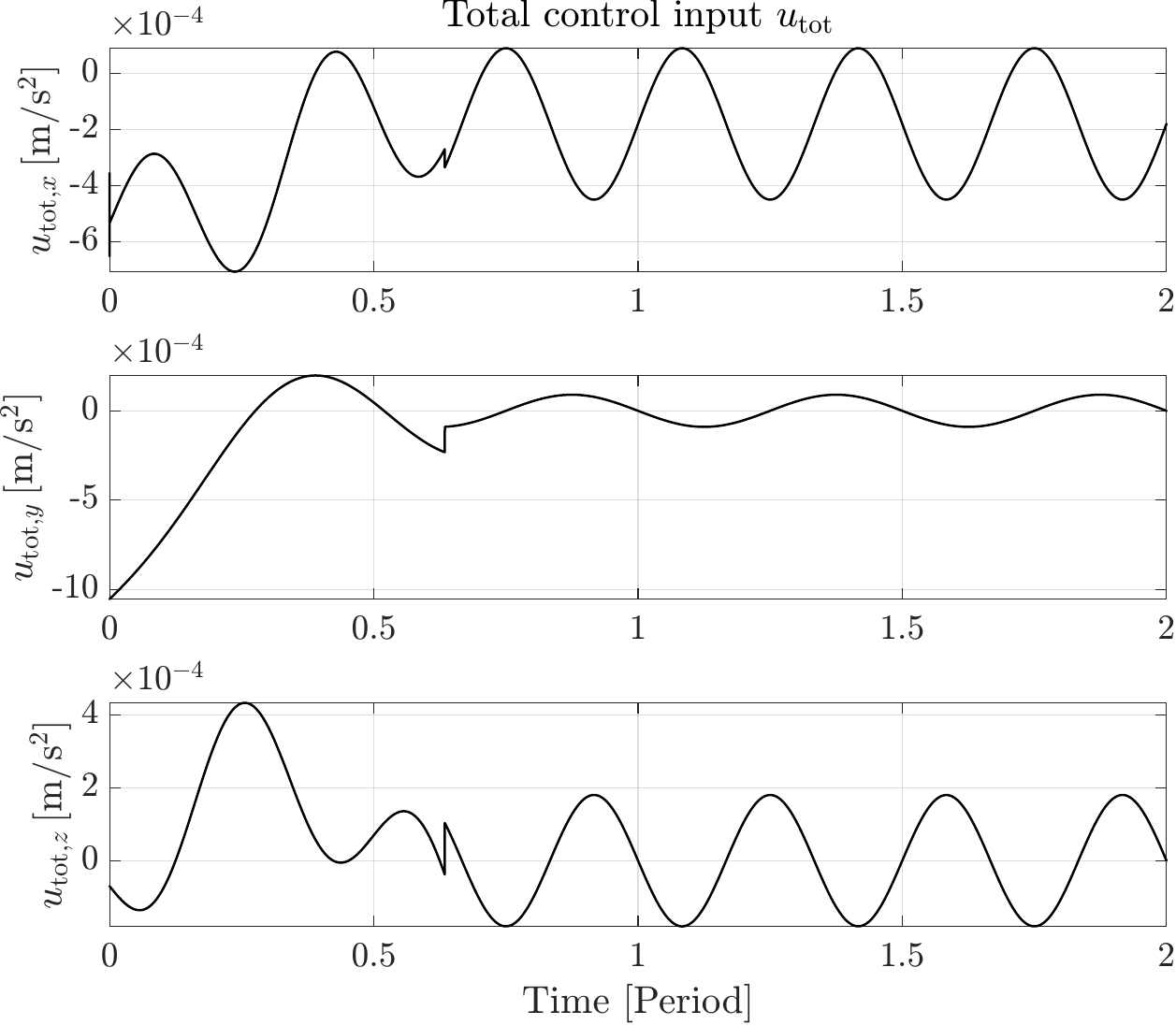}
        \caption{Levant's differentiator.}
    \end{subfigure}
    \hfill
    \begin{subfigure}[b]{0.3\textwidth}
        \centering
        \includegraphics[width=\textwidth]{u_tot.pdf}
        \caption{Proposed method.}
    \end{subfigure}
    \caption{Total control input (noise-free).}
    \label{fig:noisefreetotal}
\end{figure}

\subsection{Noisy Case}

Figures~\ref{fig:fd_noise}--\ref{fig:noisytotal} show tracking error,
disturbance estimation error, and total control input under the navigation
errors of Section~\ref{subsec:navigation_error_effect}.

\begin{figure}[H]
    \centering
    \begin{subfigure}[b]{0.45\textwidth}
        \centering
        \includegraphics[width=\textwidth]{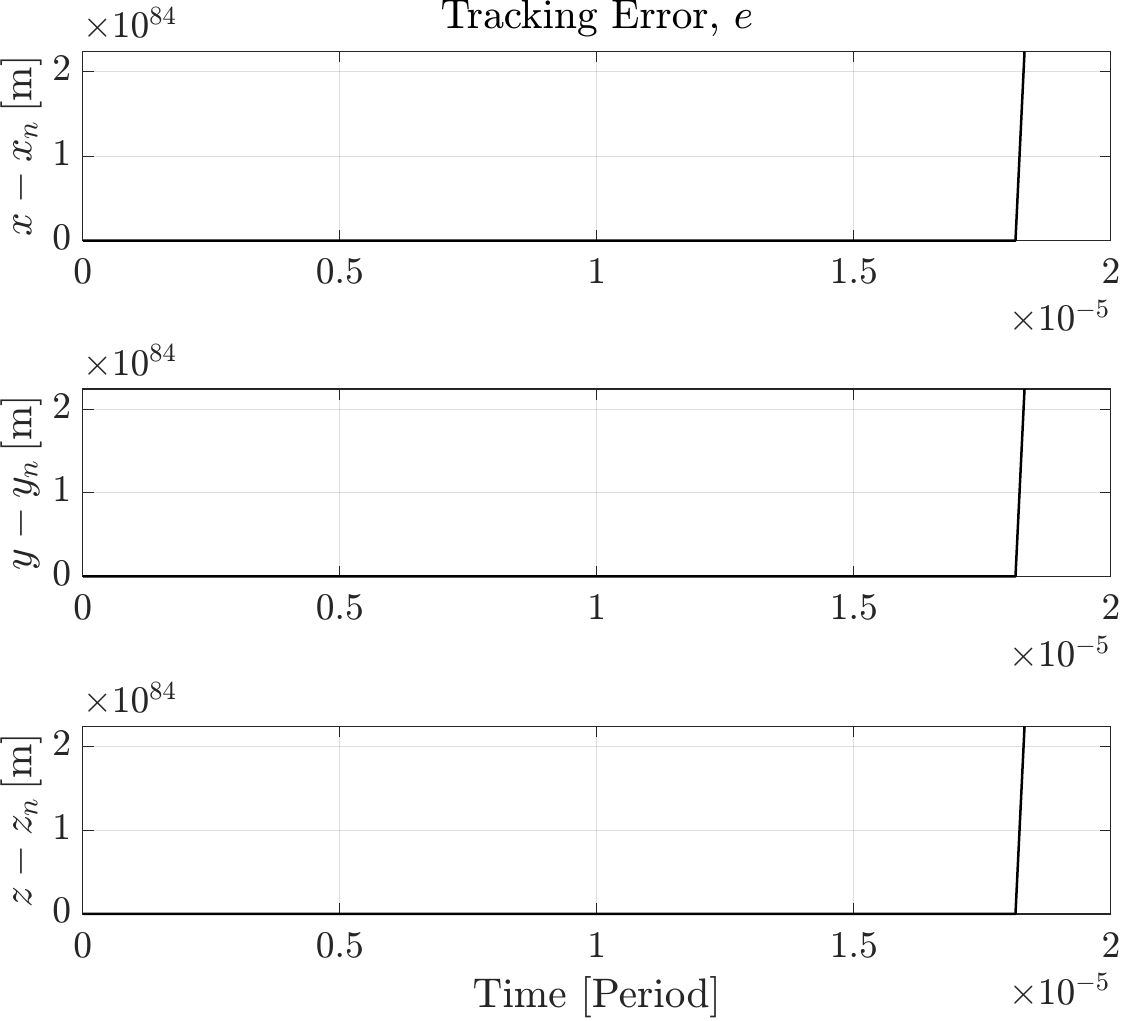}
        \caption{Tracking error}
    \end{subfigure}
    \hfill
    \begin{subfigure}[b]{0.45\textwidth}
        \centering
        \includegraphics[width=\textwidth]{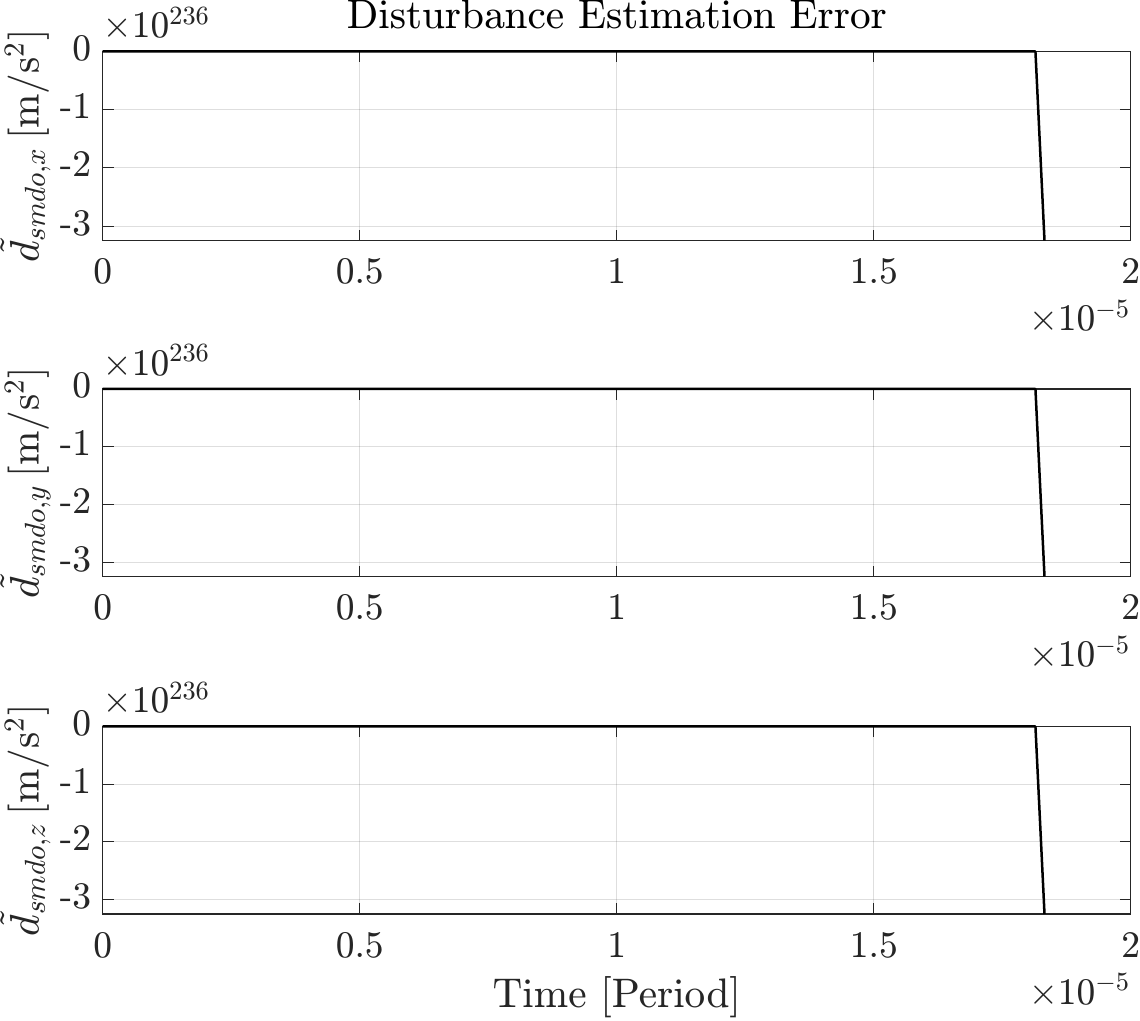}
        \caption{Disturbance estimation error}
    \end{subfigure}
    \caption{Noisy conditions: finite differencing.}
    \label{fig:fd_noise}
\end{figure}

\begin{figure}[H]
    \centering
    \begin{subfigure}[b]{0.45\textwidth}
        \centering
        \includegraphics[width=\textwidth]{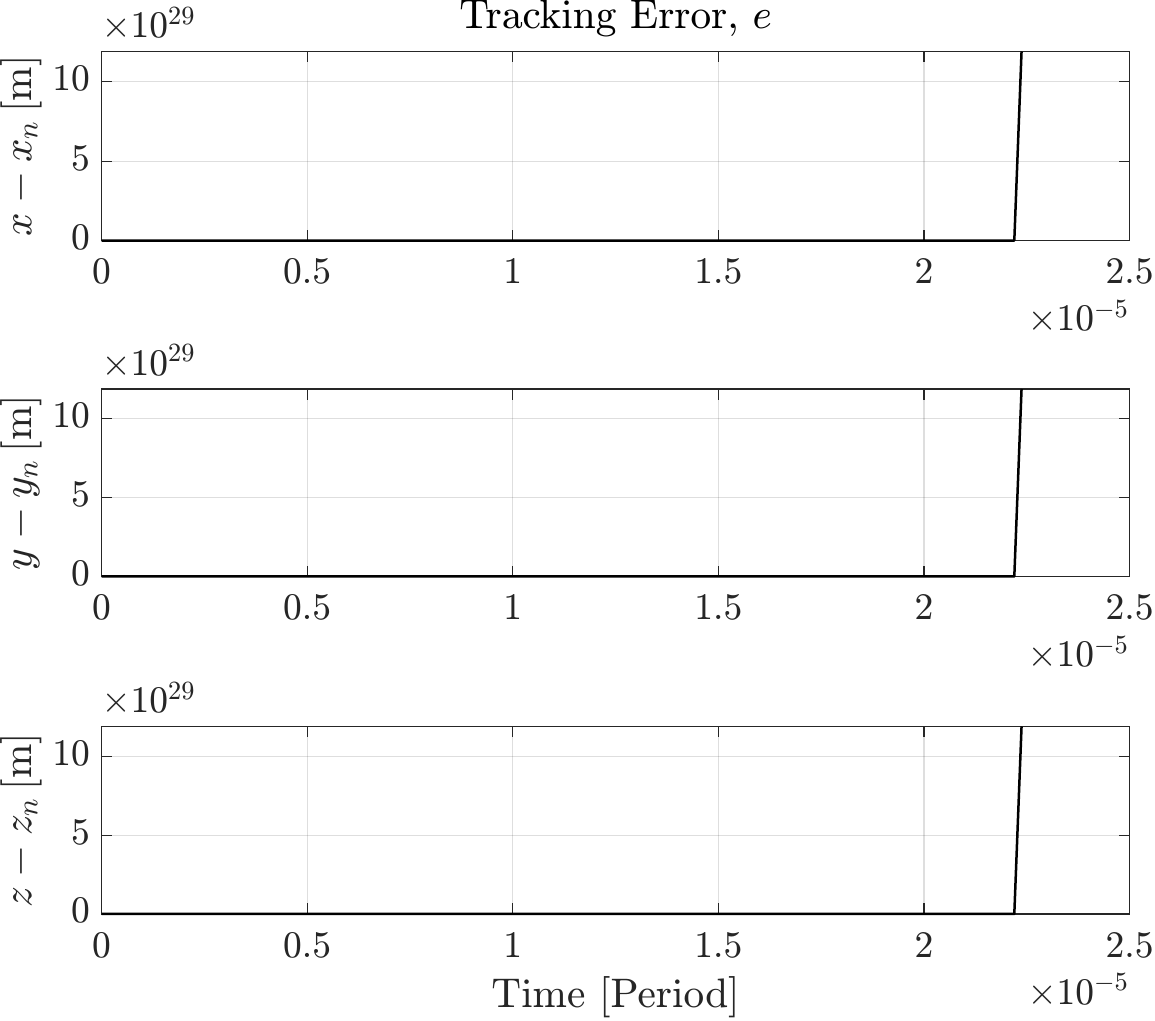}
        \caption{Tracking error}
    \end{subfigure}
    \hfill
    \begin{subfigure}[b]{0.45\textwidth}
        \centering
        \includegraphics[width=\textwidth]{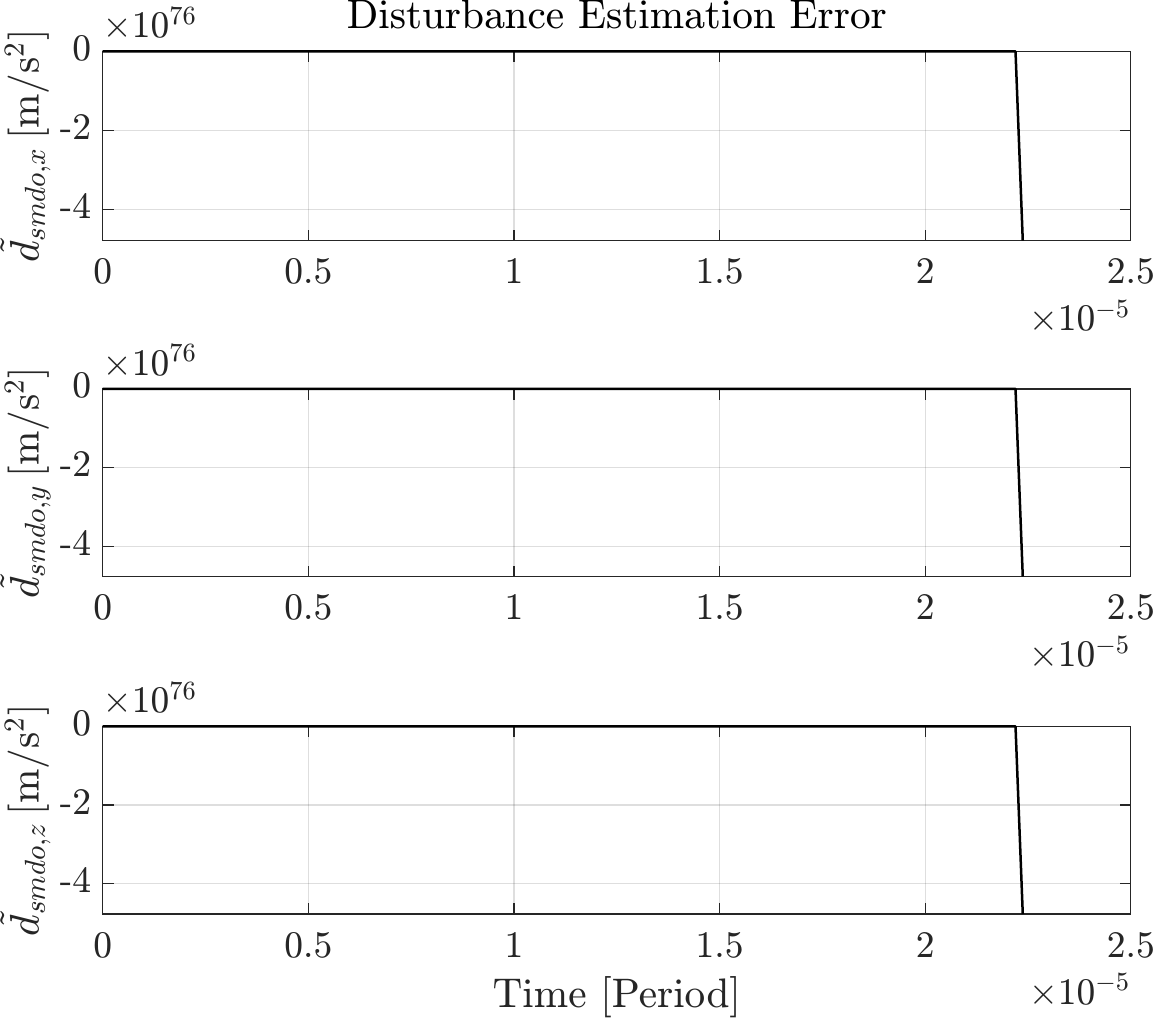}
        \caption{Disturbance estimation error}
    \end{subfigure}
    \caption{Noisy conditions: Levant's differentiator
             ($C = 0.5$, $\alpha = 0.55$, $\lambda = \sqrt{0.5}$).}
    \label{fig:lev_noise}
\end{figure}

\begin{figure}[H]
    \centering
    \begin{subfigure}[b]{0.45\textwidth}
        \centering
        \includegraphics[width=\textwidth]{Tracking_Error_nav.pdf}
        \caption{Tracking error}
    \end{subfigure}
    \hfill
    \begin{subfigure}[b]{0.45\textwidth}
        \centering
        \includegraphics[width=\textwidth]{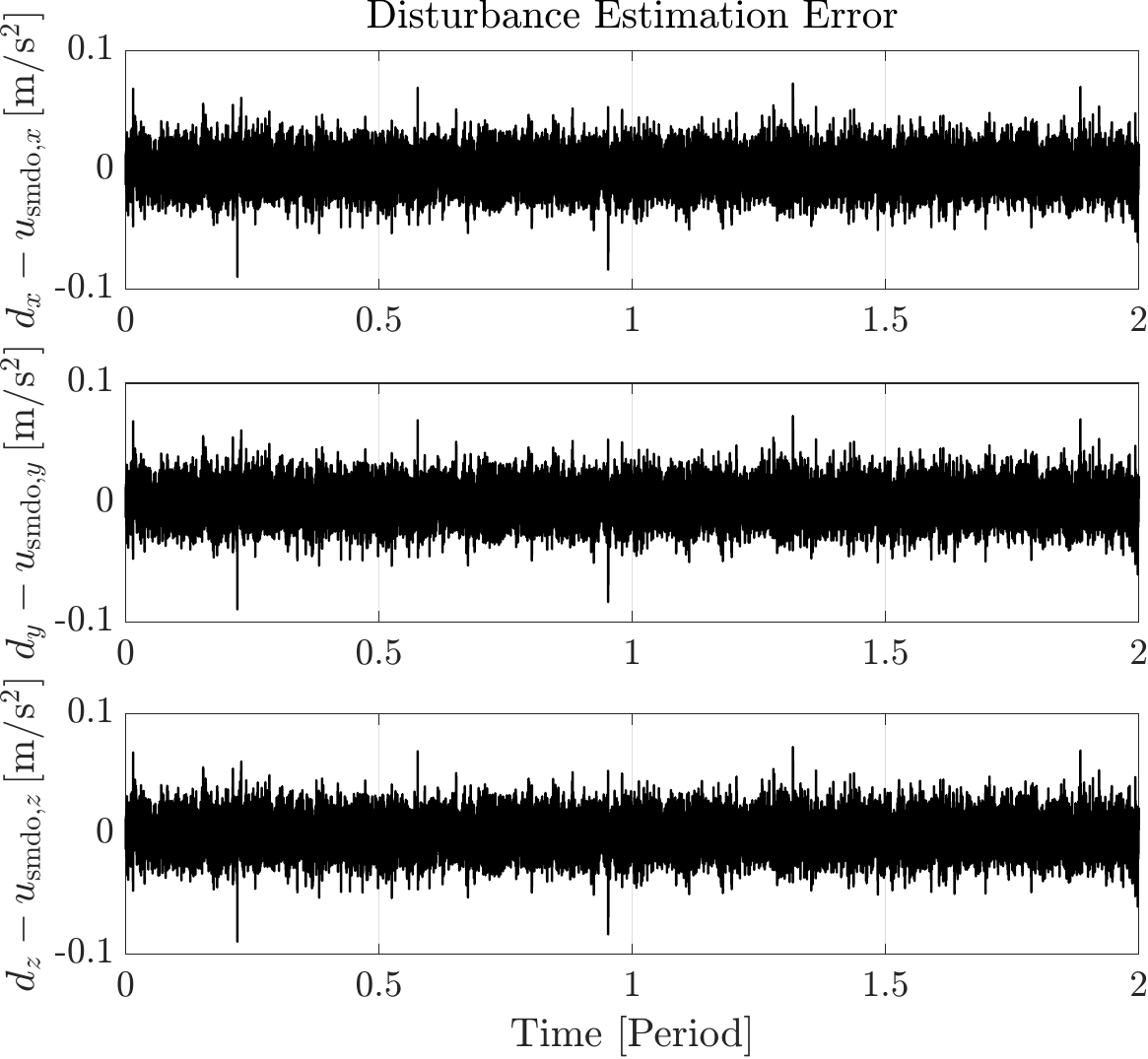}
        \caption{Disturbance estimation error}
    \end{subfigure}
    \caption{Noisy conditions: proposed surrogate derivative generator.}
    \label{fig:prop_noise}
\end{figure}

\begin{figure}[H]
    \centering
    \begin{subfigure}[b]{0.3\textwidth}
        \centering
        \includegraphics[width=\textwidth]{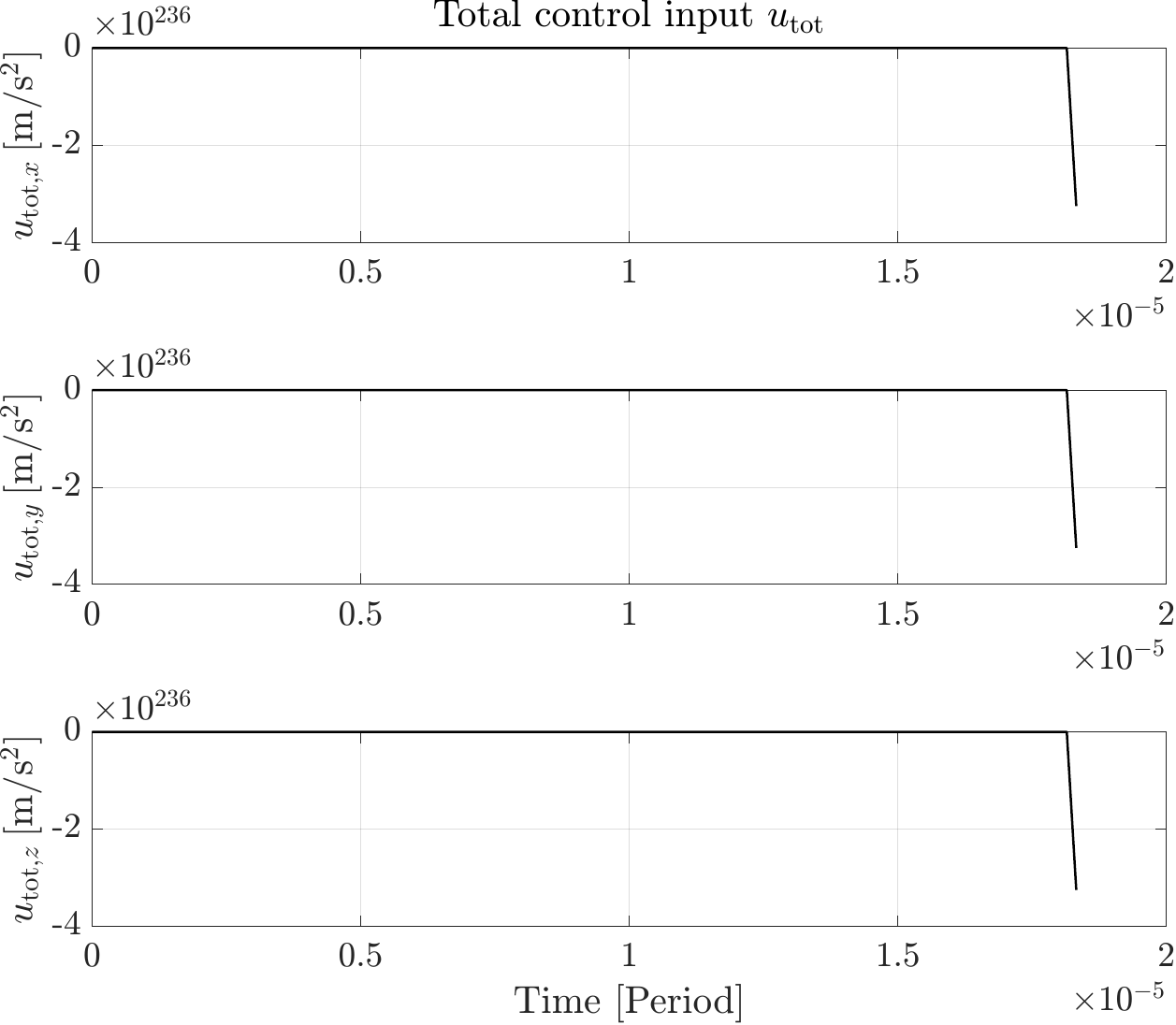}
        \caption{Finite differencing.}
    \end{subfigure}
    \hfill
    \begin{subfigure}[b]{0.3\textwidth}
        \centering
        \includegraphics[width=\textwidth]{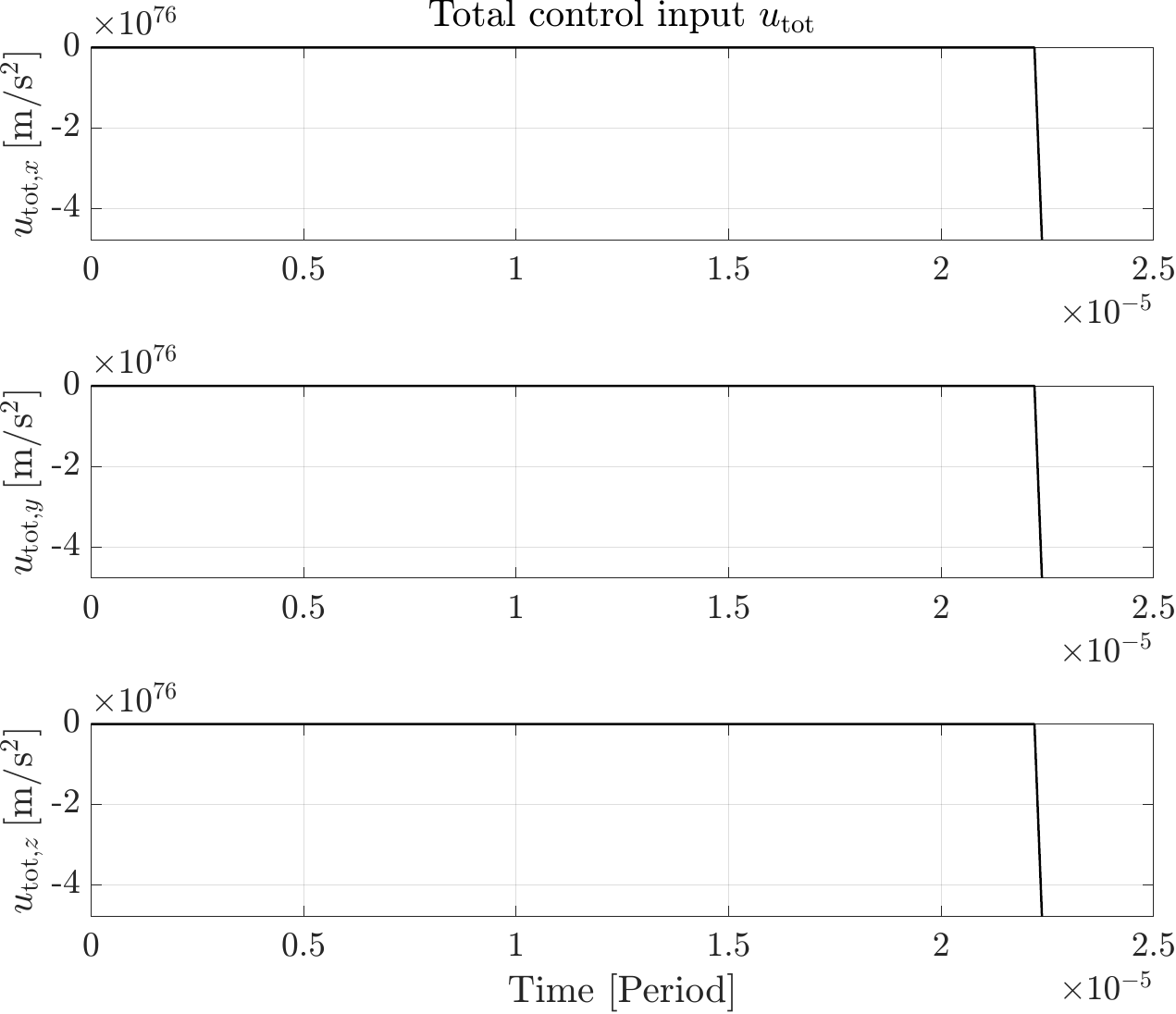}
        \caption{Levant's differentiator.}
    \end{subfigure}
    \hfill
    \begin{subfigure}[b]{0.3\textwidth}
        \centering
        \includegraphics[width=\textwidth]{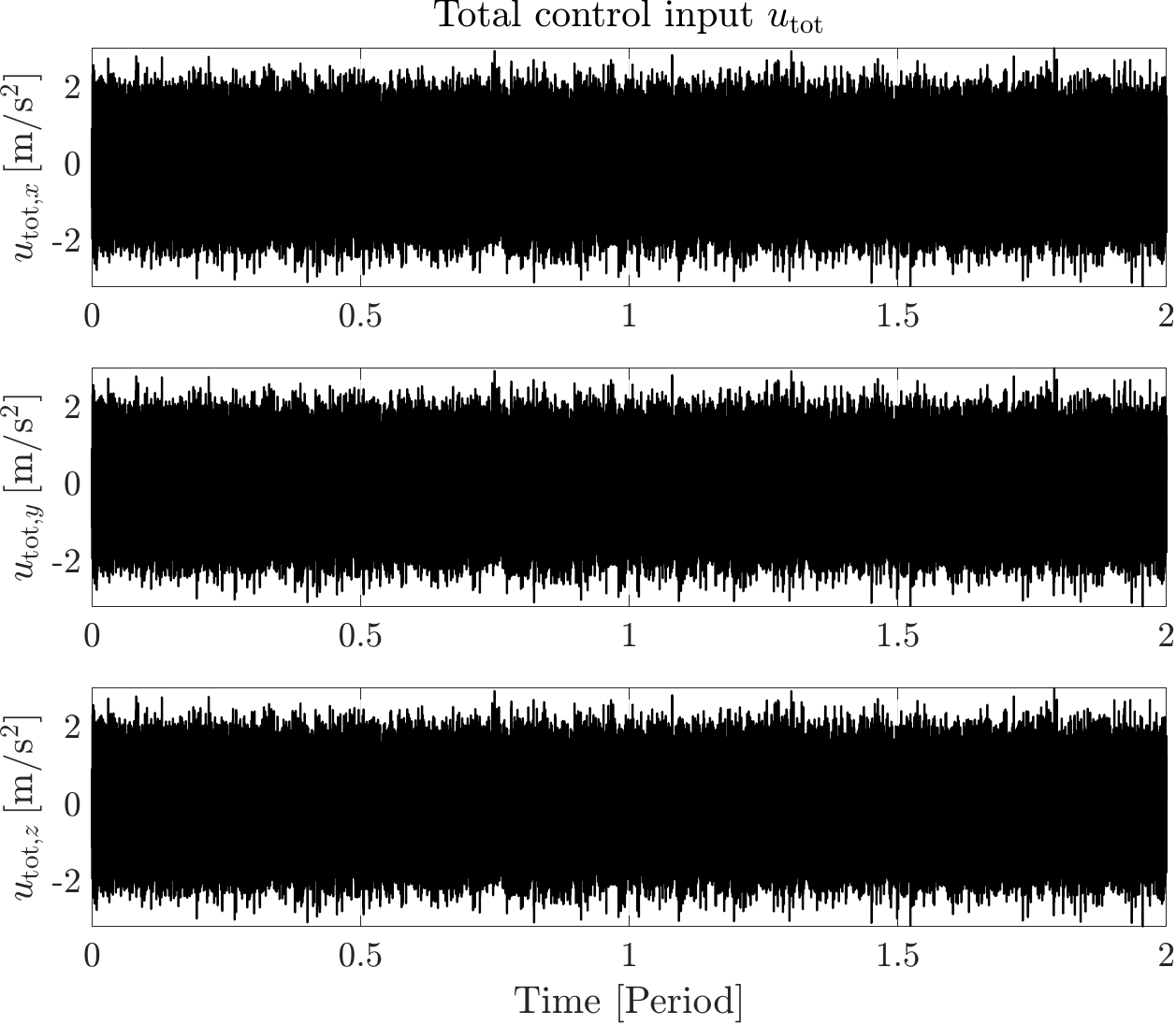}
        \caption{Proposed method.}
    \end{subfigure}
    \caption{Total control input (noisy conditions).}
    \label{fig:noisytotal}
\end{figure}

Under noise-free conditions, all three methods yield comparable tracking
performance; Levant's differentiator, however, exhibits the largest disturbance
estimation error among the three. Under navigation errors, both alternative
methods diverge and fail to produce bounded signals, whereas the proposed
surrogate derivative generator confines the tracking error within the prescribed
boundary layer. This robustness follows from the estimated-state-based
architecture combined with the first-order surrogate branch structure: navigation
error enters the closed-loop system only through the observed state, and the
second-order surfaces are evaluated from the surrogate branch, which operates
on smooth internally generated signals rather than on noisy measurements.

\bibliographystyle{jasr-model5-names}
\biboptions{authoryear}
\bibliography{references}
\end{document}